\documentclass[12pt,a4paper]{article}
\usepackage{xcolor,paralist}
\usepackage[pdftex,colorlinks=true,hypertexnames=false]{hyperref}
\definecolor{darkblue}{rgb}{0,0,.6}
\hypersetup{citecolor=darkblue,linkcolor=darkblue,urlcolor=darkblue}
\usepackage{amssymb,enumitem,booktabs,subfig,setspace,bm,longtable,xr,dsfont,orcidlink,fancyvrb}
\usepackage{algorithm}
\usepackage{algpseudocode,multirow}
\usepackage{tikz}
\usetikzlibrary{arrows.meta,positioning,shapes.geometric}

\usepackage[final,nopatch=footnote]{microtype}
\usepackage[top=0.85in, bottom=0.85in, left=0.85in, right=0.85in]{geometry}
\usepackage{natbib}
\usepackage{comment}
\usepackage{url}
\usepackage{amsmath}
\usepackage{amsfonts}
\usepackage{epsfig}
\usepackage{graphics}
\usepackage[normalem]{ulem}
\usepackage{palatino,eulervm}
\usepackage{graphicx}
\usepackage{lscape}
\usepackage{rotating}
\usepackage[linewidth=1pt]{mdframed}
\allowdisplaybreaks[4]

\providecommand{\U}[1]{\protect\rule{.1in}{.1in}}

\graphicspath{{plots/}}
\newsavebox\CBox

\usepackage{amsthm,thmtools}
\usepackage{mathrsfs}
\declaretheorem{theorem}
\declaretheorem{lemma}
\makeatletter
\def\th@newremark{\th@remark\thm@headfont{\bfseries}}
\makeatletter
\theoremstyle{newremark}
\newtheorem{remark}{Remark}

\newtheorem{corollary}{Corollary}
\newtheorem{assumption}{Assumption}
\declaretheoremstyle[
  spaceabove=6pt, spacebelow=6pt,
  headfont=\bfseries,
  notefont=\mdseries, notebraces={(}{)},
bodyfont=\normalfont,
  postheadspace=0.5em
]{mystyle}

\newcommand{\X}{\mathcal{X}}

\bibpunct{(}{)}{;}{a}{,}{ and }

\makeatletter
\newcommand*{\addFileDependency}[1]{
\typeout{(#1)}
\@addtofilelist{#1}
\IfFileExists{#1}{}{\typeout{No file #1.}}
}\makeatother

\newcommand{\Rlogo}{\protect\includegraphics[height=1.8ex,keepaspectratio]{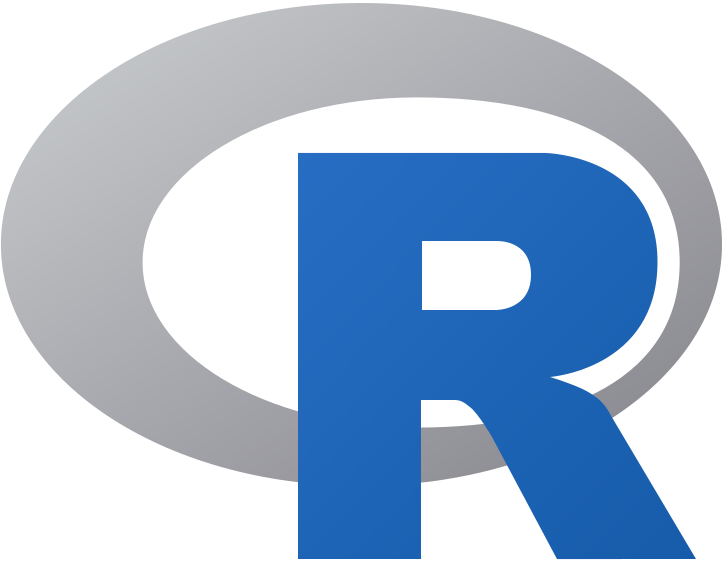}}

\newcommand{\commUB}[1]{{\leavevmode\color{purple}#1}}

\begin{document}

\title{Robust spatial scalar-on-function regression: A Fisher-consistent redescending M-estimation approach}

\author{\normalsize  Muge Mutis\footnote{Corresponding address: Department of Statistics, Yildiz Technical University, 34220, Esenler-Istanbul, Turkiye; Email address: muge.mutis@yildiz.edu.tr} \orcidlink{0000-0002-9801-4835} \\
\normalsize  Department of Statistics, Yildiz Technical University, Istanbul, Turkiye \\
\normalsize Ufuk Beyaztas \orcidlink{0000-0002-5208-4950} \\
\normalsize Department of Statistics, Marmara University, Istanbul, Turkiye\\
\normalsize Han Lin Shang \orcidlink{0000-0003-1769-6430} \\
\normalsize Department of Actuarial Studies and Business Analytics, Macquarie University, Sydney, Australia
}
\date{}
\maketitle

\begin{abstract}
We develop a Fisher-consistent redescending robust estimator for the spatial scalar-on-function regression model, where a scalar response depends on both a functional predictor and a spatial autoregressive lag. Existing estimation procedures for this model are typically based on likelihood methods or monotone-loss robust M-estimators. They may be highly sensitive to vertical outliers, leverage points in the functional predictor, and numerical instability induced by strong spatial dependence. To address these issues, we propose a new estimation framework that first applies robust functional principal component analysis to obtain a contamination-resistant finite-dimensional representation of the functional predictor and then estimates the resulting spatial regression model through a bias-corrected system of M-estimating equations. The proposed method allows redescending loss functions, including Andrews' sine and Danish losses, and jointly estimates the regression coefficients, spatial dependence parameter, and scale parameter within a unified Fisher-consistent framework. For computation, we develop a hybrid IRLS-Newton algorithm that combines weighted least-squares updates for the regression parameters with a Newton-Raphson update for the spatial parameter. We establish Fisher consistency, consistency, asymptotic normality, and the asymptotic distribution of the reconstructed slope function. Monte Carlo experiments show that the proposed estimators remain competitive under clean data and substantially outperform classical and Huber-type robust competitors under contamination, particularly in severe outlier settings. An application to French air-quality data further demonstrates improved predictive performance and stable estimation of spatial dependence. Our method has been implemented in the \texttt{fcsar} \Rlogo \ package.
\vspace{.1in}

\end{abstract}

\onehalfspacing

\section{Introduction}\label{sec:1}

Recent developments in spatial functional data analysis have been largely motivated by the need to incorporate complex functional covariates into regional comparative studies. The spatial scalar-on-function regression model (SSoFRM) extends the classical scalar-on-function regression model \citep{Hastie1993} by incorporating spatial dependence in the response. For $n$ spatial units (areal regions), the SSoFRM can be written as
\begin{equation}\label{eq:SSoFRM}
Y_i = \beta_0 + \vartheta \sum_{j=1}^n w_{ij} Y_j + \int_{\mathcal{I}} \mathcal{X}_i(t) \beta(t) dt + \varepsilon_i,
\qquad i \in \{1,\ldots,n\},
\end{equation}
where $Y_i$ is the scalar response at region $i$, $\mathcal{X}_i(t)$ is a functional predictor observed over a domain~$\mathcal{I}$, $\beta_0$ is an intercept, $\beta(t)$ is an unknown regression coefficient function, and $\mathcal{I}$ denotes a function support. The term $\vartheta \sum_{j=1}^n w_{ij}Y_j$ is a spatial autoregressive (SAR) lag with spatial weight matrix $\mathbf{W}=[w_{ij}]_{1\le i,j\le n}$, which is typically assumed to have zero diagonal entries and to be row-normalized \citep{Cliff1973, Anselin1998}; $\vartheta$ denotes the spatial dependence parameter. Model \eqref{eq:SSoFRM} allows the response at each location to depend not only on its own functional covariate $\mathcal{X}_i(t)$ but also on neighboring responses $Y_j$. This structure is appropriate for areal applications in which the responses are spatially dependent and the explanatory information is carried by spatially indexed functional covariates. In the asymptotic development presented later, however, the spatial dependence is modeled through the response side of the SSoFRM, whereas the functional predictors are treated under a simplifying independent and identically distributed (i.i.d.) assumption to make the robust projection step tractable.

Historically, most estimation procedures for SSoFRM \citep[e.g.,][]{Medina2011, Medina2012, Pineda2019, Hu2021, Huang2021, Bouka2023} have relied on likelihood-based methods. While these approaches can perform well under idealized assumptions, they are well known to be sensitive to departures from model regularity, particularly the presence of outliers. In practice, even a small fraction of anomalous observations, including vertical outliers in the response or leverage points arising from atypical functional predictors, may severely bias classical estimators, distort coefficient estimates, and compromise inference \citep{DeLuna2002}. The vulnerability of likelihood-based estimation in spatial models under contamination has long motivated the development of robust alternatives that limit the impact of aberrant observations \citep{Genton2003}.

Among the existing robust contributions, \cite{Huang2021} should be viewed as an important early reference for the present work. They proposed a robust SSoFRM by assuming a $t$-distribution for the error term, projecting the functional predictor onto an orthonormal basis, and estimating the model through an expectation-maximization algorithm. Their approach, therefore, introduces robustness by adopting a heavy-tailed likelihood specification within the SSoFRM framework. However, it remains likelihood-based and is not formulated through redescending estimating equations, direct bias-corrected spatial score equations, or a unified Fisher-consistent estimating system for the regression, spatial, and scale parameters.

A more direct precursor is \cite{Beyaztas2026}, who considered the same model class and combined robust functional dimension reduction with robust second-stage spatial estimation. In particular, they projected the functional predictor onto a low-dimensional space using robust functional principal component analysis (RFPCA), and then estimated the resulting finite-dimensional spatial autoregressive model using the Huber-type robust M-estimator of \cite{Tho2023}. That strategy already delivers an important robustification of SSoFRM by protecting the first stage against functional outliers and the second stage against vertical contamination.

Against this background, the contribution of the present paper is more specific but also more focused. We do not introduce a new robust dimension-reduction step for SSoFRM. Rather, conditional on an RFPCA-based finite-dimensional representation, we develop a new second-stage spatial estimation framework. Specifically, we replace the Huber-type second-stage estimator with a jointly bias-corrected redescending system of M-estimating equations, in which the bias correction for the SAR parameter is incorporated directly into the estimating equation itself, and the scale parameter is estimated within the same loss-driven Fisher-consistent system.

To compute the resulting estimator, we develop a hybrid numerical scheme that combines iteratively reweighted least squares (IRLS) updates for the regression coefficients with a Newton-Raphson update for the spatial parameter. In contrast to fixed-point or IRLS-only updating, the Newton step directly targets the bias-corrected spatial estimating equation, yielding a fully explicit and reproducible update for the spatial parameter. As shown later through additional computational diagnostics, this design should be viewed primarily as a stability- and transparency-oriented implementation choice rather than as a uniform speed enhancement. The scale parameter is updated via a Fisher-consistent M-scale equation aligned with the chosen loss, thereby avoiding loss-agnostic external scale estimators such as the median absolute deviation that are commonly used in related SAR M-estimation procedures \citep{Zhan2025}. This unified treatment of regression, spatial dependence, and scale parameters provides a coherent basis for asymptotic analysis and inference.

Working within a finite-dimensional projected setting, we treat the RFPCA truncation level as fixed and impose an exact-span condition so that truncation does not introduce approximation bias. Under additional assumptions including i.i.d. functional predictors, i.i.d.\ symmetric errors independent of the regressors, identification of the population estimating equations, uniform convergence of the sample estimating map, and a suitable limiting empirical-process condition, we establish Fisher consistency, consistency, asymptotic normality, and the asymptotic distribution of the reconstructed slope function. Accordingly, relative to \cite{Huang2021} and especially to the more direct precursor \cite{Beyaztas2026}, the novelty of this paper lies in providing a large-sample treatment for the proposed second-stage redescending estimator, rather than a fully general asymptotic theory for SSoFRM with growing dimension or data-driven truncation.

Theoretical developments also distinguish the present paper from the closest literature. Treating the RFPCA truncation level as fixed, we establish Fisher consistency, consistency, asymptotic normality, and the asymptotic distribution of the reconstructed slope function under spatial dependence. Thus, relative to \cite{Huang2021} and especially to the more direct precursor \cite{Beyaztas2026}, the novelty of this paper lies primarily in the second-stage spatial estimation and inference machinery: a jointly bias-corrected redescending estimating system, a unified treatment of scale, and a hybrid IRLS-Newton computational scheme.

By design, the proposed approach (``FCSAR'') offers several methodological advantages relative to the closest existing robust SSoFRM approaches. Compared with the likelihood-based $t$-error formulation of \cite{Huang2021}, it does not encode robustness primarily through a parametric heavy-tailed error distribution, but through redescending estimating equations that can assign negligible weight to extreme residual outliers. Compared with the Huber-type second-stage estimator used in \cite{Beyaztas2026}, it enforces Fisher consistency through a single bias-corrected estimating system for the full parameter vector, rather than through a monotone-loss spatial M-estimation step combined with external scale handling. In addition, the hybrid IRLS-Newton algorithm mitigates sensitivities of purely reweighting-based procedures to initialization and local convergence behavior that have been noted for SAR M-estimation \citep{Tho2023, Zhan2025}, and its computational behavior can be examined directly through reported iteration counts, runtime summaries, and convergence rates. Finally, embedding scale estimation within the same loss-driven framework eliminates the need for ad hoc calibration constants in external scale estimators. It yields a scale estimator tailored to the chosen redescending loss.

The rest of this paper is organized as follows. In Section~\ref{sec:2}, we propose the SSoFRM framework considered in this study and, building on a robust spectral decomposition of the functional predictors, derive a bias-corrected unbiased M-estimator that is computed using a hybrid IRLS-Newton optimization scheme. Section~\ref{sec:asymptotic} establishes the large-sample properties of the proposed estimator, including Fisher consistency, consistency, asymptotic normality, and the asymptotic behavior of the reconstructed slope function in the sense needed for linear functionals and $\mathcal{L}^2$-type summaries. Section~\ref{sec:simulation} reports Monte Carlo studies that examine the finite-sample estimation and prediction performance of the proposed method under both clean and contaminated settings. Section~\ref{sec:real_data} presents an empirical application to French air-quality data. Section~\ref{sec:conclusion} concludes the paper, along with some ideas on how the methodology presented can be extended. Technical assumptions, auxiliary lemmas, and all proofs are collected in the Appendix.

\section{Methodology}\label{sec:2}

\subsection{Spatial scalar-on-function regression model}\label{subsec:model}

Let $(\Omega,\mathcal{F},\mathbb{P})$ be a complete probability space supporting all random elements. Consider a spatial domain $\mathcal{D}\subset\mathbb{R}^d$ equipped with the Borel $\sigma$-algebra, and $n$ spatial units indexed by $i\in\{1,\ldots,n\}$. For each unit $i$, we observe a scalar response $Y_i:\Omega\to\mathbb{R}$ and a functional predictor $\mathcal{X}_i:\Omega\to\mathcal{H}$, where $\mathcal{H}=\mathcal{L}^2(\mathcal{I})$ is the Hilbert space of square-integrable real-valued functions on $\mathcal{I}=[0,1]$, endowed with inner product $\langle f,g\rangle_{\mathcal{H}}=\int_{\mathcal{I}} f(t)g(t) dt$ and the norm $\|f\|_{\mathcal{H}}=\langle f,f\rangle_{\mathcal{H}}^{1/2}$.
Throughout, we assume $\mathbb{E}\|\mathcal{X}_i\|_{\mathcal{H}}^2<\infty$ and, when needed, $\sup_{1\le i\le n}\|\mathcal{X}_i\|_{\mathcal{H}}<\infty$ (a.s.), so that all inner products below are well-defined.

Let $\mathbf{W}=[w_{ij}]_{1\le i,j\le n}$ be a given row-normalized spatial weight matrix with $w_{ii}=0$ and $\sum_{j=1}^n w_{ij}=1$ for each $i$. The SSoFRM with autoregressive dependence is defined by
\begin{equation}\label{eq:ssofrm}
Y_i = \beta_0 + \vartheta \sum_{j=1}^n w_{ij} Y_j + \langle \mathcal{X}_i,\beta\rangle_{\mathcal{H}} + \varepsilon_i,
\qquad i \in \{1,\ldots,n\},
\end{equation}
where $\beta_0\in\mathbb{R}$ is an intercept, $\vartheta\in\mathbb{R}$ is a spatial autoregressive parameter, and $\beta\in\mathcal{H}$ is an unknown slope function. We assume that the structural errors admit the scale representation $\varepsilon_i=\sigma u_i$, where $\sigma>0$ and $u_i$ are i.i.d.\ with a symmetric distribution $F$ satisfying $\mathbb{E}(u_i)=0$; moreover, $\mathbb{E}(\varepsilon_i\mid \mathcal{X}_1,\ldots,\mathcal{X}_n)=0$.

Define $\mathbf{Y}=(Y_1,\ldots,Y_n)^\top$, $\mathbf{1}_n=(1,\ldots,1)^\top$, and $\boldsymbol{\varepsilon}=(\varepsilon_1,\ldots,\varepsilon_n)^\top$. Let $\mathcal{T}:\mathcal{H}\to\mathbb{R}^n$ be the linear operator defined by $(\mathcal{T}f)_i=\langle \mathcal{X}_i,f\rangle_{\mathcal{H}}$ for $f\in\mathcal{H}$. Then \eqref{eq:ssofrm} can be written in matrix form as
\begin{equation}\label{eq:matrix_model}
\mathbf{Y} = \beta_0 \mathbf{1}_n + \vartheta \mathbf{W}\mathbf{Y} + \mathcal{T}\beta + \boldsymbol{\varepsilon}.
\end{equation}
Let $\mathbf{S}(\vartheta):=(\mathbf{I}_n-\vartheta\mathbf{W})^{-1}$ whenever the inverse exists. A sufficient condition for nonsingularity is $|\vartheta|<1/\rho(\mathbf{W})$, where $\rho(\mathbf{W})$ denotes the spectral radius of $\mathbf{W}$. In particular, for many row-normalized weight matrices, one has $\rho(\mathbf{W})=1$, in which case the condition reduces to $|\vartheta|<1$, which we assume throughout. Under this condition,
\begin{equation}\label{eq:reduced_form}
\mathbf{Y} = \mathbf{S}(\vartheta)\bigl(\beta_0 \mathbf{1}_n + \mathcal{T}\beta + \boldsymbol{\varepsilon}\bigr),
\end{equation}
which follows by rearranging \eqref{eq:matrix_model} and left-multiplying by $\mathbf{S}(\vartheta)$.

\subsection{Robust spectral decomposition of the functional predictor}\label{subsec:rfpca}

The infinite-dimensional nature of $\beta\in\mathcal{H}$ requires dimension reduction. Let $\mathcal{X}$ denote a generic random element in $\mathcal{H}$, and suppose $\mathcal{X}_1 ,\ldots, \mathcal{X}_n$ are identically distributed copies of $\mathcal{X}$. Classical FPCA is based on the Karhunen-Lo\`eve expansion
\begin{equation*}
\mathcal{X} = \mu + \sum_{k=1}^{K} \xi_k \phi_k, \qquad \xi_k=\langle \mathcal{X}-\mu,\phi_k\rangle_{\mathcal{H}}, \qquad
\mathrm{Var}(\xi_k)=\lambda_k,
\end{equation*}
where $K$ is the truncation level and $\mu=\mathbb{E}(\mathcal{X})$ and $(\lambda_k,\phi_k)$ are the eigenpairs of the covariance operator $\Gamma=\mathbb{E}\{(\mathcal{X}-\mu)\otimes(\mathcal{X}-\mu)\}$. Because $\Gamma$ is highly sensitive to atypical curves, we adopt the RFPCA based on projection pursuit with an M-scale dispersion functional \citep[see, e.g.,][]{Bali2011}. We emphasize that the loss function used in this step (denoted $\rho_1$ below) is distinct from the loss function $\rho_2$ (and its derivative $\psi_2$) used later for robust spatial M-estimation.

Let $\rho_1:\mathbb{R}\to\mathbb{R}_{+}$ be an even, bounded, continuously differentiable function that is nondecreasing on $\mathbb{R}_+$. For a real random variable $Z$, define its M-scale $\sigma_M(Z)$ as the unique positive solution of
\begin{equation*}
\mathbb{E}\left\lbrace\rho_1 \left(\frac{Z-\ell(Z)}{\sigma_M(Z)}\right)\right\rbrace=\delta,
\end{equation*}
where $\ell(Z)$ is a robust location functional (e.g., the median) and $\delta\in(0,\sup\rho_1)$ is chosen to ensure Fisher consistency under a reference model (typically Gaussian). A common choice is Tukey's bisquare loss
\begin{equation*}
\rho_1(u)=
\begin{cases}
\dfrac{c^2}{6}\left\{1-\left(1-\dfrac{u^2}{c^2}\right)^3\right\}, & |u|\le c,\\
\dfrac{c^2}{6}, & |u|>c,
\end{cases}
\end{equation*}
with tuning constants $(c,\delta)$ selected to attain the desired robustness-efficiency trade-off (e.g., $50\%$ breakdown under Gaussian calibration) \citep{Bali2011}.

For each direction $a\in\mathcal{H}$ with $\|a\|_{\mathcal{H}}=1$, define the projection score $Z_a=\langle \mathcal{X}, a\rangle_{\mathcal{H}}$ and the robust dispersion functional $\mathcal{V}(a):=\sigma_M^2(Z_a)$. The first robust principal direction is defined by
\begin{equation}\label{eq:rpp1}
\phi_1=\arg\max_{\|a\|_{\mathcal{H}}=1}\mathcal{V}(a),
\end{equation}
and for $k\ge 2$,
\begin{equation}\label{eq:rppk}
\phi_k=\arg\max_{\substack{\|a\|_{\mathcal{H}}=1\\ a\perp\{\phi_1,\ldots,\phi_{k-1}\}}}\mathcal{V}(a).
\end{equation}
The associated robust eigenvalues are $\lambda_k=\mathcal{V}(\phi_k)$. Let $\mu_R$ denote the corresponding robust functional location (population version) used for centering; then the robust scores are $\xi_{ik}=\langle \mathcal{X}_i-\mu_R,\phi_k\rangle_{\mathcal{H}}$, for $i \in \{1, \ldots, n\}$.

Given the observed curves $\mathcal{X}_1,\ldots,\mathcal{X}_n$, let $\widehat{\mu}$ be a robust functional location estimator and define the empirical projections $Z_{i,a}=\langle \mathcal{X}_i-\widehat{\mu},a\rangle_{\mathcal{H}}$. The empirical M-scale $\widehat{\sigma}_M(a)$ is defined as the unique positive solution of
\begin{equation}\label{eq:empMscale}
\frac{1}{n}\sum_{i=1}^n \rho_1 \left(\frac{Z_{i,a}-\widehat{\ell}(a)}{\widehat{\sigma}_M(a)}\right)=\delta,
\end{equation}
where $\widehat{\ell}(a)$ is a robust univariate location estimate computed from $Z_{1,a}, \dots, Z_{n, a}$. The estimated robust directions $\widehat{\phi}_k$ are then obtained by replacing $\mathcal{V}(a)$ with $\widehat{\sigma}_M(a)$ in \eqref{eq:rpp1}-\eqref{eq:rppk} and solving the resulting sequential optimization problems.

In practice, we compute directions up to a fixed $K$ and select the truncation level $K$ by a robust explained-dispersion rule: choose the smallest $K$ such that 
\[
\frac{\sum_{k = 1}^K \widehat{\lambda}_k}{\sum_{k\geq1} \widehat{\lambda}_k}\ge 1-\varrho, 
\]
for a small $\varrho>0$ (e.g., $\varrho=0.05$). The functional predictor is then approximated by
\begin{equation}\label{eq:rfpca_approx}
\mathcal{X}_i \approx \widehat{\mu}+\sum_{k=1}^K \widehat{\xi}_{ik}\widehat{\phi}_k,
\qquad
\widehat{\xi}_{ik}=\langle \mathcal{X}_i-\widehat{\mu},\widehat{\phi}_k\rangle_{\mathcal{H}}, \qquad i = 1, \ldots, n, \quad k = 1, \ldots, K.
\end{equation}

\subsection{Finite-dimensional projection and unbiased M-estimation}\label{subsec:Mest}

If we expand the slope function in the robust eigenbasis, $\beta(t)\approx \sum_{k=1}^K b_k \widehat{\phi}_k(t)$, where $\mathbf{b}=(b_1,\ldots,b_K)^\top\in\mathbb{R}^K$. Then, using \eqref{eq:rfpca_approx} and orthonormality of ${\widehat{\phi}_k}$, the functional linear term in~\eqref{eq:ssofrm} reduces to $\langle \mathcal{X}_i,\beta\rangle_{\mathcal{H}} \approx \sum_{k=1}^K \widehat{\xi}_{ik} b_k$. Let $\boldsymbol{\Xi}=[\widehat{\xi}_{ik}]\in\mathbb{R}^{n\times K}$ and define the (estimated) design matrix $\mathbf{Z}=[\mathbf{1}_n,\boldsymbol{\Xi}]\in\mathbb{R}^{n\times (K+1)}$ and $\boldsymbol{\gamma}=(\beta_0,\mathbf{b}^\top)^\top\in\mathbb{R}^{K+1}$. Then, the projected model is a standard SAR regression:
\begin{equation}\label{eq:proj_SAR}
\mathbf{Y}=\vartheta \mathbf{W}\mathbf{Y}+\mathbf{Z}\boldsymbol{\gamma}+\boldsymbol{\varepsilon}, \qquad \text{equivalently} \qquad (\mathbf{I}_n-\vartheta\mathbf{W})\mathbf{Y} = \mathbf{Z}\boldsymbol{\gamma}+\boldsymbol{\varepsilon}.
\end{equation}

Define the residuals
\begin{equation*}
e_i(\boldsymbol{\gamma},\vartheta) = \{(\mathbf{I}_n-\vartheta\mathbf{W})\mathbf{Y} - \mathbf{Z}\boldsymbol{\gamma}\}_i,\qquad i \in \{1,\ldots,n\},
\end{equation*}
and let $\mathbf{e}(\boldsymbol{\gamma},\vartheta) = (e_1,\ldots,e_n)^\top$. Let $\rho_2:\mathbb{R}\to\mathbb{R}+$ be an even, continuously differentiable loss function with derivative $\psi_2 = \rho_2'$. We consider Danish (soft‑redescending) and Andrews (hard‑redescending) losses as follows:
\begin{equation*}
\begin{aligned}
\text{Danish}:\quad
\rho_2^{(\text{D})}(u)&=
\begin{cases}
\frac{1}{2}u^{2}, & |u|\le c,\\[2mm]
v-\frac{1}{2}c\exp \bigl(\frac{-u^{2}}{c}\bigr), & |u|>c,
\end{cases}
&\quad
\psi_2^{(\text{D})}(u) = \rho_2^{(\text{D})\prime}(u)&=
\begin{cases}
u, & |u|\le c,\\[2mm]
u\exp \bigl(\frac{-u^{2}}{c}\bigr), & |u|>c;
\end{cases}
\\[4mm]
\text{Andrews}:\quad
\rho_2^{(\text{A})}(u)&=
\begin{cases}
c^{2}\bigl[1-\cos \bigl(\frac{u}{c}\bigr)\bigr], & |u|\le \pi c,\\[2mm]
2c^{2}, & |u|>\pi c,
\end{cases}
&\quad
\psi_2^{(\text{A})}(u) = \rho_2^{(\text{A})\prime}(u)&=
\begin{cases}
c\sin \bigl(\frac{u}{c}\bigr), & |u|\le \pi c,\\[2mm]
0, & |u|>\pi c.
\end{cases}
\end{aligned}
\end{equation*}
with tuning constants chosen to achieve a balance between robustness and efficiency (e.g., $c=1.339$ for Andrews, $c=2.767$ and $v=3.915$ for Danish). Let $\sigma>0$ denote a scale parameter, and set $u_i = e_i/\sigma$.

For the SAR model, the naive M-estimating equation for the spatial parameter is generally biased. Indeed, at the true parameter value $\boldsymbol{\theta}_0$, one has $e_i(\boldsymbol{\theta}_0)=\varepsilon_i$, but because the spatial lag $(\mathbf{W}\mathbf{Y})_i$ is endogenous, the corresponding score is not centered: $\mathbb{E}\{\psi_2(\varepsilon_i/\sigma_0)(\mathbf{W}\mathbf{Y})_i\} \neq 0$. Therefore, a bias correction is required in the estimating equation for $\vartheta$ to obtain Fisher consistency. Following \citet{Zhan2025}, we define the bias term for observation $i$ as
\begin{equation}\label{eq:bias_term}
\mathrm{bias}_i(\boldsymbol{\theta}) = \bigl\{\mathbf{W} (\mathbf{I}_n - \vartheta \mathbf{W})^{-1}\bigr\}_{ii} \sigma \kappa,
\end{equation}
where $\boldsymbol{\theta}=(\boldsymbol{\gamma}^\top,\sigma,\vartheta)^\top$ and $\kappa = \mathbb{E}\{\psi_2(\varepsilon_i/\sigma_0)(\varepsilon_i/\sigma_0)\}$ is the consistency constant associated with the scale equation. In practical terms, $\kappa$ is a calibration constant that makes the scale equation Fisher-consistent under the reference error model. Writing $U=\varepsilon_i/\sigma_0$, one has $\kappa=\mathbb{E}\{\psi_2(U)U\}$ where the expectation is taken under the standardized reference distribution for $U$. Thus, $\kappa$ depends only on the chosen score function $\psi_2$ and not on the unknown regression or spatial parameters. For a given loss, it can therefore be computed once in advance and then treated as fixed throughout estimation. In our implementation, no iterative numerical approximation of $\kappa$ is required inside the algorithm: we use precomputed values corresponding to the selected redescending loss, namely $\kappa=0.829$ for Andrews' sine loss and $\kappa=0.872$ for Danish loss. The unbiased estimating equations are then expressed as
\begin{align}
\sum_{i=1}^n \psi_2 \left(\frac{e_i}{\sigma}\right)\mathbf{z}_i &= \mathbf{0}_{K+1}, \label{eq:unbiased_gamma} \\
\sum_{i=1}^n \left\lbrace \psi_2 \left(\frac{e_i}{\sigma}\right)\frac{e_i}{\sigma} - \kappa \right\rbrace &= 0, \label{eq:unbiased_sigma} \\
\sum_{i=1}^n \left\lbrace \psi_2 \left(\frac{e_i}{\sigma}\right)(\mathbf{W}\mathbf{Y})_i - \mathrm{bias}_i(\boldsymbol{\theta}) \right\rbrace &= 0, \label{eq:unbiased_theta}
\end{align}
where $\mathbf{z}_i^\top$ is the $i$\textsuperscript{th} row of $\mathbf{Z}$. At the true parameter $\boldsymbol{\theta}_0$, each of these sums has expectation zero, ensuring Fisher consistency of the resulting M‑estimator $\widehat{\boldsymbol{\theta}}$.

\subsection{Hybrid IRLS–Newton algorithm}\label{subsec:algorithm}

Solving the unbiased system \eqref{eq:unbiased_gamma}–\eqref{eq:unbiased_theta} requires an iterative scheme that handles the nonlinear dependence on $\vartheta$ through the bias term. We propose a hybrid algorithm that combines IRLS for the regression coefficients $\boldsymbol{\gamma}$ with a Newton–Raphson update for the spatial parameter $\vartheta$, while updating the scale $\sigma$ using a Fisher‑consistent fixed‑point equation.

To start the iterations, we need a $\sqrt{n}$‑consistent initial estimator. We use the bias‑corrected naive least squares (NLS) estimator proposed by \citet{Zhan2025}. First, form the augmented matrix $\mathbf{A}=[\mathbf{W}\mathbf{Y},\mathbf{Z}]$ and compute the raw NLS estimates $(\widehat{\vartheta}_{\mathrm{NLS}},\widehat{\boldsymbol{\gamma}}_{\mathrm{NLS}})$ by ordinary least squares solving $\mathbf{A}^\top\mathbf{A} \boldsymbol{\theta}_{\text{reg}} = \mathbf{A}^\top\mathbf{Y}$. Then, using the plug‑in formulas derived in \citet{Zhan2025} (see also Appendix), compute the estimated bias $\widehat{\operatorname{bias}}(\widehat{\vartheta}_{\mathrm{NLS}})$ and set $\vartheta^{(0)} = \widehat{\vartheta}_{\mathrm{NLS}} - \widehat{\operatorname{bias}}(\widehat{\vartheta}_{\mathrm{NLS}})$, where $\boldsymbol{\gamma}^{(0)} = \widehat{\boldsymbol{\gamma}}_{\mathrm{NLS}}$. The initial residuals are $\mathbf{e}^{(0)} = \mathbf{Y} - \vartheta^{(0)}\mathbf{W}\mathbf{Y} - \mathbf{Z}\boldsymbol{\gamma}^{(0)}$, and the scale is initialized as the scaled median absolute deviation, $\sigma^{(0)} = \operatorname{median}(|\mathbf{e}^{(0)}|)/ 0.6745$.

Given current iterates $(\boldsymbol{\gamma}^{(h)},\vartheta^{(h)},\sigma^{(h)})$, compute the residuals $e_i^{(h)} = \{(\mathbf{I}_n-\vartheta^{(h)}\mathbf{W})\mathbf{Y} - \mathbf{Z}\boldsymbol{\gamma}^{(h)}\}_i$, standardized residuals $u_i^{(h)} = e_i^{(h)}/\sigma^{(h)}$, and weights $\varphi_i^{(h)} = p(u_i^{(h)})$ where $p(u)=\psi_2(u)/u$ (with $p(0)=\psi_2'(0)$). Let $\boldsymbol{\Pi}_{\text{diag}}^{(h)} = \operatorname{diag}(\varphi_i^{(h)})$.

For fixed $\vartheta^{(h)}$ and $\sigma^{(h)}$, \eqref{eq:unbiased_gamma} is solved by weighted least squares. Specifically, with $u_i^{(h)}=e_i^{(h)}/\sigma^{(h)}$ and $\varphi_i^{(h)}=p(u_i^{(h)})$, define $\boldsymbol{\Pi}^{(h)}_{\mathrm{diag}}=\operatorname{diag}(\varphi_i^{(h)})$. We then compute
\begin{equation}\label{eq:gamma_update}
\boldsymbol{\gamma}^{(h+1)} = \left(\mathbf{Z}^{\top}\boldsymbol{\Pi}^{(h)}_{\mathrm{diag}}\mathbf{Z}\right)^{-1}
\mathbf{Z}^{\top}\boldsymbol{\Pi}^{(h)}_{\mathrm{diag}}
\bigl(\mathbf{Y}-\vartheta^{(h)}\mathbf{W}\mathbf{Y}\bigr),
\end{equation}
whenever the weighted normal matrix is numerically invertible. In implementation, if $\mathbf{Z}^{\top}\boldsymbol{\Pi}^{(h)}_{\mathrm{diag}}\mathbf{Z}$ is singular or the linear solve fails numerically, we set $\boldsymbol{\gamma}^{(h+1)}=\boldsymbol{\gamma}^{(h)}$ for that outer iteration. This gives a fully specified fallback rule and avoids introducing an additional regularization parameter into the main algorithm.

Equation~\eqref{eq:unbiased_theta} is nonlinear in $\vartheta$. Using the newly computed $\boldsymbol{\gamma}^{(h+1)}$ and the current scale $\sigma^{(h)}$, define
\begin{equation*}
S_\vartheta(\vartheta) = \sum_{i=1}^{n} \psi_2 \left(\frac{e_i}{\sigma^{(h)}}\right)(\mathbf{W}\mathbf{Y})_i - \sum_{i=1}^{n}\mathrm{bias}_i\!\left(\boldsymbol{\gamma}^{(h+1)},\vartheta,\sigma^{(h)}\right),
\end{equation*}
where $e_i=\{(\mathbf{I}_n-\vartheta\mathbf{W})\mathbf{Y}-\mathbf{Z}\boldsymbol{\gamma}^{(h+1)}\}_i$. Its derivative is
\begin{equation*}
\frac{\partial S_\vartheta}{\partial \vartheta} = \sum_{i=1}^{n}
\psi_2' \left(\frac{e_i}{\sigma^{(h)}}\right) \frac{-(\mathbf{W}\mathbf{Y})_i}{\sigma^{(h)}}(\mathbf{W}\mathbf{Y})_i
- \sum_{i=1}^{n} \frac{\partial}{\partial\vartheta} \mathrm{bias}_i \left(\boldsymbol{\gamma}^{(h+1)},\vartheta,\sigma^{(h)}\right).
\end{equation*}
We update $\vartheta$ by a damped Newton step inside a fixed compact admissible interval $\Theta_\vartheta=[-0.99, ~0.99]$ which is compatible with the standing assumption $|\vartheta|<1$ for the row-standardized spatial weight matrices used in this paper. Let $J_{\max}$ denote the maximum number of inner Newton iterations; in implementation, we use $J_{\max}=10$. Starting from $\vartheta_{\mathrm{old}}=\vartheta^{(h)}$, define $\vartheta_{\mathrm{cand}} = \vartheta_{\mathrm{old}} - \tau \left(
\frac{\partial S_\vartheta}{\partial\vartheta} \right)^{-1} S_\vartheta(\vartheta_{\mathrm{old}})$ and then project the candidate back into the admissible region by $\vartheta_{\mathrm{new}}=\Pi_{\Theta_\vartheta}(\vartheta_{\mathrm{cand}}) = \min\{0.99,\max\{-0.99,\vartheta_{\mathrm{cand}}\}\}$. Thus, every Newton iterate remains in the region where $(\mathbf{I}_n-\vartheta\mathbf{W})$ is invertible. We terminate the inner loop when $|\vartheta_{\mathrm{new}}-\vartheta_{\mathrm{old}}|<10^{-8}$ or when $J_{\max}$ iterations have been reached.

With $\boldsymbol{\gamma}^{(h+1)}$ and $\vartheta^{(h+1)}$ available, compute the updated residuals $\mathbf{e}^{(h+1)} = \mathbf{Y} - \vartheta^{(h+1)} \mathbf{W} \mathbf{Y} - \mathbf{Z}\boldsymbol{\gamma}^{(h+1)}$. For reproducibility, we use recomputed weights in the scale step rather than leaving this optional. That is, we set $u_i^{(h+1)}= e_i^{(h+1)} / \sigma^{(h)}$ and $\varphi_i^{(h+1)}=p \left(u_i^{(h+1)}\right)$, and update the scale by
\begin{equation}\label{eq:sigma_update}
\sigma^{(h+1)} = \left\{ \frac{1}{n\kappa} \sum_{i=1}^{n} \varphi_i^{(h+1)} \bigl(e_i^{(h+1)}\bigr)^2 \right\}^{1/2}.
\end{equation}
This ensures that the scale estimator remains aligned with the chosen loss function. Here, the constant $\kappa$ is determined entirely by the chosen redescending score function and the reference standardized error law. Accordingly, for each loss function, we evaluate $\kappa$ once before optimization and keep it fixed during the iterations. This means that the scale update does not require any additional numerical integration inside the IRLS-Newton loop.

\begin{algorithm}[!htbp]
\caption{\small FCSAR: Fisher-consistent robust estimation for SSoFRM} \label{alg:FCSAR}
\begin{algorithmic}[1]
\Require $\mathbf{Y}$, $\{\mathcal{X}_i\}$, $\mathbf{W}$, loss type (Andrews/Danish), RFPCA tuning $(\rho_1,\delta)$, truncation level $K$, outer tolerance $\texttt{tol}$, maximum outer iterations $H$, damping factor $\tau$, maximum inner Newton iterations $J_{\max}$.
\Ensure $\widehat{\boldsymbol{\gamma}}$, $\widehat{\sigma}$, $\widehat{\vartheta}$, $\widehat{\beta}(\cdot)$.

\Statex \textbf{Step A: RFPCA.}
\State Compute robust functional location $\widehat{\mu}$ and robust directions $\widehat{\phi}_1,\ldots,\widehat{\phi}_K$ by projection-pursuit maximization \eqref{eq:rpp1}--\eqref{eq:rppk} using empirical M-scales \eqref{eq:empMscale}.
\State Compute scores $\widehat{\xi}_{ik}$ and form $\mathbf{Z}=[\mathbf{1}_n,\boldsymbol{\Xi}]$.

\Statex \textbf{Step B: Bias-corrected NLS initialization.}
\State Set $\mathbf{A}=[\mathbf{W}\mathbf{Y},\mathbf{Z}]$.
\State Compute raw NLS $(\widehat{\vartheta}_{\mathrm{NLS}},\widehat{\boldsymbol{\gamma}}_{\mathrm{NLS}})$ by solving $\mathbf{A}^{\top}\mathbf{A}\,\boldsymbol{\theta}_{\mathrm{reg}}=\mathbf{A}^{\top}\mathbf{Y}$.
\State Compute $\widehat{\operatorname{bias}}(\widehat{\vartheta}_{\mathrm{NLS}})$ using the formula in Appendix~\ref{sec:appendix_bias}.
\State Set $\vartheta^{(0)}=\Pi_{[-0.99, ~0.99]} \Bigl( \widehat{\vartheta}_{\mathrm{NLS}}-\widehat{\operatorname{bias}}(\widehat{\vartheta}_{\mathrm{NLS}}) \Bigr)$ and $\boldsymbol{\gamma}^{(0)}=\widehat{\boldsymbol{\gamma}}_{\mathrm{NLS}}$.
\State Compute $\mathbf{e}^{(0)}=\mathbf{Y}-\vartheta^{(0)}\mathbf{W}\mathbf{Y}-\mathbf{Z}\boldsymbol{\gamma}^{(0)}$ and set $\sigma^{(0)}=\operatorname{median}(|\mathbf{e}^{(0)}|)/0.6745$.

\Statex \textbf{Step C: Hybrid IRLS-Newton iterations.}
\For{$h=0,1,\ldots,H-1$}
    \State Compute $u_i^{(h)}=e_i^{(h)}/\sigma^{(h)}$, $\varphi_i^{(h)}=p(u_i^{(h)})$, and $\boldsymbol{\Pi}^{(h)}_{\mathrm{diag}}=\operatorname{diag}(\varphi_i^{(h)})$.    
    \State Form $\mathbf{M}^{(h)}=\mathbf{Z}^{\top}\boldsymbol{\Pi}^{(h)}_{\mathrm{diag}}\mathbf{Z}$ and $\mathbf{r}^{(h)}=\mathbf{Z}^{\top}\boldsymbol{\Pi}^{(h)}_{\mathrm{diag}}(\mathbf{Y}-\vartheta^{(h)}\mathbf{W}\mathbf{Y})$.
    \If{$\mathbf{M}^{(h)}$ is numerically invertible}
        \State Set $\boldsymbol{\gamma}^{(h+1)}=(\mathbf{M}^{(h)})^{-1}\mathbf{r}^{(h)}$.
    \Else
        \State Set $\boldsymbol{\gamma}^{(h+1)}=\boldsymbol{\gamma}^{(h)}$.
    \EndIf
    \Statex \textbf{Inner Newton loop for $\vartheta$.}
    \State Set $\vartheta_{\mathrm{old}}=\vartheta^{(h)}$.
    \For{$j=1,\ldots,J_{\max}$}
        \State Compute $e_i=\{(\mathbf{I}_n-\vartheta_{\mathrm{old}}\mathbf{W})\mathbf{Y}-\mathbf{Z}\boldsymbol{\gamma}^{(h+1)}\}_i$ and $u_i=e_i/\sigma^{(h)}$.
        \State Compute $S_\vartheta(\vartheta_{\mathrm{old}})$ and $\partial S_\vartheta/\partial\vartheta$.
        \State Set $\vartheta_{\mathrm{cand}} = \vartheta_{\mathrm{old}} - \tau \left( \frac{\partial S_\vartheta}{\partial\vartheta}
        \right)^{-1} S_\vartheta(\vartheta_{\mathrm{old}})$.
        \State Project onto the admissible region $\vartheta_{\mathrm{new}}=\Pi_{[-0.99, ~0.99]}(\vartheta_{\mathrm{cand}})$.
        \If{$|\vartheta_{\mathrm{new}}-\vartheta_{\mathrm{old}}|<10^{-8}$}
            \State \textbf{break}
        \EndIf
        \State Set $\vartheta_{\mathrm{old}}=\vartheta_{\mathrm{new}}$.
    \EndFor
    \State Set $\vartheta^{(h+1)}=\vartheta_{\mathrm{new}}$.
    \State Compute updated residuals $\mathbf{e}^{(h+1)}=\mathbf{Y} -\vartheta^{(h+1)}\mathbf{W}\mathbf{Y} -\mathbf{Z}\boldsymbol{\gamma}^{(h+1)}$.
    \State Recompute weights using $u_i^{(h+1)}=e_i^{(h+1)}/\sigma^{(h)}$ and $\varphi_i^{(h+1)}=p(u_i^{(h+1)})$.
    \State Update $\sigma^{(h+1)} = \left\{ \frac{1}{n\kappa} \sum_{i=1}^{n} \varphi_i^{(h+1)}(e_i^{(h+1)})^2 \right\}^{1/2}$.
    \If{$\max \bigl( \|\boldsymbol{\gamma}^{(h+1)} - \boldsymbol{\gamma}^{(h)}\|_{\infty}, |\vartheta^{(h+1)}-\vartheta^{(h)}|, |\sigma^{(h+1)}-\sigma^{(h)}| \bigr)<\texttt{tol}$}
        \State \textbf{break}
    \EndIf
\EndFor
\Statex \textbf{Step D: Recover the slope function.}
\State Set $\widehat{\boldsymbol{\gamma}}=\boldsymbol{\gamma}^{(h+1)}$, $\widehat{\sigma}=\sigma^{(h+1)}$, $\widehat{\vartheta}=\vartheta^{(h+1)}$, and $\widehat{\beta}(t)=\sum_{k=1}^{K}\widehat{\gamma}_{k+1}\widehat{\phi}_k(t)$.
\State \Return $\widehat{\boldsymbol{\gamma}}$, $\widehat{\sigma}$, $\widehat{\vartheta}$, $\widehat{\beta}(\cdot)$.
\end{algorithmic}
\end{algorithm}

In all numerical experiments reported in Sections~\ref{sec:simulation} and \ref{sec:real_data}, we use $J_{\max}=10$, $\tau=0.5$, and the interval $[-0.99,0.99]$ for $\vartheta$. We iterate the three updates until the maximum relative change in $(\boldsymbol{\gamma},\vartheta,\sigma)$ falls below a tolerance $\texttt{tol}$ (e.g., $10^{-6}$). The complete procedure is summarized in Algorithm~\ref{alg:FCSAR}. Upon convergence, the slope function is recovered as $\widehat{\beta}(t) = \sum_{k=1}^K \widehat{\gamma}_{k+1}\widehat{\phi}_k(t)$, where $\widehat{\gamma}_{k}$ are the elements of $\widehat{\boldsymbol{\gamma}}$ corresponding to the RFPCA scores. The resulting estimator is Fisher-consistent by construction, since it solves the bias-corrected estimating system. Moreover, the use of redescending score functions substantially reduces the impact of extreme residual outliers, while the hybrid IRLS-Newton scheme provides a fully specified and reproducible update for the spatial parameter under strong spatial dependence and contamination. Its computational cost and convergence behavior are examined empirically in Section~\ref{sec:simulation}. In all numerical experiments, the redescending stage uses the standard tuning constants described in Section~\ref{subsec:Mest}: for Andrews' sine loss we take $c=1.339$, and for Danish loss we take $c=2.767$ together with $v=3.915$. These are conventional choices that balance robustness and efficiency and are consistent with the calibration constants used in the scale equation. For the RFPCA stage, the retained dimension is selected by the $95\%$ explained-variation rule described in Section~\ref{subsec:rfpca}. In the empirical implementation, this choice is intended as a practical and stable dimension-selection device rather than as a tuning parameter optimized separately for each method.

For ease of presentation, Figure~\ref{fig:fcsar_flowchart} summarizes the main computational steps of Algorithm~\ref{alg:FCSAR}.

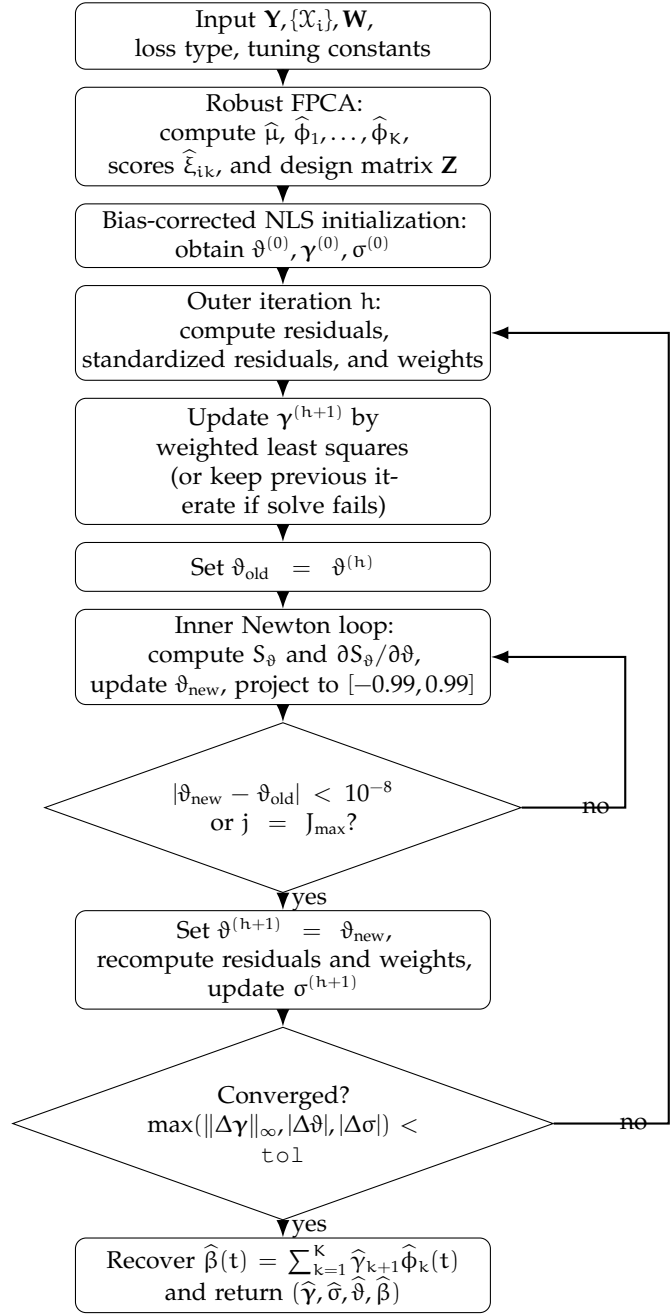
\begin{figure}[!htb]
\centering
\begin{tikzpicture}[scale=0.76, transform shape, node distance=3mm and 8mm, >=Latex, every node/.style={font=\normalfont}, block/.style={rectangle, rounded corners, draw, align=center, minimum width=7.0cm, minimum height=0.85cm, text width=7.0cm, inner sep=3pt},
decision/.style={diamond, draw, align=center, aspect=2.8, inner sep=1pt, text width=5.4cm}, line/.style={->, thick}]

\node[block] (input) {Input $\mathbf{Y}, \{\mathcal{X}_i\}, \mathbf{W}$,\\ loss type, tuning constants};
\node[block, below=of input] (rfpca) {Robust FPCA:\\ compute $\widehat{\mu}$, $\widehat{\phi}_1,\ldots,\widehat{\phi}_K$,\\ scores $\widehat{\xi}_{ik}$, and design matrix $\mathbf{Z}$};
\node[block, below=of rfpca] (init) {Bias-corrected NLS initialization:\\ obtain $\vartheta^{(0)}, \boldsymbol{\gamma}^{(0)}, \sigma^{(0)}$};
\node[block, below=of init] (weights) {Outer iteration $h$: compute residuals,\\ standardized residuals, and weights};
\node[block, below=of weights] (gamma) {Update $\boldsymbol{\gamma}^{(h+1)}$ by weighted least squares\\ (or keep previous iterate if solve fails)};
\node[block, below=of gamma] (newtoninit) {Set $\vartheta_{\text{old}}=\vartheta^{(h)}$};
\node[block, below=of newtoninit] (newton) {Inner Newton loop:\\ compute $S_\vartheta$ and $\partial S_\vartheta/\partial\vartheta$,\\ update $\vartheta_{\text{new}}$, project to $[-0.99,0.99]$};
\node[decision, below=of newton] (innerstop) {$|\vartheta_{\text{new}}-\vartheta_{\text{old}}|<10^{-8}$\\ or $j=J_{\max}$?};
\node[block, below=of innerstop] (sigma) {Set $\vartheta^{(h+1)}=\vartheta_{\text{new}}$,\\ recompute residuals and weights,\\ update $\sigma^{(h+1)}$};
\node[decision, below=of sigma] (outerstop) {Converged?\\ $\max(\|\Delta\boldsymbol{\gamma}\|_\infty,|\Delta\vartheta|,|\Delta\sigma|)<\texttt{tol}$};
\node[block, below=of outerstop] (output) {Recover $\widehat{\beta}(t)=\sum_{k=1}^K \widehat{\gamma}_{k+1}\widehat{\phi}_k(t)$\\ and return $(\widehat{\boldsymbol{\gamma}},\widehat{\sigma},\widehat{\vartheta},\widehat{\beta})$};

\draw[line] (input) -- (rfpca);
\draw[line] (rfpca) -- (init);
\draw[line] (init) -- (weights);
\draw[line] (weights) -- (gamma);
\draw[line] (gamma) -- (newtoninit);
\draw[line] (newtoninit) -- (newton);
\draw[line] (newton) -- (innerstop);
\draw[line] (innerstop) -- node[right] {yes} (sigma);
\draw[line] (sigma) -- (outerstop);
\draw[line] (outerstop) -- node[right] {yes} (output);

\draw[line] (innerstop.east) -- ++(1.8,0) node[midway,right] {no} |- (newton.east);
\draw[line] (outerstop.east) -- ++(2.0,0) node[midway,right] {no} |- (weights.east);

\end{tikzpicture}
\caption{\small Flowchart of the proposed FCSAR estimation algorithm.}
\label{fig:fcsar_flowchart}
\end{figure}

\section{Asymptotic properties}\label{sec:asymptotic}

This section develops a large-sample theory for the proposed FCSAR estimator under spatial dependence. For readability, all technical regularity conditions are collected in Appendix~\ref{sec:appendix_conditions}, and all proofs are deferred to Appendix~\ref{sec:proofs}. The scope of the theory is intentionally limited to a finite-dimensional projected setting. Throughout, we treat the truncation level $K$ as fixed, so that after RFPCA the problem reduces to a finite-dimensional SAR model with generated regressors. In addition, Assumption~\ref{as:4} rules out truncation bias by requiring that the slope function be represented exactly in the span used for projection. Consequently, the present theory does not cover growing-$K$ asymptotics, data-driven truncation selection, or approximation error from an infinite-dimensional slope lying outside the retained basis. The asymptotic theory is not intended to justify the data-driven component-selection rule used later in the numerical studies, where the retained dimension is chosen by a $95\%$ explained-variation criterion. That rule is adopted there as a practical finite-sample implementation device, whereas the theory here isolates the behavior of the estimator under a fixed-$K$, exact-projection regime.

A further limitation concerns the dependence structure of the functional predictors. Although the model is formulated for spatially indexed functional covariates, the asymptotic theory does not address a general spatially dependent functional predictor field. Instead, for theoretical tractability and to rely on available consistency results for the robust projection-pursuit RFPCA step, Assumption~\ref{as:2} treats $\mathcal{X}_1,\ldots,\mathcal{X}_n$ as i.i.d. copies of a square-integrable process. Thus, in the present theory, spatial dependence is introduced through the SAR response mechanism and the sequence of spatial weight matrices, rather than through a spatial dependence model for the predictor curves themselves.

The asymptotic development is also derived under i.i.d. functional predictors and i.i.d. symmetric errors that are independent of the regressors, together with a triangular-array framework for the spatial weight matrices. Moreover, the consistency and asymptotic normality results rely on high-level conditions ensuring identification of the population estimating map, uniform convergence of the sample estimating equations, and a suitable limiting empirical-process behavior. These assumptions are restrictive, but they provide a tractable first framework for isolating the large-sample behavior of the proposed bias-corrected redescending estimator in the projected SSoFRM.

\subsection{Fisher consistency of the proposed estimator}\label{subsec:fisher}

Because the spatial lag $(\mathbf{W}\mathbf{Y})_i$ is endogenous in SAR models, the naive M-score equation 
\[
\sum_{i=1}^n \psi_2(e_i/\sigma)(\mathbf{W}\mathbf{Y})_i=0,
\]
is generally biased and does not vanish at the true parameter. We therefore incorporate the analytical bias correction term $\mathrm{bias}_i(\boldsymbol{\theta})$ directly into the estimating equation for $\vartheta$. The next result formalizes that the resulting procedure is Fisher consistent.

\begin{theorem}\label{thm:FCSAR_fisher}
Suppose that: 
\begin{inparaenum}
\item[(i)] the functional predictor $\mathcal{X}$ satisfies the assumptions of Lemma~\ref{lem:rfpca_fisher}; 
\item[(ii)] assumptions~\ref{as:4},~\ref{as:5}, and~\ref{as:6} hold; and \item[(iii)] the population bias-corrected estimating equation has a unique root at $\boldsymbol{\theta}_0$. 
\end{inparaenum}
Then the M-functional defined by the bias-corrected estimating system is Fisher consistent:
\begin{equation*}
\widehat{\beta}(\mathbb{P}_{\mathcal{X} Y})(t) = \beta(t)\ \ (\forall t), \qquad \widehat{\beta}_0(\mathbb{P}_{\mathcal{X} Y}) = \beta_0, \qquad \widehat{\vartheta}(\mathbb{P}_{\mathcal{X} Y}) = \vartheta_0, \qquad \widehat{\sigma}(\mathbb{P}_{\mathcal{X} Y}) = \sigma_0,
\end{equation*}
where $\mathbb{P}_{\mathcal{X} Y}$ denotes the joint distribution of $(Y,\mathcal{X})$ under the true parameter $\boldsymbol{\theta}_0 = (\beta_0,\mathbf{b}^\top,\sigma_0,\vartheta_0)^\top$.
\end{theorem}

Theorem~\ref{thm:FCSAR_fisher} shows that the proposed estimating system is correctly centered at the model, in the sense that its population counterpart vanishes at the true parameter. The proof is given in Appendix~\ref{sec:proofs}, building on Lemmas~\ref{lemma:bias_term}, \ref{lem:rfpca_fisher}, and \ref{lem:fisher_consistency}.

\subsection{Consistency}\label{subsec:consistency}

Let $\boldsymbol{\theta}=(\boldsymbol{\gamma}^\top,\sigma,\vartheta)^\top \in \Theta \subset \mathbb{R}^{K+3}$. Define the sample estimating map
\begin{equation}\label{eq:Psi_n_def}
\boldsymbol{\Psi}_n(\boldsymbol{\theta}) := \frac{1}{n}\sum_{i=1}^n \boldsymbol{\psi}_{2,n,i}(\boldsymbol{\theta}),
\end{equation}
where $\boldsymbol{\psi}_{2,n,i}(\boldsymbol{\theta})$ is the $(K+3)$-vector stacking the three unbiased estimating contributions in \eqref{eq:unbiased_gamma}-\eqref{eq:unbiased_theta}. Let the population map be $\boldsymbol{\Psi}(\boldsymbol{\theta}) := \lim_{n\to\infty}\mathbb{E}\{\boldsymbol{\Psi}_n(\boldsymbol{\theta})\}$. By Lemma~\ref{lem:fisher_consistency}, $\boldsymbol{\Psi}(\boldsymbol{\theta}_0)=\mathbf{0}$.
We view $\widehat{\boldsymbol{\theta}}_n$ as a Z-estimator defined by solving $\boldsymbol{\Psi}_n(\widehat{\boldsymbol{\theta}}_n)=\mathbf{0}$ (or approximately solving it in the sense that
$\|\boldsymbol{\Psi}_n(\widehat{\boldsymbol{\theta}}_n)\|=o_p(1)$).

\begin{theorem}\label{thm:consistency_main}
Let Assumptions~\ref{as:1}-\ref{as:13} hold. Define $\boldsymbol{\Psi}_n(\boldsymbol{\theta})$ as in \eqref{eq:Psi_n_def}. Then, any sequence of estimators $\{\widehat{\boldsymbol{\theta}}_n\}$ satisfying $\boldsymbol{\Psi}_n(\widehat{\boldsymbol{\theta}}_n)=\mathbf{0}$ or $\|\boldsymbol{\Psi}_n(\widehat{\boldsymbol{\theta}}_n)\|=o_p(1)$ also satisfies
\begin{equation*}
\widehat{\boldsymbol{\theta}}_n \xrightarrow{p} \boldsymbol{\theta}_0 \qquad\text{as } n\to\infty.
\end{equation*}
\end{theorem}

The proof follows a standard $Z$-estimation argument: Assumption~\ref{as:10} guarantees identification of the unique zero of the population estimating map, while Assumption~\ref{as:13} provides uniform convergence of the sample estimating map over the compact parameter space $\Theta$. A full proof is given in Appendix~\ref{sec:proofs}.

\subsection{Asymptotic normality}\label{subsec:asymp_normal}

We next derive the limiting distribution of the full parameter estimator. Let 
\[
\boldsymbol{\mathcal A}(\boldsymbol{\theta}_0) :=\lim_{n\to\infty} \mathbb{E} \left\{ \frac{\partial}{\partial\boldsymbol{\theta}^\top}
\boldsymbol{\Psi}_n(\boldsymbol{\theta}_0) \right\},
\]
denote the limiting Jacobian matrix, and let
\begin{equation}\label{eq:longrunB}
\boldsymbol{\mathcal B}(\boldsymbol{\theta}_0) := \lim_{n\to\infty}
\operatorname{Var} \left\lbrace \frac{1}{\sqrt{n}} \sum_{i=1}^n \boldsymbol{\psi}_{2,n,i}(\boldsymbol{\theta}_0) \right\rbrace,
\end{equation}
be the asymptotic covariance matrix of the empirical estimating process. Under Assumption~\ref{as:14}, the limit in \eqref{eq:longrunB} exists and is finite.

\begin{theorem}\label{thm:asymptotic_normality_main}
Assume Assumptions~\ref{as:1}-\ref{as:14} hold. Let $\widehat{\boldsymbol{\theta}}_n$ satisfy either $\boldsymbol{\Psi}_n(\widehat{\boldsymbol{\theta}}_n)=\mathbf{0}$, or, more generally, $\|\boldsymbol{\Psi}_n(\widehat{\boldsymbol{\theta}}_n)\|=o_p(n^{-1/2})$, and suppose $\widehat{\boldsymbol{\theta}}_n \xrightarrow{p} \boldsymbol{\theta}_0$. Then
\begin{equation*}
\sqrt{n}(\widehat{\boldsymbol{\theta}}_n-\boldsymbol{\theta}_0)
\ \Rightarrow\ N \left(\mathbf{0}, \boldsymbol{\mathcal V}\right),
\qquad \boldsymbol{\mathcal V} := \boldsymbol{\mathcal A}^{-1}\boldsymbol{\mathcal B}\boldsymbol{\mathcal A}^{-\top},
\end{equation*}
where $\boldsymbol{\mathcal A}=\boldsymbol{\mathcal A}(\boldsymbol{\theta}_0)$ and $\boldsymbol{\mathcal B}=\boldsymbol{\mathcal B}(\boldsymbol{\theta}_0)$.
\end{theorem}

The proof is given in Appendix~\ref{sec:proofs}. Theorem~\ref{thm:asymptotic_normality_main} yields Wald-type inference for any smooth functional of $\boldsymbol{\theta}$
once $\boldsymbol{\mathcal{V}}$ is estimated consistently.

\begin{corollary}\label{cor:sandwich}
Under the conditions of Theorem~\ref{thm:asymptotic_normality_main}, a consistent estimator of $\boldsymbol{\mathcal{V}}$ is
\begin{equation*}
\widehat{\boldsymbol{\mathcal{V}}} := \widehat{\boldsymbol{\mathcal{A}}}^{-1} \widehat{\boldsymbol{\mathcal{B}}} \widehat{\boldsymbol{\mathcal{A}}}^{-\top},
\end{equation*}
where
$\widehat{\boldsymbol{\mathcal A}} := \frac{\partial} {\partial\boldsymbol{\theta}^\top} \boldsymbol{\Psi}_n(\widehat{\boldsymbol{\theta}}_n)$, and $\widehat{\boldsymbol{\mathcal B}}$ is any consistent estimator of $\operatorname{Var} \left( \frac{1}{\sqrt{n}} \sum_{i=1}^n \boldsymbol{\psi}_{2,n,i}(\widehat{\boldsymbol{\theta}}_n) \right)$. Consequently, for any fixed vector $\mathbf{c}\in\mathbb{R}^{K+3}$,
\begin{equation*}
\sqrt{n} \mathbf{c}^\top(\widehat{\boldsymbol{\theta}}_n-\boldsymbol{\theta}_0)
\ \Rightarrow\ N \left(0,\ \mathbf{c}^\top\boldsymbol{\mathcal{V}}\mathbf{c}\right),
\qquad \mathbf{c}^\top\widehat{\boldsymbol{\mathcal{V}}}\mathbf{c} \xrightarrow{p} \mathbf{c}^\top\boldsymbol{\mathcal{V}}\mathbf{c}.
\end{equation*}
\end{corollary}

\subsection{Asymptotics of the reconstructed slope function}\label{subsec:asymp_beta}

Recall that the reconstructed slope function is $\widehat{\beta}(t)=\sum_{k=1}^K \widehat{b}_k \widehat{\phi}_k(t)$. Because $\widehat{\beta}$ depends on both the regression coefficients and the estimated RFPCA eigenfunctions, its limit law follows from a functional delta-method argument.

\begin{theorem}\label{thm:asymp_beta_main}
Let Assumptions~\ref{as:1}-\ref{as:14} hold. Assume further that the eigenvalues $\lambda_1>\cdots>\lambda_K>0$ of the covariance operator are distinct and that the eigenfunctions $\phi_k$ are uniquely defined up to sign. Then, as $n\to\infty$,
\begin{equation*}
\sqrt{n}\bigl(\widehat{\beta}-\beta\bigr) \xrightarrow{d} \mathcal{Z} \qquad\text{in }\mathcal{L}^2[0,1],
\end{equation*}
where $\mathcal{Z}$ is a mean-zero Gaussian random element with covariance operator characterized by
\begin{align*}
\operatorname{Cov}\bigl(\langle \mathcal{Z},f\rangle,\langle \mathcal{Z},g\rangle\bigr)
&= \sum_{k,j=1}^K (\boldsymbol{\mathcal{V}}_{\mathbf{bb}})_{kj} \langle \phi_k,f\rangle \langle \phi_j,g\rangle \\
&\quad + \sum_{k,j=1}^K b_k b_j
\mathbb{E}\Bigl\lbrace\langle \operatorname{IF}(\mathcal{X};\phi_k,\mathbb{P}_{\mathcal{X}}),f\rangle \langle \operatorname{IF}(\mathcal{X};\phi_j,\mathbb{P}_{\mathcal{X}}),g\rangle\Bigr\rbrace, \qquad \forall f,g\in\mathcal{L}^2[0,1].
\end{align*}
Here, $\boldsymbol{\mathcal{V}}_{\mathbf{bb}}$ denotes the $K\times K$ block of
$\boldsymbol{\mathcal{V}}=\boldsymbol{\mathcal{A}}^{-1}\boldsymbol{\mathcal{B}}\boldsymbol{\mathcal{A}}^{-\top}$ corresponding to $\sqrt{n}(\widehat{\mathbf{b}}-\mathbf{b})$ in Theorem~\ref{thm:asymptotic_normality_main},
and $\operatorname{IF}(\cdot;\phi_k,\mathbb{P}_{\mathcal{X}})$ is the influence function of the $k$\textsuperscript{th} RFPCA eigenfunction.
\end{theorem}

A complete proof is provided in Appendix~\ref{sec:proofs}. Theorem~\ref{thm:asymp_beta_main} directly implies asymptotic normality for any fixed linear functional of $\beta$.
\begin{corollary}\label{cor:beta_linear}
Under the conditions of Theorem~\ref{thm:asymp_beta_main}, for any fixed $f\in\mathcal{L}^2[0,1]$,
\begin{equation*}
\sqrt{n} \langle \widehat{\beta}-\beta,f\rangle \ \xrightarrow{d}\ N \left\{0,\ \operatorname{Var}\bigl(\langle \mathcal{Z},f\rangle\bigr)\right\}.
\end{equation*}
\end{corollary}

\begin{remark}\label{rem:beta_pointwise}
Pointwise inference for $\beta(t)$ at a fixed $t$ and the construction of uniform confidence bands typically require additional regularity to strengthen the convergence topology from $\mathcal{L}^2[0,1]$ to, e.g., $\mathcal{C}[0,1]$, so that both $\widehat{\beta}$ and the Gaussian limit $\mathcal{Z}$ admit continuous versions. We therefore focus on inference for linear functionals and $\mathcal{L}^2$-type summaries, which are justified under the present assumptions. In particular, the current theory, by itself, does not yield valid pointwise standard errors or simultaneous confidence bands for the entire coefficient function.
\end{remark}

\section{Simulation study}\label{sec:simulation}

We conduct a series of Monte Carlo experiments to evaluate the finite-sample estimation and prediction performance of the proposed Fisher-consistent redescending estimator under both clean and contaminated settings. Following the simulation architecture of \citet{Beyaztas2026}, we compare the proposed method with three competing procedures: 
\begin{inparaenum}
\item[(i)] the classical FPCA-based likelihood estimator, obtained by combining standard FPCA with the nonrobust likelihood-based estimation step, 
\item[(ii)] the robust two-stage estimator of \citet{Beyaztas2026}, which combines RFPCA with a Huber-type spatial M-estimation step, and
\item[(iii)] the non-Fisher-consistent redescending counterparts of the proposed approaches, denoted by FCSAR-A\textsuperscript{*} and FCSAR-D\textsuperscript{*}, respectively, which are obtained by combining RFPCA with the redescending M-estimators of \cite{Zhan2025}.
\end{inparaenum}
Since the proposed method can be implemented with different redescending loss functions, we consider two versions throughout the simulation study: the Andrews-loss version (FCSAR-A) and the Danish-loss version (FCSAR-D). For all FPCA- and RFPCA-based procedures, the number of retained components is selected so that the cumulative proportion of explained variation is at least $95\%$.

The component-selection rule used in this section is intended for practical implementation rather than for a direct verification of the asymptotic assumptions in Section~\ref{sec:asymptotic}. In particular, the theory there treats the truncation level as fixed, whereas here the retained dimension is selected in a data-driven manner by explained variation. Accordingly, the simulation study is designed to assess finite-sample estimation and prediction performance of the competing methods under practically relevant tuning, rather than to numerically verify every assumption used in the asymptotic development.

In all simulation settings, the functional predictor $\X_i(t)$ is observed on 101 equally spaced grid points over $\mathcal{I}=[0,1]$. For each Monte Carlo replication, we consider training sample sizes $n\in\{100,250,500\}$ and spatial dependence parameters $\vartheta\in\{0.25,0.50,0.75\}$. To align the simulation design with the exact-span assumption used in the asymptotic section, the regression coefficient function is taken to lie in the same finite-dimensional span as the generated functional predictors. Specifically, we set $\beta(t)=\sum_{j=1}^{5} b_j \upsilon_j(t)$, $t\in[0,1]$, where $b_j=j^{-1}$ and $\upsilon_j(t)=\sin(j\pi t)-\cos(j\pi t)$ for $j=1,\ldots,5$. The spatial weight matrix $\mathbf{W}=[w_{ij}]_{1\le i,j\le n}$ is constructed from a one-dimensional regular grid and then row-normalized. Specifically, for $i\neq j$, let $d_{ij}=|i-j|$ and define
\begin{equation*}
w_{ij}=\frac{d_{ij}^{-1}}{\sum_{\ell\neq i} d_{i\ell}^{-1}}, \qquad w_{ii}=0.
\end{equation*}
We use the inverse-distance specification $d_{ij}^{-1}$ to induce a moderate spatial decay over the one-dimensional grid while still allowing non-negligible contributions from more distant units. A faster decay, such as $d_{ij}^{-2}$, would produce a much more localized dependence structure, in which the nearest neighbors dominate more strongly, and would therefore represent a different spatial regime.

We consider the following two scenarios.
\begin{inparaenum}
\item[1)] For each $i\in\{1,\ldots,n\}$, the functional predictor is generated as $\X_i(t)=\sum_{j=1}^{5}\kappa_{ij}\upsilon_j(t)$, where $\kappa_{ij}\sim \mathcal{N}(0,4j^{-3/2})$ independently and $\upsilon_j(t)=\sin(j\pi t)-\cos(j\pi t)$, $j=1,\ldots,5$. Let $m_i=\int_0^1 \mathcal{X}_i(t)\beta(t) dt$ and $\mathbf{m}=(m_1,\ldots,m_n)^\top$. The error vector is generated as $\boldsymbol{\varepsilon}=(\varepsilon_1,\ldots,\varepsilon_n)^\top$, where $\varepsilon_i\sim \mathcal{N}(0,1)$ independently. The scalar response vector is then generated according to the SSoFRM:
\begin{equation*}
\mathbf{Y}=(\mathbf{I}_n-\vartheta\mathbf{W})^{-1}(\mathbf{m}+\boldsymbol{\varepsilon}).
\end{equation*}

\item[2)] To assess robustness, we contaminate both the functional predictor and the scalar response component. For each contamination level $\alpha\in\{0.05,0.10\}$, we randomly select a set $\mathcal{C}\subset\{1,\ldots,n\}$ with cardinality $|\mathcal{C}|=\lfloor \alpha n\rfloor$. For observations with indices in $\mathcal{C}$, the clean functional predictor is replaced by $\widetilde{\mathcal{X}}_i(t)=\sum_{j=1}^{5}\widetilde{\kappa}_{ij}\upsilon_j(t)$, where $\widetilde{\kappa}_{ij}\sim \mathcal{N}(1,4j^{-3/2})$ independently. This produces atypical curves of the same general form as the clean process but shifted away from the bulk of the sample. Correspondingly, for $i\in\mathcal{C}$, we define $\widetilde{m}_i=\int_0^1 \widetilde{\X}_i(t)\beta(t) dt$. In addition, the scalar response contamination is introduced through the error term. For contaminated observations, we generate $\widetilde{\varepsilon}_i\sim \mathcal{N}(m_{\varepsilon},1)$, where the contamination mean takes one of the two values $m_{\varepsilon}\in\{10,20\}$. The larger mean creates more extreme vertical outliers, and is therefore expected to highlight the advantage of the redescending estimators relative to monotone-loss robust procedures.

Hence, defining
\begin{equation*}
m_i^{\ast}=
\begin{cases}
\widetilde{m}_i, & i\in\mathcal{C},\\
m_i, & i\notin\mathcal{C},
\end{cases}
\qquad
\varepsilon_i^{\ast}=
\begin{cases}
\widetilde{\varepsilon}_i, & i\in\mathcal{C},\\
\varepsilon_i, & i\notin\mathcal{C},
\end{cases}
\end{equation*}
and writing $\mathbf{m}^{\ast}=(m_1^{\ast},\ldots,m_n^{\ast})^\top$ and $\boldsymbol{\varepsilon}^{\ast}=(\varepsilon_1^{\ast},\ldots,\varepsilon_n^{\ast})^\top$, the contaminated response vector is generated as
\begin{equation*}
\mathbf{Y}=(\mathbf{I}_n-\vartheta\mathbf{W})^{-1}(\mathbf{m}^{\ast}+\boldsymbol{\varepsilon}^{\ast}).
\end{equation*}
Therefore, the contaminated setting contains both functional leverage points and vertical outliers in the response.
\end{inparaenum}

For each configuration, the simulation experiment is repeated $B=500$ times. The estimation performance of the competing procedures is evaluated using root mean squared error (RMSE) for $\widehat{\vartheta}$ and $\widehat{\sigma}$ ($\operatorname{RMSE}(\widehat{\vartheta})$ and $\operatorname{RMSE}(\widehat{\sigma})$), and integrated mean squared error (IMSE) for $\widehat{\beta}$ ($\operatorname{IMSE}(\widehat{\beta}))$. Specifically, for $\widehat{\theta}\in\{\widehat{\vartheta},\widehat{\sigma}\}$, we compute
\begin{equation*}
\operatorname{RMSE}(\widehat{\theta}) = \left\{ \frac{1}{B}\sum_{s=1}^{B} \bigl(\widehat{\theta}^{(s)}-\theta_0\bigr)^2
\right\}^{1/2},
\end{equation*}
where $\widehat{\theta}^{(s)}$ denotes the estimate obtained from the $s^{\textsuperscript{th}}$ Monte Carlo replication and $\theta_0$ is the corresponding true parameter value. For the slope function, we compute
\begin{equation*}
\operatorname{IMSE}(\widehat{\beta}) = \frac{1}{B}\sum_{s=1}^{B}
\int_0^1 \left\{ \widehat{\beta}^{(s)}(t)-\beta(t)
\right\}^{2} dt,
\end{equation*}
where $\widehat{\beta}^{(s)}(t)$ is the estimated slope function from the $s^{\textsuperscript{th}}$ replication.

To assess prediction performance, for each replication and each simulation configuration, we generate an independent test sample of size $n_{\text{test}}=1000$. The test sample is generated from the same scenario as the training sample, using an independently generated spatial weight matrix $\mathbf{W}^{\text{new}}$ constructed in the same manner. Let $\X_i^{\text{new}}(t)$ and $Y_i^{\text{new}}$ denote the functional predictor and scalar response for the $i^{\textsuperscript{th}}$ test observation, respectively. The corresponding out-of-sample predicted response is obtained by
\begin{equation*}
\widehat{Y}_i^{\text{new}} = \left[ (\mathbf{I}_{n_{\text{test}}}-\widehat{\vartheta}\mathbf{W}^{\text{new}})^{-1} \widehat{\mathbf{m}}^{\text{new}}
\right]_i,
\end{equation*}
where
\begin{equation*}
\widehat{\mathbf{m}}^{\text{new}} = \left( \int_0^1 \X_1^{\text{new}}(t) \widehat{\beta}(t) dt, \ldots, \int_0^1 \X_{n_{\text{test}}}^{\text{new}}(t)\widehat{\beta}(t) dt \right)^\top.
\end{equation*}
The predictive performance is then summarized by the mean squared prediction error (MSPE)
\begin{equation*}
\operatorname{MSPE} = \frac{1}{Bn_{\text{test}}} \sum_{s=1}^{B}\sum_{i=1}^{n_{\text{test}}} \left( \widehat{Y}_{i}^{(s),\text{new}}-Y_{i}^{(s),\text{new}} \right)^2.
\end{equation*}

The results recorded from our Monte Carlo simulations are presented in Tables~\ref{tab:tab1}, \ref{tab:tab2_5}, and \ref{tab:tab2_10}. Table~\ref{tab:tab1} presents the results for the clean-data setting, whereas Tables~\ref{tab:tab2_5} and \ref{tab:tab2_10} report the results under the $5\%$ and $10\%$ contamination settings, respectively. Overall, the numerical findings clearly confirm the main motivation of the proposed methodology: when the data are generated without contamination, all procedures perform similarly, and the gain from robustification is naturally limited; however, once functional leverage points and vertical outliers are introduced, the proposed redescending FCSAR estimators become markedly superior to both the classical FPCA-based likelihood approach and the RFPCA-Huber estimator.
\begin{table}[!htb]
\centering
\caption{\small Monte Carlo summary measures for the clean-data setting: mean values of $\mathrm{RMSE}(\widehat{\vartheta})$, $\mathrm{RMSE}(\widehat{\sigma})$, $\mathrm{IMSE}(\widehat{\beta})$, and $\mathrm{MSPE}$ over $500$ replications. Results are presented for $n\in\{100,250,500\}$ and $\vartheta\in\{0.25,0.50,0.75\}$. Standard errors are given in parentheses for $\mathrm{IMSE}(\widehat{\beta})$ and $\mathrm{MSPE}$.}
\label{tab:tab1}
\tabcolsep 0.085in
\renewcommand{\arraystretch}{0.9}
\begin{small}
\begin{tabular}{@{}llccccccccccc@{}}
\toprule
& & \multicolumn{3}{c}{$\vartheta=0.25$}
& & \multicolumn{3}{c}{$\vartheta=0.50$}
& & \multicolumn{3}{c}{$\vartheta=0.75$} \\
\cmidrule(lr){3-5} \cmidrule(lr){7-9} \cmidrule(lr){11-13}
Metric & Method
& $n=100$ & $250$ & $500$
& & $100$ & $250$ & $500$
& & $100$ & $250$ & $500$ \\
\midrule

RMSE($\widehat{\vartheta}$) & FPCA & 0.15 & 0.08 & 0.09 && 0.12 & 0.09 &                              0.07 && 0.12 & 0.08 & 0.05 \\
                            & RFPCA & 0.15 & 0.09 & 0.09 && 0.12 & 0.09 & 0.07 && 0.17 & 0.08 & 0.05 \\
                            & FCSAR-D\textsuperscript{*} & 0.15 & 0.09 & 0.09 && 0.14 & 0.09 & 0.07 && 0.11 & 0.08 & 0.05 \\
                            & FCSAR-A\textsuperscript{*} & 0.15 & 0.09 & 0.09 && 0.13 & 0.09 & 0.07 && 0.11 & 0.08 & 0.05 \\
                            & FCSAR-D & 0.16 & 0.09 & 0.09 && 0.12 & 0.09 & 0.07 && 0.12 & 0.08 & 0.05 \\
                            & FCSAR-A & 0.15 & 0.09 & 0.09 && 0.12 & 0.09 & 0.07 && 0.12 & 0.08 & 0.05 \\

\cmidrule(lr){2-13}

RMSE($\widehat{\sigma}$)    & FPCA & 0.07 & 0.05 & 0.04 && 0.08 & 0.05 &                             0.04 && 0.07 & 0.05 & 0.04 \\
                            & RFPCA & 0.07 & 0.04 & 0.04 && 0.08 & 0.04 & 0.04 && 0.11 & 0.05 & 0.04 \\
                            & FCSAR-D\textsuperscript{*} & 0.13 & 0.07 & 0.05 && 0.12 & 0.08 & 0.05 && 0.13 & 0.08 & 0.05 \\
                            & FCSAR-A\textsuperscript{*} & 0.13 & 0.07 & 0.05 && 0.12 & 0.08 & 0.05 && 0.14 & 0.08 & 0.05 \\           & FCSAR-D & 0.08 & 0.07 & 0.07 && 0.09 & 0.08 & 0.07 && 0.10 & 0.08 & 0.08 \\
                            & FCSAR-A & 0.10 & 0.07 & 0.07 && 0.12 & 0.07 & 0.06 && 0.12 & 0.07 & 0.06 \\

\cmidrule(lr){2-13}

IMSE($\widehat{\beta}$)     & FPCA    & 0.18 & 0.17 & 0.16 && 0.17 &                             0.16 & 0.16 && 0.17 & 0.17 & 0.17 \\
                            &         & (0.05) & (0.05) & (0.05) && (0.05) & (0.05) & (0.05) && (0.05) & (0.05) & (0.05) \\
                            & RFPCA   & 0.16 & 0.15 & 0.13 && 0.16 & 0.15 & 0.13 && 0.17 & 0.14 & 0.15 \\
                            &         & (0.07) & (0.07) & (0.06) && (0.08) & (0.06) & (0.06) && (0.09) & (0.07) & (0.07) \\
                            & FCSAR-D\textsuperscript{*} & 0.16 & 0.15 & 0.13 && 0.16 & 0.15 & 0.13 && 0.17 & 0.14 & 0.15 \\
                            &         & (0.08) & (0.07) & (0.06) && (0.08) & (0.06) & (0.06) && (0.09) & (0.07) & (0.07) \\
                            & FCSAR-A\textsuperscript{*} & 0.16 & 0.15 & 0.13 && 0.16 & 0.15 & 0.13 && 0.17 & 0.14 & 0.15 \\
                            &         & (0.07) & (0.07) & (0.06) && (0.08) & (0.06) & (0.06) && (0.09) & (0.17) & (0.07) \\
                            & FCSAR-D & 0.16 & 0.15 & 0.13 && 0.16 & 0.15 & 0.13 && 0.17 & 0.14 & 0.15 \\
                            &         & (0.07) & (0.07) & (0.06) && (0.08) & (0.06) & (0.06) && (0.09) & (0.07) & (0.07) \\
                            & FCSAR-A & 0.16 & 0.15 & 0.13 && 0.16 & 0.15 & 0.13 && 0.17 & 0.14 & 0.15 \\
                            &         & (0.07) & (0.07) & (0.06) && (0.08) & (0.06) & (0.06) && (0.09) & (0.07) & (0.07) \\

\cmidrule(lr){2-13}

MSPE                        & FPCA    & 1.15 & 1.06 & 1.07 && 1.23 &                             1.14 & 1.09 && 2.04 & 1.35 & 1.23 \\
                            &         & (0.19) & (0.10) & (0.07) && (0.32) & (0.13) & (0.09) && (2.62) & (0.30) & (0.20) \\
                            & RFPCA   & 1.13 & 1.04 & 1.05 && 1.23 & 1.13 & 1.07 && 2.33 & 1.35 & 1.22 \\
                            &         & (0.20) & (0.09) & (0.07) && (0.31) & (0.13) & (0.08) && (2.62) & (0.36) & (0.20) \\
                            & FCSAR-D\textsuperscript{*} & 1.16 & 1.05 & 1.05 && 1.39 & 1.14 & 1.08 && 2.62 & 1.63 & 1.30 \\
                            &         & (0.32) & (0.09) & (0.07) && (0.88) & (0.14) & (0.09) && (3.37) & (1.32) & (0.38) \\
                            & FCSAR-A\textsuperscript{*} & 1.16 & 1.05 & 1.05 && 1.37 & 1.14 & 1.08 && 2.74 & 1.94 & 1.31 \\
                            &         & (0.36) & (0.09) & (0.07) && (0.89) & (0.14) & (0.08) && (3.88) & (3.79) & (0.44) \\
                            & FCSAR-D & 1.13 & 1.05 & 1.05 && 1.24 & 1.13 & 1.08 && 2.06 & 1.34 & 1.22 \\
                            &         & (0.20) & (0.09) & (0.07) && (0.34) & (0.13) & (0.08) && (2.59) & (0.31) & (0.22) \\
                            & FCSAR-A & 1.13 & 1.04 & 1.05 && 1.23 & 1.13 & 1.08 && 2.07 & 1.36 & 1.22    \\
                            &         & (0.20) & (0.09) & (0.07) && (0.34) & (0.13) & (0.08) && (1.91) & (0.40) & (0.20)  \\

\bottomrule
\end{tabular}
\end{small}
\end{table}

We first consider the uncontaminated setting in Table~\ref{tab:tab1}. In this case, all six procedures perform similarly in the estimation of the spatial dependence parameter, and the values of $\mathrm{RMSE}(\widehat{\vartheta})$ decrease steadily as the sample size increases. This indicates that, in the absence of contamination, neither the redescending losses nor the Fisher-consistent bias correction introduces any appreciable loss in first-order estimation accuracy for the spatial parameter. A similar conclusion holds for prediction: the $\mathrm{MSPE}$ values are very close across methods, especially for moderate and large sample sizes, so the proposed procedures do not sacrifice predictive performance under the reference model. The main differences in Table~\ref{tab:tab1} appear in scale estimation. Because the data are Gaussian and uncontaminated, the classical FPCA-based likelihood estimator and, in several configurations, the RFPCA-Huber estimator achieve slightly smaller values of $\mathrm{RMSE}(\widehat{\sigma})$ than the redescending procedures. This reflects the expected robustness-efficiency trade-off: when no outliers are present, aggressive downweighting is unnecessary and may induce a mild efficiency loss. Importantly, however, this loss remains limited. The non-Fisher-consistent redescending benchmarks FCSAR-D\textsuperscript{*} and FCSAR-A\textsuperscript{*} behave similarly to their Fisher-consistent counterparts in the uncontaminated regime, while the values of $\mathrm{IMSE}(\widehat{\beta})$ and $\mathrm{MSPE}$ remain very close across all redescending procedures. Thus, Table~\ref{tab:tab1} shows that the proposed FCSAR estimators remain highly competitive under the reference model, with only a modest loss in scale efficiency and essentially no degradation in estimation of $\vartheta$, recovery of $\beta(\cdot)$, or prediction.

The situation changes substantially once contamination is introduced. Table~\ref{tab:tab2_5} reports the results for the $5\%$ contamination setting. Even under this moderate contamination level, the classical FPCA-based likelihood estimator deteriorates sharply, and the deterioration is especially pronounced for $\mathrm{RMSE}(\widehat{\sigma})$, $\mathrm{IMSE}(\widehat{\beta})$, and $\mathrm{MSPE}$. The robust two-stage RFPCA-Huber procedure is much more stable than FPCA, confirming the benefit of robust functional dimension reduction together with bounded-loss spatial estimation. Nevertheless, Table~\ref{tab:tab2_5} also shows that this is insufficient to fully control the impact of contamination.

The non-Fisher-consistent redescending benchmarks, FCSAR-D\textsuperscript{*} and FCSAR-A\textsuperscript{*}, provide an informative intermediate comparison. Relative to RFPCA-Huber, they already improve markedly in most contaminated configurations, showing that the use of redescending losses in the second-stage spatial estimator is itself highly beneficial. However, the proposed Fisher-consistent procedures, FCSAR-D and FCSAR-A, generally yield an additional improvement over their non-Fisher-consistent counterparts, particularly for scale estimation, slope recovery, and prediction. This pattern supports the methodological claim that redescending estimation and Fisher-consistent bias correction play complementary roles in contaminated SSoFRM problems.

\begin{center}
\tabcolsep 0.07in
\renewcommand{\arraystretch}{0.9}
\begin{small}
\begin{longtable}{@{}lllccccccccccc@{}}
\caption{\small Monte Carlo mean values of $\mathrm{RMSE}(\widehat{\vartheta})$, $\mathrm{RMSE}(\widehat{\sigma})$, $\mathrm{IMSE}(\widehat{\beta})$, and $\mathrm{MSPE}$ under the $5\%$ contamination setting. Results are based on $500$ replications for $n\in\{100,250,500\}$ and $\vartheta\in\{0.25,0.50,0.75\}$, with contaminated response errors generated from $\widetilde{\varepsilon}\sim N(\mu_{\widetilde{\varepsilon}},1)$ for $\mu_{\widetilde{\varepsilon}}\in\{10,20\}$. Standard errors are reported in parentheses for $\mathrm{IMSE}(\widehat{\beta})$ and $\mathrm{MSPE}$. Here, FCSAR-D and FCSAR-A denote the proposed FCSAR estimators based on Danish and Andrews losses, respectively.}
\label{tab:tab2_5} \\
\toprule
& & & \multicolumn{3}{c}{$\vartheta=0.25$}
& & \multicolumn{3}{c}{$\vartheta=0.50$}
& & \multicolumn{3}{c}{$\vartheta=0.75$} \\
\cmidrule(lr){4-6} \cmidrule(lr){8-10} \cmidrule(lr){12-14}
$\mu_{\widetilde{\varepsilon}}$ & Metric & Method
& $n=100$ & $250$ & $500$
& & $100$ & $250$ & $500$
& & $100$ & $250$ & $500$ \\
\midrule
\endfirsthead
\toprule
& & & \multicolumn{3}{c}{$\vartheta=0.25$}
& & \multicolumn{3}{c}{$\vartheta=0.50$}
& & \multicolumn{3}{c}{$\vartheta=0.75$} \\
\cmidrule(lr){4-6} \cmidrule(lr){8-10} \cmidrule(lr){12-14}
$\mu_{\widetilde{\varepsilon}}$ & Metric & Method
& $n=100$ & $250$ & $500$
& & $100$ & $250$ & $500$
& & $100$ & $250$ & $500$ \\
\midrule
\endhead
\midrule
\multicolumn{14}{r}{Continued on next page} \\
\endfoot
\endlastfoot
\multirow[t]{24}{*}{$10$}
& \multirow[t]{4}{*}{RMSE($\widehat{\vartheta}$)}
& FPCA & 0.25 & 0.17 & 0.15 && 0.23 & 0.17 & 0.12 && 0.22 & 0.13 & 0.09 \\
& & RFPCA & 0.16 & 0.10 & 0.07 && 0.13 & 0.09 & 0.06 && 0.42 & 0.07 & 0.06 \\
& & FCSAR-D\textsuperscript{*} & 0.14 & 0.09 & 0.07 && 0.11 & 0.07 & 0.06 && 0.10 & 0.06 & 0.05 \\
& & FCSAR-A\textsuperscript{*} & 0.14 & 0.09 & 0.06 && 0.11 & 0.07 & 0.06 && 0.10 & 0.06 & 0.05 \\
& & FCSAR-D & 0.14 & 0.10 & 0.07 && 0.11 & 0.07 & 0.06 && 0.09 & 0.06 & 0.05 \\
& & FCSAR-A & 0.14 & 0.10 & 0.07 && 0.12 & 0.07 & 0.06 && 0.10 & 0.07 & 0.05 \\

\cmidrule(lr){2-14}
& \multirow[t]{4}{*}{RMSE($\widehat{\sigma}$)}
& FPCA & 1.39 & 1.32 & 1.36 && 1.36 & 1.34 & 1.38 && 1.35 & 1.34 & 1.37 \\
& & RFPCA & 0.19 & 0.18 & 0.20 && 0.20 & 0.19 & 0.20 && 0.39 & 0.19 & 0.21 \\
& & FCSAR-D\textsuperscript{*} & 0.15 & 0.09 & 0.08 && 0.14 & 0.10 & 0.09 && 0.15 & 0.08 & 0.07 \\
& & FCSAR-A\textsuperscript{*} & 0.15 & 0.09 & 0.08 && 0.13 & 0.10 & 0.09 && 0.15 & 0.08 & 0.06 \\
& & FCSAR-D & 0.08 & 0.06 & 0.04 && 0.08 & 0.06 & 0.05 && 0.09 & 0.06 & 0.04 \\
& & FCSAR-A & 0.15 & 0.12 & 0.10 && 0.14 & 0.11 & 0.10 && 0.15 & 0.11 & 0.10 \\

\cmidrule(lr){2-14}
& \multirow[t]{8}{*}{IMSE($\widehat{\beta}$)}
& FPCA & 0.69 & 0.72 & 0.80 && 0.67 & 0.67 & 0.75 && 0.85 & 0.67 & 0.77 \\
& & & (0.56) & (0.33) & (0.27) && (0.51) & (0.32) & (0.25) && (0.64) & (0.37) & (0.29) \\
& & RFPCA & 0.21 & 0.14 & 0.14 && 0.19 & 0.14 & 0.14 && 0.23 & 0.16 & 0.14 \\
& & & (0.12) & (0.06) & (0.05) && (0.12) & (0.05) & (0.05) && (0.27) & (0.06) & (0.05)  \\
& & FCSAR-D\textsuperscript{*} & 0.15 & 0.12 & 0.10 && 0.15 & 0.12 & 0.11 && 0.15 & 0.13 & 0.10 \\
& & & (0.09) & (0.06) & (0.05) && (0.08) & (0.05) & (0.05) && (0.07) & (0.06) & (0.05) \\
& & FCSAR-A\textsuperscript{*} & 0.15 & 0.12 & 0.10 && 0.15 & 0.12 & 0.11 && 0.15 & 0.13 & 0.10 \\
& & & (0.09) & (0.06) & (0.05) && (0.08) & (0.05) & (0.05) && (0.07) & (0.06) & (0.05) \\
& & FCSAR-D & 0.15 & 0.12 & 0.10 && 0.15 & 0.12 & 0.11 && 0.15 & 0.13 & 0.10 \\
& & & (0.09) & (0.06) & (0.05) && (0.08) & (0.05) & (0.05) && (0.07) & (0.06) & (0.05) \\
& & FCSAR-A & 0.15 & 0.12 & 0.10 && 0.15 & 0.12 & 0.11 && 0.15 & 0.13 & 0.10 \\
& & & (0.09) & (0.06) & (0.05) && (0.08) & (0.05) & (0.05) && (0.07) & (0.06) & (0.05) \\

\cmidrule(lr){2-14}
& \multirow[t]{8}{*}{MSPE}
& FPCA & 2.11 & 1.90 & 1.93 && 3.04 & 2.40 & 2.44 && 7.13 & 4.98 & 5.32 \\
& & & (0.50) & (0.30) & (0.22) && (1.33) & (0.54) & (0.33) && (5.07) & (1.99) & (1.57) \\
& & RFPCA & 1.20 & 1.11 & 1.09 && 1.40 & 1.18 & 1.15 && 5.34 & 1.64 & 1.55 \\
& & & (0.23) & (0.12) & (0.07) && (0.38) & (0.16) & (0.12) && (8.16) & (0.64) & (0.52) \\
& & FCSAR-D\textsuperscript{*} & 1.15 & 1.08 &  1.04 && 1.71 & 1.15 & 1.10 && 6.48 & 2.31 & 1.48 \\
& & & (0.22) & (0.13) & (0.07) && (2.72) & (0.21) & (0.12) && (13.24) & (3.58) & (0.88) \\
& & FCSAR-A\textsuperscript{*} & 1.15 & 1.08 & 1.04 && 1.61 & 1.15 & 1.10 && 5.39 & 2.10 & 1.48 \\
& & & (0.22) & (0.13) & (0.06) && (2.26) & (0.22) & (0.11) && (10.06) & (2.25) & (0.86) \\
& & FCSAR-D & 1.13 & 1.08 & 1.04 && 1.39 & 1.13 & 1.09 && 2.60 & 1.63 & 1.35 \\
& & & (0.20) & (0.11) & (0.06) && (0.68) & (0.14) & (0.10) && (2.48) & (0.93) & (0.52) \\
& & FCSAR-A & 1.14 & 1.08 & 1.04 && 1.39 & 1.14 & 1.09 && 2.52 & 1.61 & 1.34 \\
& & & (0.20) & (0.12) & (0.06) && (0.62) & (0.15) & (0.09) && (2.09) & (0.82) & (0.47) \\

\midrule

\multirow[t]{24}{*}{$20$}
& \multirow[t]{4}{*}{RMSE($\widehat{\vartheta}$)}
& FPCA & 0.29 & 0.23 & 0.18 && 0.29 & 0.21 & 0.13 && 0.27 & 0.17 & 0.10 \\
& & RFPCA & 0.13 & 0.13 & 0.11 && 0.11 & 0.09 & 0.05 && 0.59 & 0.61 & 0.66 \\
& & FCSAR-D\textsuperscript{*} & 0.08 & 0.07 & 0.04 && 0.09 & 0.06 & 0.05 && 0.07 & 0.05 & 0.03 \\
& & FCSAR-A\textsuperscript{*} & 0.08 & 0.07 & 0.05 && 0.09 & 0.06 & 0.05 && 0.07 &  0.05 & 0.03 \\
& & FCSAR-D & 0.08 & 0.07 & 0.05 && 0.09 & 0.06 & 0.05 && 0.07 & 0.05 & 0.03 \\
& & FCSAR-A & 0.09 & 0.07 & 0.05 && 0.09 & 0.06 & 0.05 && 0.08 & 0.05 & 0.03 \\

\cmidrule(lr){2-14}
& \multirow[t]{4}{*}{RMSE($\widehat{\sigma}$)}
& FPCA & 3.31 & 3.25 & 3.35 && 3.29 & 3.25 & 3.34 && 3.32 & 3.25 & 3.33 \\
& & RFPCA & 0.18 & 0.21 & 0.21 && 0.20 & 0.19 & 0.21 && 0.78 & 0.65 & 0.74 \\
& & FCSAR-D\textsuperscript{*} & 0.13 & 0.10 & 0.10 && 0.13 & 0.09 & 0.09 && 0.13 & 0.08 & 0.07 \\
& & FCSAR-A\textsuperscript{*} & 0.14 & 0.10 & 0.10 && 0.13 & 0.09 & 0.09 && 0.13 & 0.08 & 0.07 \\
& & FCSAR-D & 0.09 & 0.06 & 0.05 && 0.08 & 0.06 & 0.05 && 0.09 & 0.06 & 0.05 \\
& & FCSAR-A & 0.15 & 0.11 & 0.10 && 0.15 & 0.11 & 0.10 && 0.14 & 0.11 & 0.10 \\

\cmidrule(lr){2-14}
& \multirow[t]{8}{*}{IMSE($\widehat{\beta}$)}
& FPCA & 2.58 & 2.41 & 2.66 && 2.88 & 2.67 & 2.77 && 2.86 & 2.75 & 2.84 \\
& & & (2.00) & (1.36) & (1.13) && (2.67) & (1.41) & (1.04) && (2.38) & (1.65) & (1.17) \\
& & RFPCA & 0.22 & 0.15 & 0.14 && 0.19 & 0.14 & 0.13 && 0.32 & 0.20 & 0.21 \\
& & & (0.21) & (0.07) & (0.05) && (0.10) & (0.06) & (0.05) && (0.37) & (0.17) & (0.14) \\
& & FCSAR-D\textsuperscript{*} & 0.15 & 0.12 & 0.11 && 0.14 & 0.11 & 0.10 && 0.15 & 0.13 & 0.10 \\
& & & (0.09) & (0.06) & (0.06) && (0.08) & (0.06) & (0.05) && (0.08) & (0.07) & (0.06) \\
& & FCSAR-A\textsuperscript{*} & 0.15 & 0.12 & 0.11 && 0.14 & 0.11 & 0.10 && 0.15 & 0.13 & 0.10 \\
& & & (0.10) & (0.06) & (0.06) && (0.08) & (0.05) &  (0.05) && (0.08) & (0.07) & (0.06) \\
& & FCSAR-D & 0.15 & 0.12 & 0.11 && 0.14 & 0.11 & 0.10 && 0.15 & 0.13 & 0.10 \\
& & & (0.10) & (0.06) & (0.06) && (0.08) & (0.06) & (0.05) && (0.08) & (0.07) & (0.06) \\
& & FCSAR-A & 0.15 & 0.12 & 0.11 && 0.14 & 0.11 & 0.10 && 0.15 & 0.13 & 0.10 \\
& & & (0.10) & (0.06) & (0.06) && (0.08) & (0.05) & (0.05) && (0.08) & (0.07) & (0.06) \\

\cmidrule(lr){2-14}
& \multirow[t]{8}{*}{MSPE}
& FPCA & 4.77 & 4.21 & 4.23 && 8.46 & 6.52 & 6.46 && 19.41 & 17.55 & 17.87 \\
& & & (1.71) & (0.81) & (0.70) && (11.75) & (2.73) & (1.02) && (7.77) & (4.96) & (3.98) \\
& & RFPCA & 1.25 & 1.14 & 1.12 && 1.50 & 1.30 & 1.17 && 12.62 & 8.74 & 10.27 \\
& & & (0.31) & (0.16) & (0.10) && (0.63) & (0.40) & (0.17) && (11.39) & (7.82) & (6.99) \\
& & FCSAR-D\textsuperscript{*} & 1.13 & 1.07 & 1.05 && 1.54 & 1.16 & 1.09 && 4.38 & 2.16 & 1.63 \\
& & & (0.17) & (0.12) & (0.06) && (1.49) & (0.28) & (0.13) && (10.99) & (2.05) & (0.89) \\
& & FCSAR-A\textsuperscript{*} & 1.14 & 1.06 & 1.05 && 1.56 & 1.16 & 1.09 && 7.38 & 2.14 & 1.62 \\
& & & (0.18) & (0.11) & (0.06) && (1.63) & (0.28) & (0.13) && (39.83) & (1.99) & (0.79) \\
& & FCSAR-D & 1.13 & 1.06 & 1.05 && 1.43 & 1.16 & 1.09 && 3.04 & 1.77 & 1.46 \\
& & & (0.17) & (0.11) & (0.06) && (1.02) & (0.28) & (0.14) && (4.28) & (1.03) & (0.50) \\
& & FCSAR-A & 1.13 & 1.06 & 1.05 && 1.43 & 1.16 & 1.09 && 3.00 & 1.84 & 1.50 \\
& & & (0.18) & (0.11) & (0.06) && (0.90) & (0.28) & (0.14) && (3.15) & (1.13) & (0.53) \\
\bottomrule
\end{longtable}
\end{small}
\end{center}

\vspace{-.2in}

A more detailed inspection of Table~\ref{tab:tab2_5} reveals two further patterns. First, when $\mu_{\widetilde{\varepsilon}}=10$, the two proposed Fisher-consistent redescending estimators behave quite similarly, and the same is true for their non-Fisher-consistent counterparts. In this moderate contamination regime, the differences between Danish and Andrews losses are not dramatic, although FCSAR-D is sometimes marginally better for $\mathrm{RMSE}(\widehat{\vartheta})$ or $\mathrm{MSPE}$, while FCSAR-A can show a slight advantage for $\mathrm{RMSE}(\widehat{\sigma})$ in some configurations. More importantly, both FCSAR-D and FCSAR-A generally improve upon FCSAR-D\textsuperscript{*} and FCSAR-A\textsuperscript{*}, especially when spatial dependence is stronger. Second, when $\mu_{\widetilde{\varepsilon}}=20$, the separation becomes much clearer. In this more demanding setting, the FPCA-based estimator is highly unstable, and the RFPCA-Huber estimator, while still much better than FPCA, can also deteriorate substantially, particularly when $\vartheta$ is large. The starred redescending competitors remain considerably more stable than RFPCA-Huber, but the proposed Fisher-consistent estimators typically deliver the smallest values of $\mathrm{RMSE}(\widehat{\sigma})$, $\mathrm{IMSE}(\widehat{\beta})$, and $\mathrm{MSPE}$. Thus, even at the $5\%$ contamination level, the revised benchmark set indicates that the combination of redescending loss and Fisher-consistent spatial correction yields the best overall balance between robustness in estimation and robustness in prediction.

Table~\ref{tab:tab2_10} presents the most demanding scenario, namely $10\%$ contamination. Here, the superiority of the redescending procedures becomes even more visible. The classical FPCA-based likelihood estimator is severely affected across all metrics, and the deterioration becomes dramatic when both the contamination magnitude and the spatial dependence level are high. The RFPCA-Huber estimator still improves substantially over FPCA, but its performance can become unstable in several configurations, particularly for large $\vartheta$ and large contamination magnitude. The benchmarks FCSAR-D\textsuperscript{*} and FCSAR-A\textsuperscript{*} show that replacing the Huber step by a non-Fisher-consistent redescending estimator already yields a major gain. In most settings, these two procedures are markedly more stable than RFPCA-Huber, especially for $\mathrm{RMSE}(\widehat{\sigma})$, $\mathrm{IMSE}(\widehat{\beta})$, and $\mathrm{MSPE}$. However, the proposed Fisher-consistent procedures still improve further on these starred counterparts. This gain is most evident in scale estimation and prediction for the more difficult configurations with strong spatial dependence and larger contamination magnitudes.

\begin{center}
\tabcolsep 0.07in
\renewcommand{\arraystretch}{0.9}
\begin{small}
\begin{longtable}{@{}lllccccccccccc@{}}
\caption{\small Monte Carlo mean values of $\mathrm{RMSE}(\widehat{\vartheta})$, $\mathrm{RMSE}(\widehat{\sigma})$, $\mathrm{IMSE}(\widehat{\beta})$, and $\mathrm{MSPE}$ under the $10\%$ contamination setting. Results are based on $500$ replications for $n\in\{100,250,500\}$ and $\vartheta\in\{0.25,0.50,0.75\}$, with contaminated response errors generated from $\widetilde{\varepsilon}\sim N(\mu_{\widetilde{\varepsilon}},1)$ for $\mu_{\widetilde{\varepsilon}}\in\{10,20\}$. Standard errors are reported in parentheses for $\mathrm{IMSE}(\widehat{\beta})$ and $\mathrm{MSPE}$. Here, FCSAR-D and FCSAR-A denote the proposed FCSAR estimators based on Danish and Andrews losses, respectively.}
\label{tab:tab2_10} \\
\toprule
& & & \multicolumn{3}{c}{$\vartheta=0.25$}
& & \multicolumn{3}{c}{$\vartheta=0.50$}
& & \multicolumn{3}{c}{$\vartheta=0.75$} \\
\cmidrule(lr){4-6} \cmidrule(lr){8-10} \cmidrule(lr){12-14}
$\mu_{\widetilde{\varepsilon}}$ & Metric & Method
& $n=100$ & $250$ & $500$
& & $100$ & $250$ & $500$
& & $100$ & $250$ & $500$ \\
\midrule
\endfirsthead
\toprule
& & & \multicolumn{3}{c}{$\vartheta=0.25$}
& & \multicolumn{3}{c}{$\vartheta=0.50$}
& & \multicolumn{3}{c}{$\vartheta=0.75$} \\
\cmidrule(lr){4-6} \cmidrule(lr){8-10} \cmidrule(lr){12-14}
$\mu_{\widetilde{\varepsilon}}$ & Metric & Method
& $n=100$ & $250$ & $500$
& & $100$ & $250$ & $500$
& & $100$ & $250$ & $500$ \\
\midrule
\endhead
\midrule
\multicolumn{14}{r}{Continued on next page} \\
\endfoot
\endlastfoot
\multirow[t]{24}{*}{$10$}
& \multirow[t]{4}{*}{RMSE($\widehat{\vartheta}$)}
& FPCA & 0.27 & 0.18 & 0.15 && 0.27 & 0.16 & 0.13 && 0.24 & 0.14 & 0.10 \\
& & RFPCA & 0.14 & 0.09 & 0.08 && 0.14 & 0.07 & 0.08 && 0.36 & 0.46 & 0.53 \\
& & FCSAR-D\textsuperscript{*} & 0.11 & 0.06 & 0.06 && 0.10 & 0.06 & 0.05 && 0.09 & 0.07 & 0.05 \\
& & FCSAR-A\textsuperscript{*} & 0.11 & 0.06 & 0.06 && 0.10 & 0.06 & 0.05 && 0.09 & 0.07 & 0.05 \\
& & FCSAR-D & 0.11 & 0.07 & 0.06 && 0.11 & 0.06 & 0.06 && 0.10 & 0.07 & 0.05 \\
& & FCSAR-A & 0.12 & 0.07 & 0.06 && 0.11 & 0.06 & 0.06 && 0.09 & 0.07 & 0.05 \\

\cmidrule(lr){2-14}
& \multirow[t]{4}{*}{RMSE($\widehat{\sigma}$)}
& FPCA & 1.87 & 1.88 & 1.89 && 1.90 & 1.88 & 1.88 && 1.90 & 1.87 & 1.88 \\
& & RFPCA & 0.73 & 0.74 & 0.72 && 0.78 & 0.71 & 0.71 && 1.03 & 1.04 & 0.98 \\
& & FCSAR-D\textsuperscript{*} & 0.26 & 0.17 & 0.15 && 0.27 & 0.15 & 0.15 && 0.16 & 0.13 & 0.13 \\
& & FCSAR-A\textsuperscript{*} & 0.19 & 0.17 & 0.15 && 0.21 & 0.15 & 0.15 && 0.16 & 0.13 & 0.13 \\
& & FCSAR-D & 0.22 & 0.06 & 0.04 && 0.23 & 0.05 & 0.04 && 0.32 & 0.05 & 0.04 \\
& & FCSAR-A & 0.17 & 0.15 & 0.14 && 0.17 & 0.15 & 0.14 && 0.19 & 0.14 & 0.14 \\

\cmidrule(lr){2-14}
& \multirow[t]{8}{*}{IMSE($\widehat{\beta}$)}
& FPCA & 2.69 & 2.89 & 2.80 && 2.56 & 2.95 & 2.88 && 2.78 & 3.06 & 3.01 \\
& & & (1.60) & (0.90) & (0.59) && (1.61) & (0.94) & (0.55) && (1.53) & (0.95) & (0.75) \\
& & RFPCA & 0.65 & 0.59 & 0.45 && 0.87 & 0.46 & 0.42 && 0.96 & 0.73 & 0.56 \\
& & & (0.64) & (0.80) & (0.33) && (1.43) & (0.32) & (0.20) && (1.18) & (0.69) & (0.35) \\
& & FCSAR-D\textsuperscript{*} & 0.21 & 0.09 & 0.08 && 0.37 & 0.09 & 0.08 && 0.13 & 0.09 & 0.08 \\
& & & (0.89) & (0.05) & (0.05) && (1.69) & (0.05) & (0.04) && (0.08) & (0.06) & (0.03) \\
& & FCSAR-A\textsuperscript{*} & 0.13 & 0.09 & 0.08 && 0.12 & 0.09 & 0.08 && 0.13 & 0.09 & 0.08 \\
& & & (0.08) & (0.05) & (0.05) && (0.08) & (0.05) & (0.04) && (0.08) & (0.06) & (0.03) \\
& & FCSAR-D & 0.24 & 0.09 & 0.08 && 0.39 & 0.09 & 0.08 && 0.46 & 0.09 & 0.08 \\
& & & (0.83) & (0.05) & (0.05) && (1.83) & (0.05) & (0.04) && (1.59) & (0.06) & (0.03) \\
& & FCSAR-A & 0.13 & 0.09 & 0.08 && 0.13 & 0.09 & 0.08 && 0.13 & 0.09 & 0.08  \\
& & & (0.09) & (0.05) & (0.05) && (0.08) & (0.05) & (0.04) && (0.08) & (0.06) & (0.03) \\

\cmidrule(lr){2-14}
& \multirow[t]{8}{*}{MSPE}
& FPCA & 4.35 & 4.02 & 3.86 && 6.34 & 5.53 & 5.37 && 18.84 & 14.75 & 14.81 \\
& & & (1.04) & (0.59) & (0.48) && (1.99) & (1.22) & (0.91) && (9.58) & (4.89) & (3.54) \\
& & RFPCA & 1.68 & 1.47 & 1.37 && 2.16 & 1.66 & 1.52 && 11.20 & 8.54 & 8.95 \\
& & & (0.48) & (0.28) & (0.20) && (0.98) & (0.41) & (0.37) && (13.29) & (8.73) & (9.47) \\
& & FCSAR-D\textsuperscript{*} & 1.22 & 1.06 & 1.06 && 1.71 & 1.21 & 1.15 && 6.99 & 3.25 & 2.05 \\
& & & (0.39) & (0.10) & (0.09) && (1.56) & (0.27) & (0.17) && (13.67) & (3.86) & (1.50) \\
& & FCSAR-A\textsuperscript{*} & 1.20 & 1.06 & 1.06 && 1.59 & 1.21 & 1.14 && 7.33 & 3.24 & 2.10 \\
& & & (0.33) & (0.10) & (0.09) && (1.16) & (0.26) & (0.16) && (16.12) & (3.74) & (1.69) \\
& & FCSAR-D & 1.24 & 1.06 & 1.05 && 1.56 & 1.18 & 1.13 && 3.91 & 2.21 & 1.80 \\
& & & (0.48) & (0.10) & (0.08) && (0.99) & (0.22) & (0.13) && (4.84) & (1.89) & (1.17) \\
& & FCSAR-A & 1.19 & 1.06 & 1.06 && 1.49 & 1.19 & 1.12 && 3.57 & 2.18 & 1.85 \\
& & & (0.29) & (0.10) & (0.08) && (0.68) & (0.21) & (0.13) && (4.96) & (1.82) & (1.24) \\

\midrule

\multirow[t]{24}{*}{$20$}
& \multirow[t]{4}{*}{RMSE($\widehat{\vartheta}$)}
& FPCA & 0.35 & 0.20 & 0.20 && 0.31 & 0.17 & 0.14 && 0.27 & 0.21 & 0.13 \\
& & RFPCA & 0.14 & 0.15 & 0.13 && 0.18 & 0.11 & 0.09 && 0.79 & 0.78 & 0.83 \\
& & FCSAR-D\textsuperscript{*} & 0.08 & 0.05 & 0.05 && 0.07 & 0.04 & 0.04 && 0.05 & 0.05 & 0.04 \\
& & FCSAR-A\textsuperscript{*} & 0.08 & 0.05 & 0.05 && 0.07 & 0.04 & 0.04 && 0.05 & 0.05 & 0.04 \\
& & FCSAR-D & 0.08 & 0.06 & 0.05 && 0.08 & 0.05 & 0.04 && 0.07 & 0.05 & 0.03 \\
& & FCSAR-A & 0.08 & 0.06 & 0.05 && 0.08 & 0.05 & 0.04 && 0.06 & 0.05 & 0.03 \\

\cmidrule(lr){2-14}
& \multirow[t]{4}{*}{RMSE($\widehat{\sigma}$)}
& FPCA & 4.38 & 4.41 & 4.45 && 4.45 & 4.45 & 4.43 && 4.42 & 4.42 & 4.44 \\
& & RFPCA & 0.78 & 0.77 & 0.72 && 0.83 & 0.78 & 0.73 && 2.44 & 2.07 & 2.00 \\
& & FCSAR-D\textsuperscript{*} & 0.19 & 0.17 & 0.15 && 0.18 & 0.17 & 0.15 && 0.15 & 0.13 & 0.13 \\
& & FCSAR-A\textsuperscript{*} & 0.18 & 0.17 & 0.15 && 0.17 & 0.16 & 0.15 && 0.15 & 0.13 & 0.13 \\
& & FCSAR-D & 0.09 & 0.06 & 0.04 && 0.09 & 0.06 & 0.04 && 0.42 & 0.06 & 0.04 \\
& & FCSAR-A & 0.18 & 0.15 & 0.14 && 0.18 & 0.15 & 0.13 && 0.19 & 0.15 & 0.13 \\

\cmidrule(lr){2-14}
& \multirow[t]{8}{*}{IMSE($\widehat{\beta}$)}
& FPCA & 11.42 & 11.46 & 10.97 && 10.39 & 11.13 & 11.18 && 11.43 & 11.66 & 11.77 \\
& & & (6.45) & (3.80) & (2.21) && (5.70) & (3.57) & (1.84) && (5.45) & (3.82) & (2.61) \\
& & RFPCA & 0.75 & 0.55 & 0.44 && 1.01 & 0.75 & 0.41 && 3.35 & 1.26 & 1.31 \\
& & & (0.74) & (0.49) & (0.33) && (1.94) & (3.07) & (0.17) && (6.78) & (1.31) & (1.27) \\
& & FCSAR-D\textsuperscript{*} & 0.12 & 0.09 & 0.08 && 0.12 & 0.08 & 0.07 && 0.12 & 0.10 & 0.08 \\
& & & (0.07) & (0.05) & (0.04) && (0.07) & (0.05) & (0.04) && (0.07) & (0.06) & (0.03) \\
& & FCSAR-A\textsuperscript{*} & 0.12 & 0.09 & 0.08 && 0.12 & 0.09 & 0.07 && 0.12 & 0.10 & 0.08 \\
& & & (0.07) & (0.05) & (0.04) && (0.07) & (0.05) & (0.04) && (0.08) & (0.06) & (0.03) \\
& & FCSAR-D & 0.12 & 0.09 & 0.08 && 0.12 & 0.09 & 0.07 && 0.49 & 0.10 & 0.08 \\
& & & (0.07) & (0.05) & (0.04) && (0.08) & (0.05) & (0.04) && (3.66) & (0.06) & (0.03) \\
& & FCSAR-A & 0.12 & 0.09 & 0.08 && 0.12 & 0.09 & 0.07 && 0.13 & 0.10 & 0.08 \\
& & & (0.08) & (0.05) & (0.04) && (0.08) & (0.05) & (0.04) && (0.08) & (0.06) & (0.03) \\

\cmidrule(lr){2-14}
& \multirow[t]{8}{*}{MSPE}
& FPCA & 13.23 & 12.51 & 11.87 && 20.45 & 18.70 & 18.28 && 64.54 & 55.36 & 53.55 \\
& & & (3.13) & (2.15) & (1.64) && (5.30) & (3.55) & (2.63) && (24.46) & (16.66) & (9.99) \\
& & RFPCA & 1.74 & 1.54 & 1.46 && 2.93 & 2.16 & 1.82 && 50.82 & 47.66 & 49.41 \\
& & & (0.56) & (0.33) & (0.33) && (1.64) & (1.59) & (0.82) && (30.87) & (20.48) & (16.10) \\
& & FCSAR-D\textsuperscript{*} & 1.21 & 1.09 & 1.07 && 1.83 & 1.22 & 1.16 && 7.00 & 3.77 & 2.23 \\
& & & (0.27) & (0.11) & (0.09) && (2.06) & (0.35) & (0.17) && (11.29) & (5.16) & (1.61) \\
& & FCSAR-A\textsuperscript{*} & 1.21 & 1.08 & 1.07 && 1.78 & 1.23 & 1.16 && 6.66 & 3.74 & 2.28 \\
& & & (0.27) & (0.12) & (0.09) && (1.65) & (0.36) & (0.18) && (9.03) & (4.50) & (1.52) \\
& & FCSAR-D & 1.20 & 1.09 & 1.07 && 1.68 & 1.21 & 1.16 && 5.84 & 3.21 & 2.24 \\
& & & (0.26) & (0.12) & (0.09) && (1.39) & (0.30) & (0.17) && (11.37) & (3.46) & (1.70) \\
& & FCSAR-A & 1.21 & 1.08 & 1.07 && 1.65 & 1.22 & 1.17 && 4.98 & 3.38 & 2.18 \\
& & & (0.25) & (0.12) & (0.09) && (1.21) & (0.31) & (0.17) && (5.72) & (3.59) & (1.50) \\
\bottomrule
\end{longtable}
\end{small}
\end{center}

\vspace{-.2in}

When the contamination level is $10\%$ and $\mu_{\widetilde{\varepsilon}}=10$, FCSAR-A generally emerges as the best overall procedure, with a particularly clear advantage in $\mathrm{RMSE}(\widehat{\sigma})$ and often in $\mathrm{IMSE}(\widehat{\beta})$ and $\mathrm{MSPE}$. When $\mu_{\widetilde{\varepsilon}}=20$, the strongest evidence in favor of the proposed methodology appears. In this setting, which combines a nontrivial contamination proportion, large vertical outliers, and functional leverage points, the proposed FCSAR estimators remain stable and accurate, whereas the benchmark methods degrade sharply. Among the two proposed methods, FCSAR-A is often the most reliable overall in the most severe configurations, while FCSAR-D remains clearly superior to the benchmark methods and becomes highly competitive as $n$ grows. This distinction is consistent with the fact that Andrews' loss completely rejects sufficiently extreme residuals, whereas Danish loss remains softer and can retain a small residual influence from very large outliers.

The graphical summaries in Figures~\ref{fig:Fig_1} and~\ref{fig:Fig_2} reinforce the conclusions drawn from Tables~\ref{tab:tab1},~\ref{tab:tab2_5}, and~\ref{tab:tab2_10}. Figure~\ref{fig:Fig_1} displays the true coefficient function together with the Monte Carlo mean of the estimated coefficient functions across different contamination settings and spatial dependence levels for $n=250$. In the uncontaminated case, the mean estimated coefficient functions produced by all methods track the true curve reasonably well, which is consistent with the numerical results showing only modest differences among the competitors under the reference model. Once contamination is introduced, however, the visual gap between the true coefficient function and the estimated mean functions becomes increasingly severe for the nonrobust FPCA estimator and remains visible for the RFPCA-Huber procedure. The non-Fisher-consistent redescending competitors FCSAR-D\textsuperscript{*} and FCSAR-A\textsuperscript{*} already substantially reduce this distortion, but the proposed Fisher-consistent procedures remain closest to the true coefficient function overall, especially under stronger contamination and stronger spatial dependence.

\begin{figure}[!htbp]
\centering
\includegraphics[width=6cm]{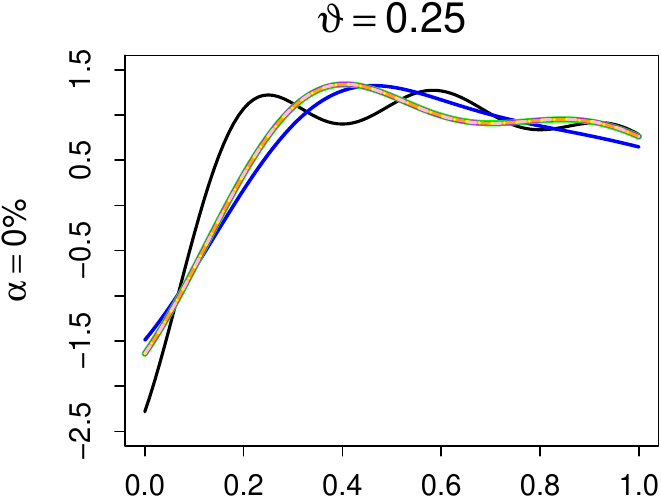}
\quad
\includegraphics[width=5.475cm]{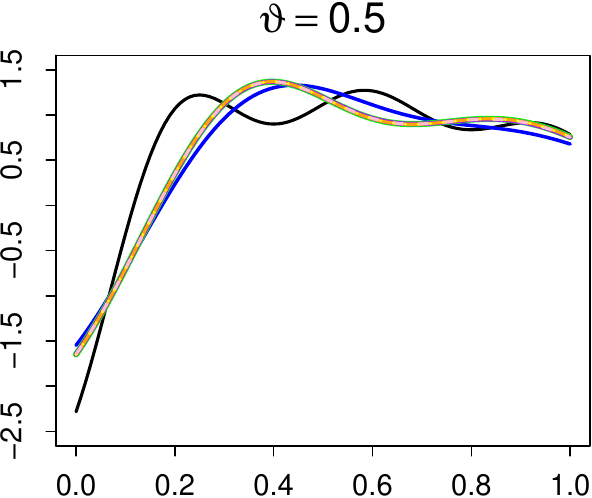}
\quad
\includegraphics[width=5.475cm]{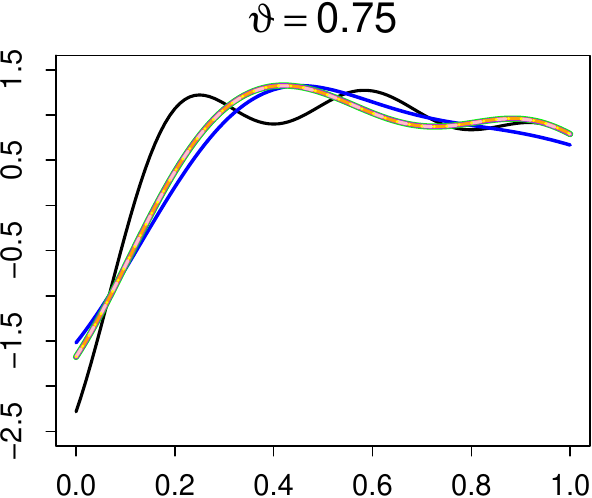}
\\
\includegraphics[width=6cm]{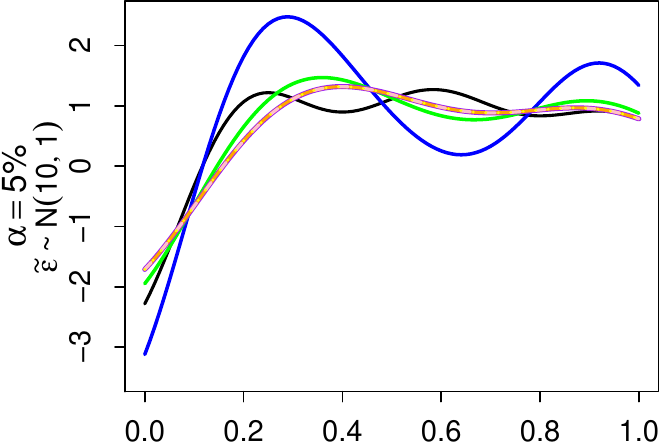}
\quad
\includegraphics[width=5.475cm]{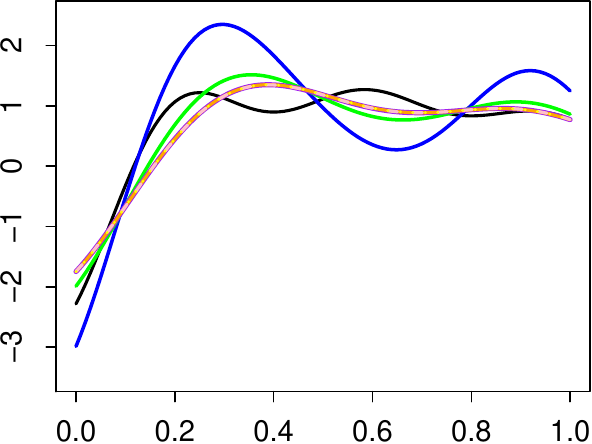}
\quad
\includegraphics[width=5.475cm]{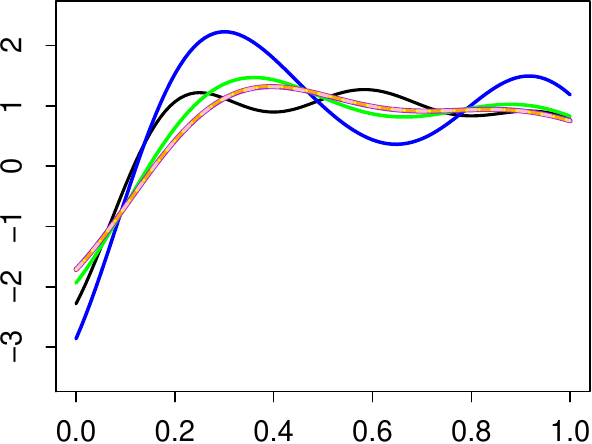}
\\
\includegraphics[width=6cm]{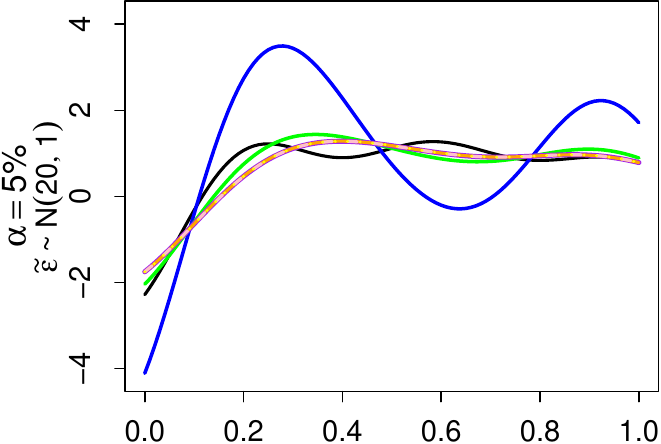}
\quad
\includegraphics[width=5.475cm]{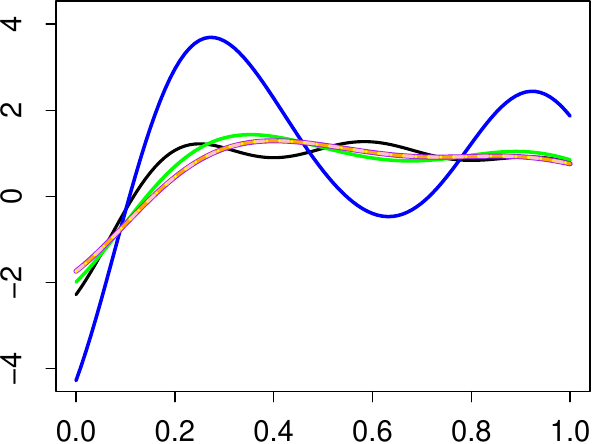}
\quad
\includegraphics[width=5.475cm]{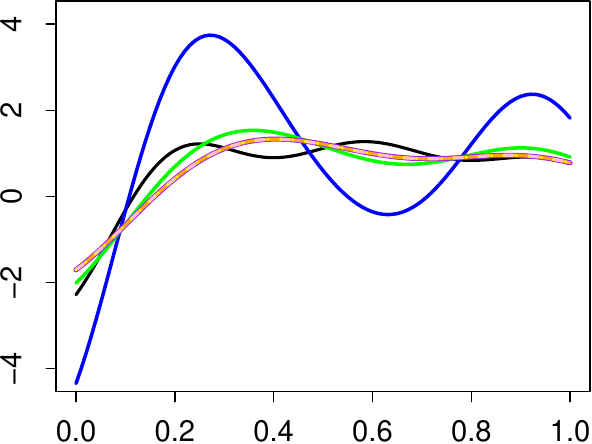}
\\
\includegraphics[width=6cm]{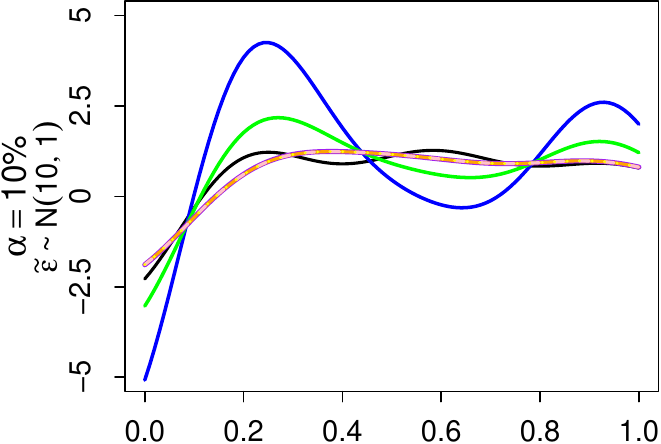}
\quad
\includegraphics[width=5.475cm]{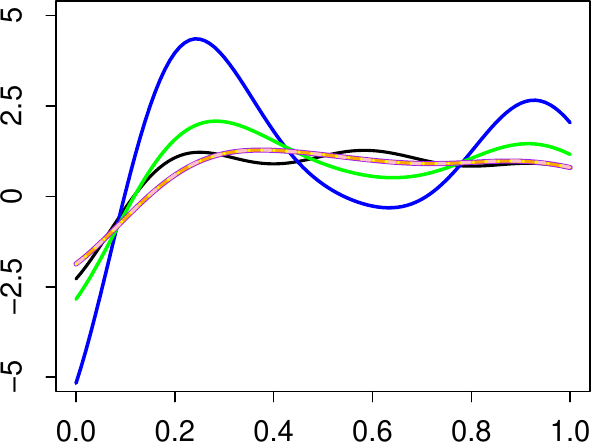}
\quad
\includegraphics[width=5.475cm]{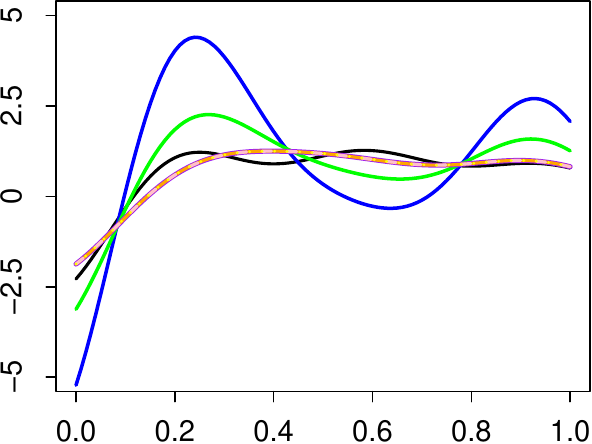}
\\
\includegraphics[width=6cm]{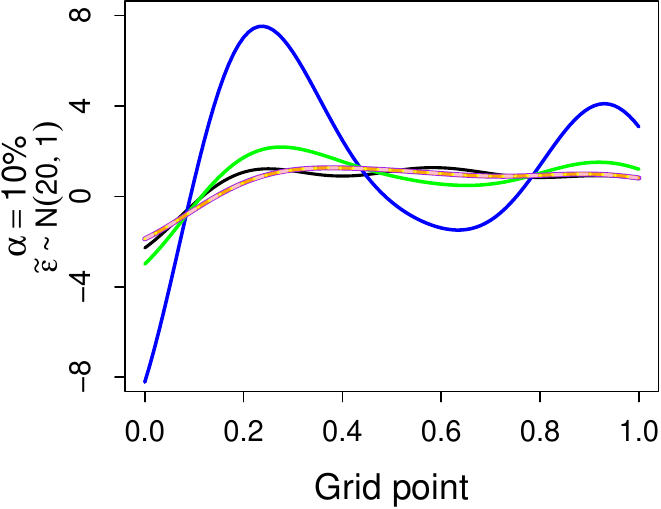}
\quad
\includegraphics[width=5.475cm]{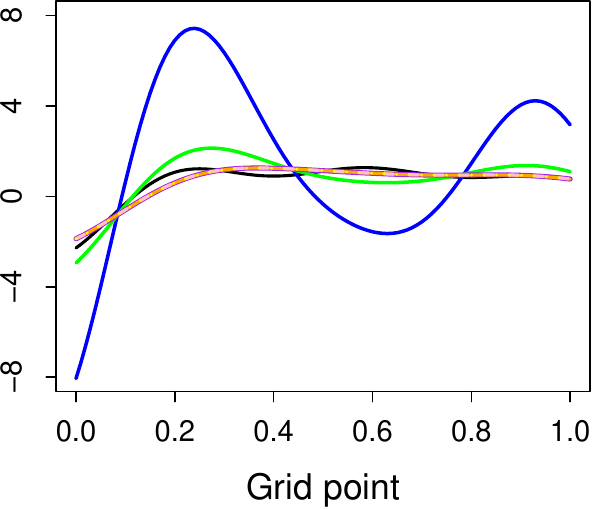}
\quad
\includegraphics[width=5.475cm]{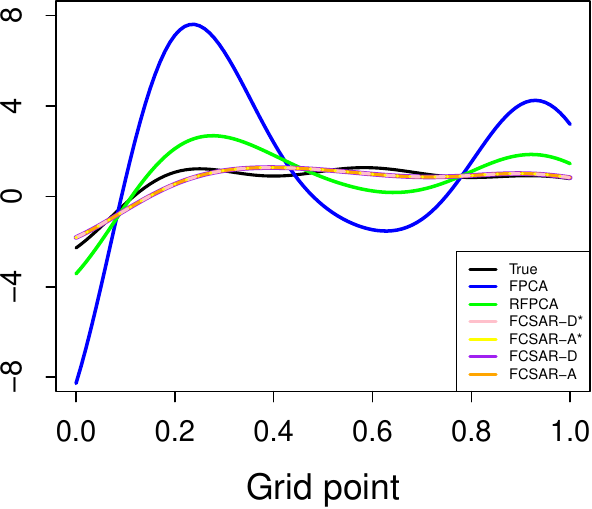}
\caption{\small{Plots of the generated true regression coefficient functions and mean functions of the estimated coefficient functions. The columns represent weak ($\vartheta=0.25$), moderate ($\vartheta=0.5$), and strong ($\vartheta=0.75$) spatial effects, respectively. The first row corresponds to the uncontaminated case; the second and third rows are generated under $5\%$ contamination ($\mu_{\widetilde{\varepsilon}}=10$ and $\mu_{\widetilde{\varepsilon}}=20$, respectively), and the fourth and fifth rows under $10\%$ contamination ($\mu_{\widetilde{\varepsilon}}=10$ and $\mu_{\widetilde{\varepsilon}}=20$, respectively).
The figure is presented for $n=250$. The dashed line is included to improve visibility where curves overlap.} \label{fig:Fig_1}}
\end{figure}

Figure~\ref{fig:Fig_2} provides a more detailed view by displaying the individual estimated coefficient functions together with their Monte Carlo mean curves for the representative setting $n=250$ and $\vartheta=0.5$. In the uncontaminated case, the estimated curves from all procedures are relatively concentrated around the true coefficient function. Under contamination, the FPCA-based estimator shows the largest dispersion and the clearest systematic departure of its mean curve from the target, indicating both instability and bias. The RFPCA-Huber estimator improves substantially upon FPCA, while the starred redescending competitors show a further reduction in dispersion. The proposed FCSAR-D and FCSAR-A estimators produce the tightest clustering around the true coefficient function and the closest mean curves across the contaminated settings. The visual comparison also suggests that the Andrews-based version provides the strongest protection in the most severe contamination scenarios, whereas the Danish-based version remains highly competitive and often behaves similarly when contamination is moderate.

\begin{figure}[!htbp]
\centering
\includegraphics[width=3.1cm]{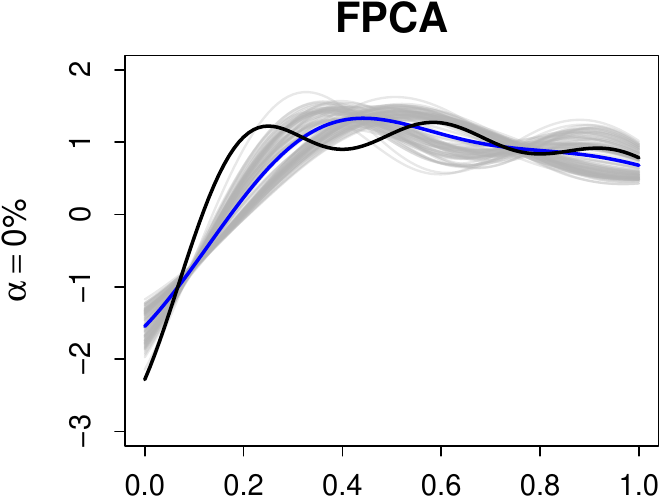}
\includegraphics[width=2.8cm]{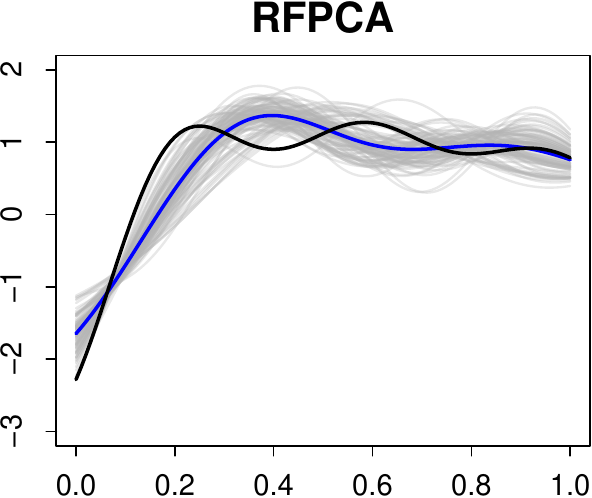}
\includegraphics[width=2.8cm]{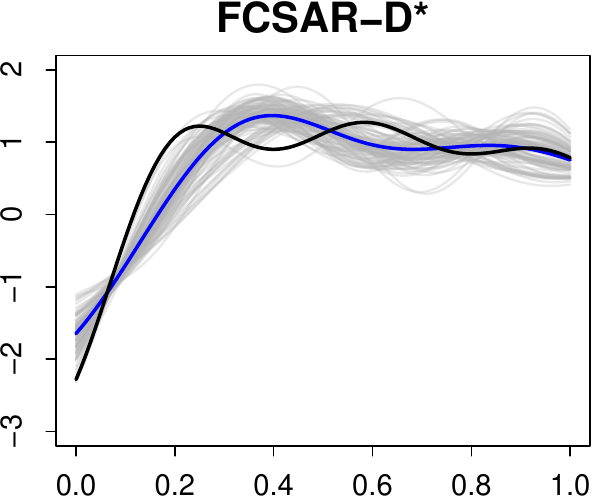}
\includegraphics[width=2.8cm]{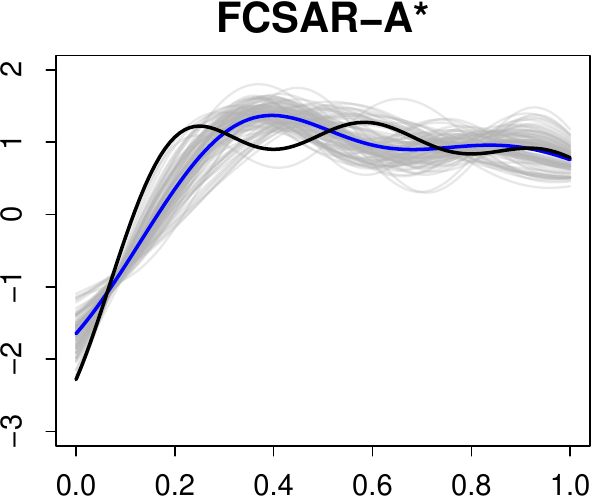}
\includegraphics[width=2.8cm]{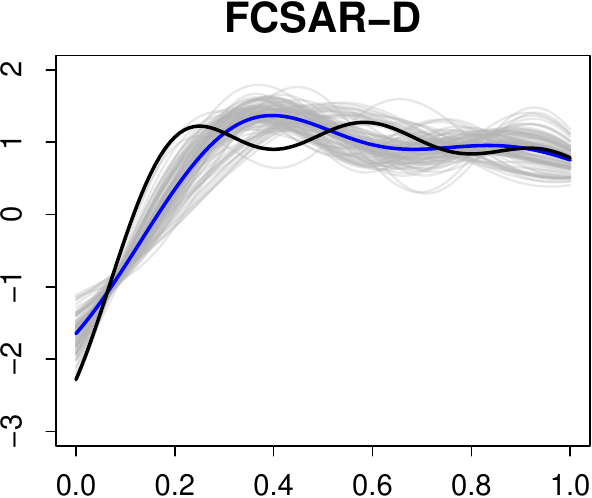}
\includegraphics[width=2.8cm]{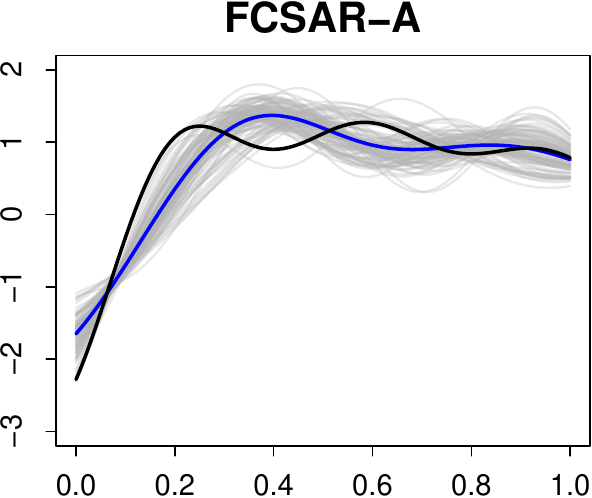}
\\
\centering
\includegraphics[width=3.1cm]{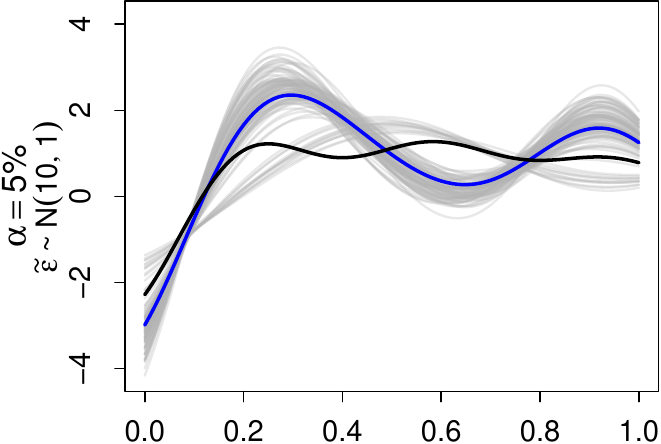}
\includegraphics[width=2.8cm]{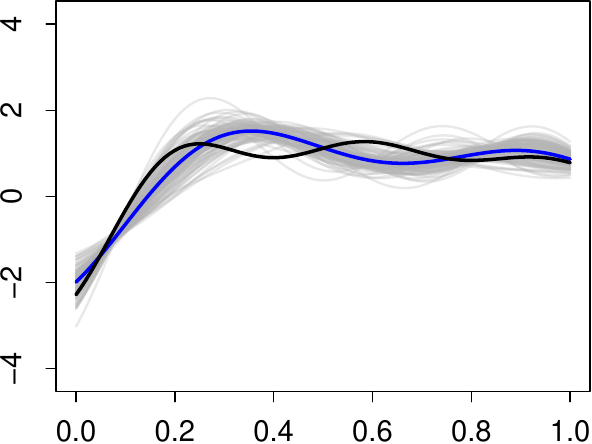}
\includegraphics[width=2.8cm]{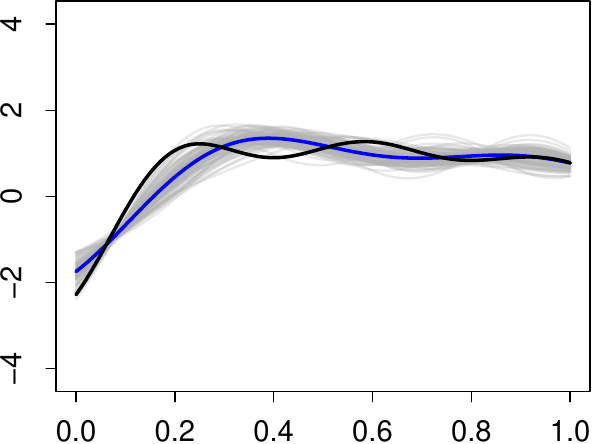}
\includegraphics[width=2.8cm]{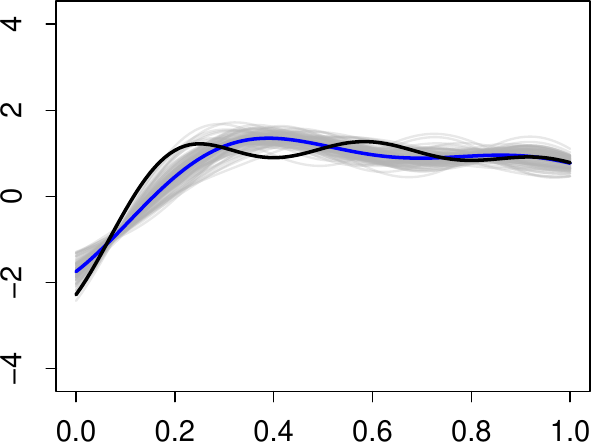}
\includegraphics[width=2.8cm]{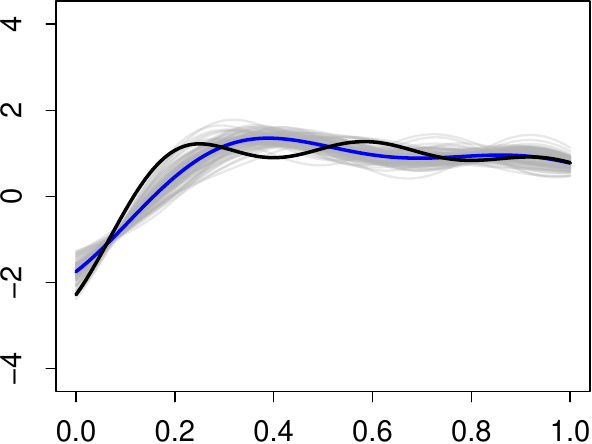}
\includegraphics[width=2.8cm]{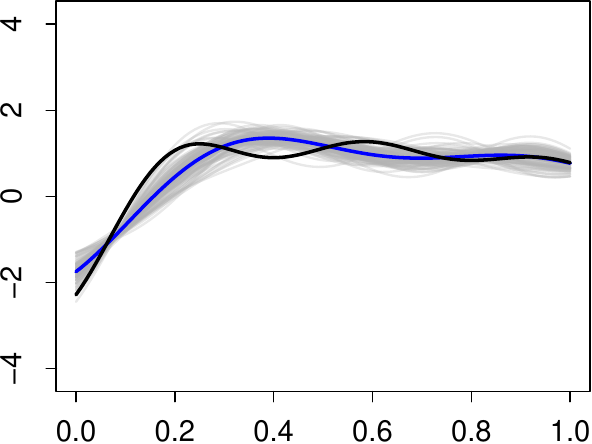}
\\
\centering
\includegraphics[width=3.1cm]{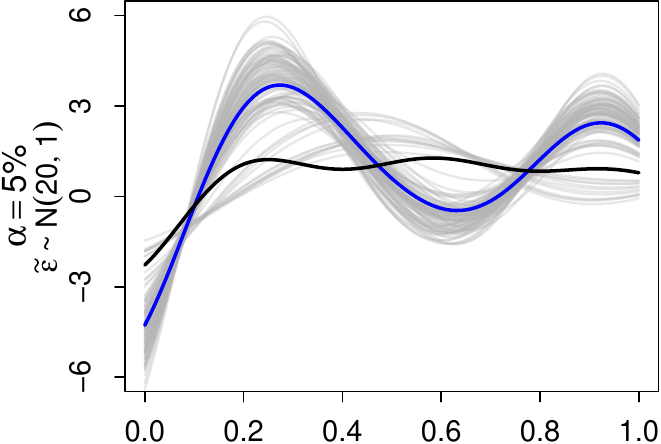}
\includegraphics[width=2.8cm]{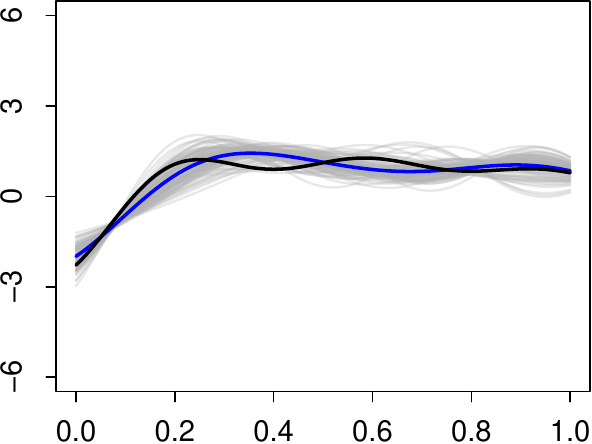}
\includegraphics[width=2.8cm]{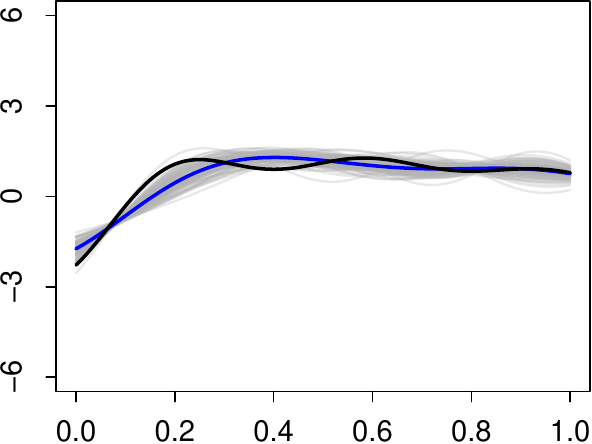}
\includegraphics[width=2.8cm]{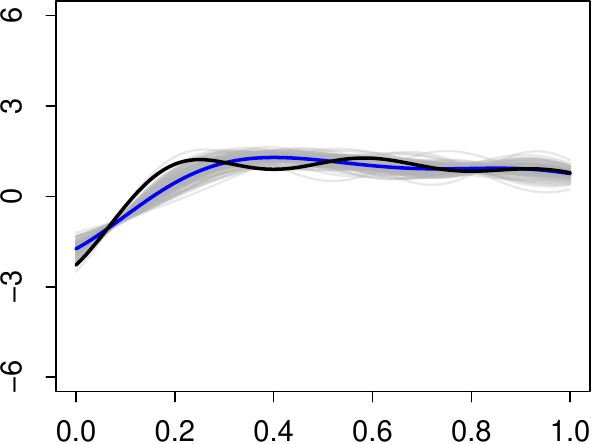}
\includegraphics[width=2.8cm]{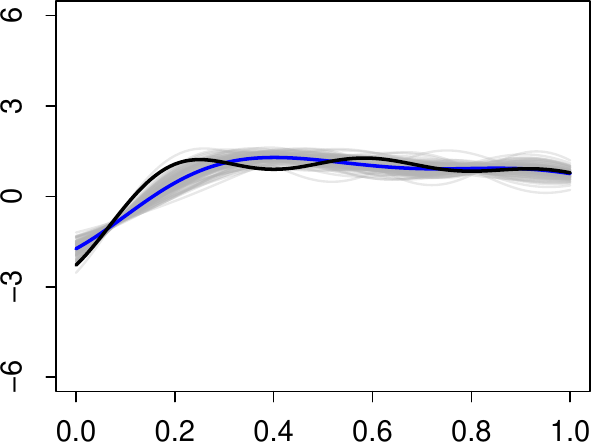}
\includegraphics[width=2.8cm]{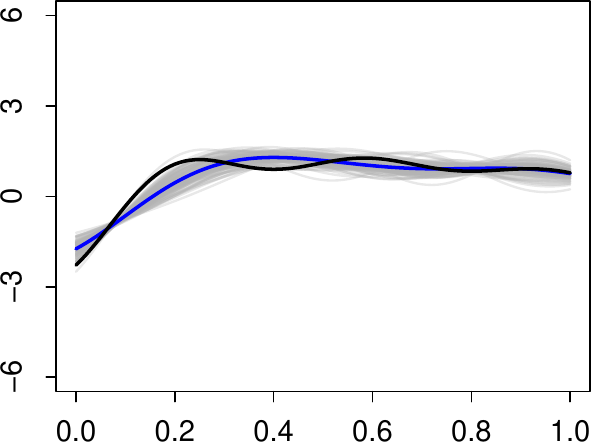}
\\
\centering
\includegraphics[width=3.1cm]{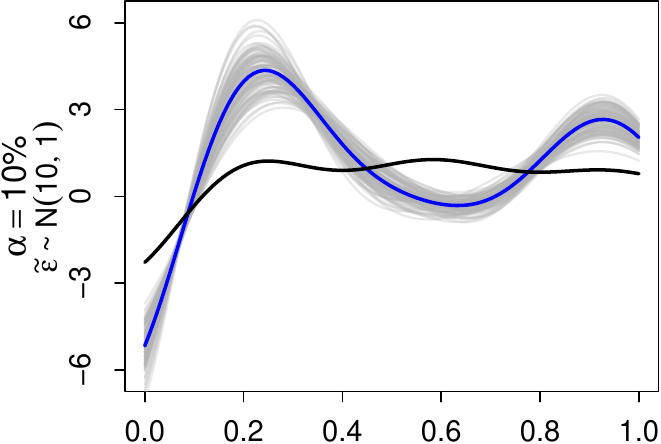}
\includegraphics[width=2.8cm]{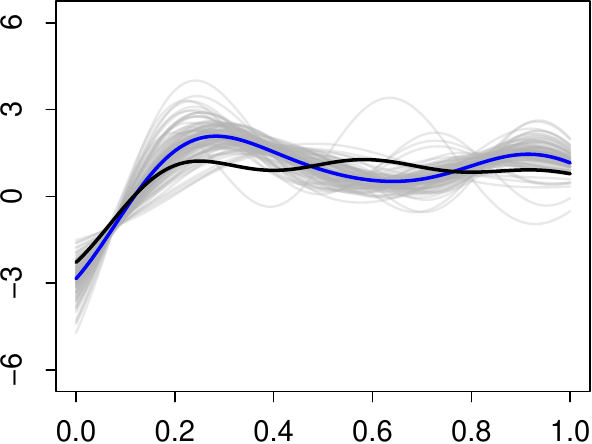}
\includegraphics[width=2.8cm]{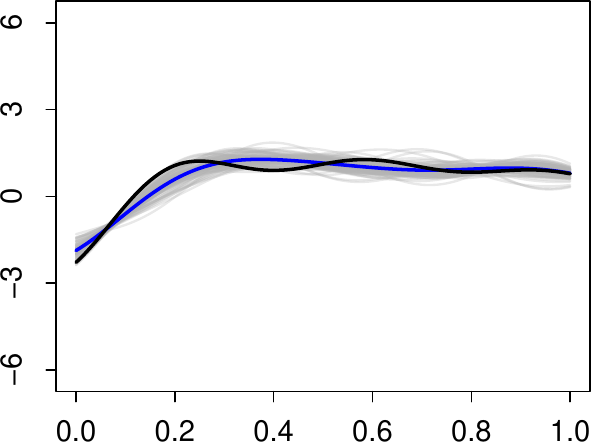}
\includegraphics[width=2.8cm]{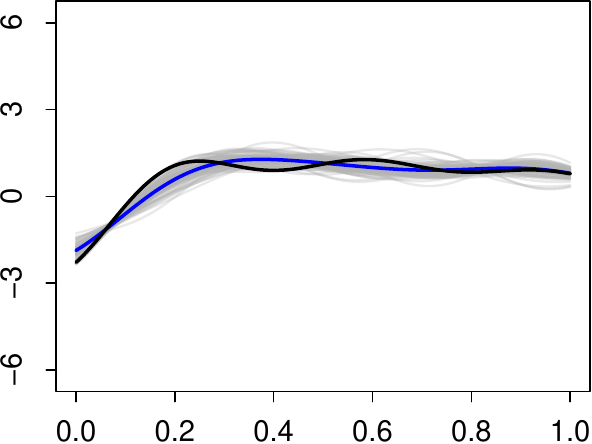}
\includegraphics[width=2.8cm]{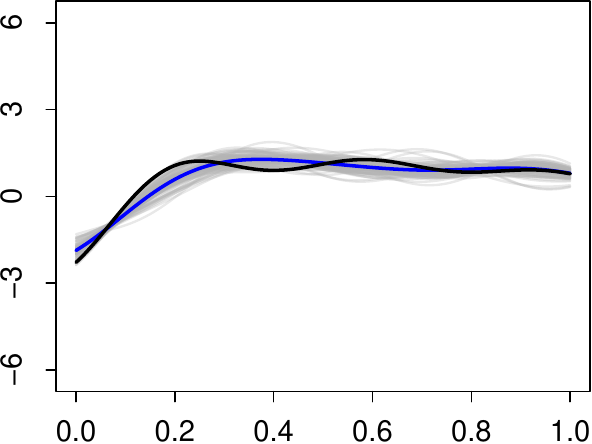}
\includegraphics[width=2.8cm]{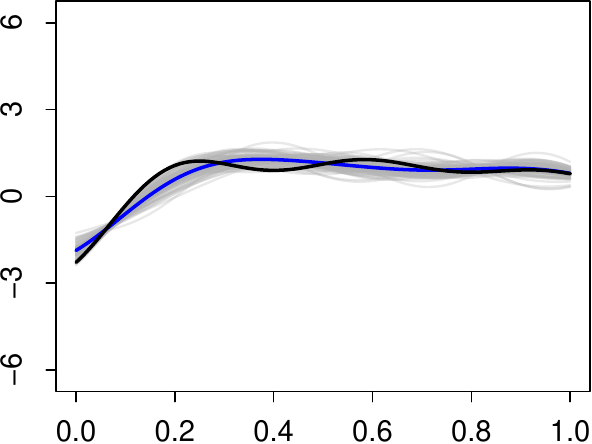}
\\
\centering
\includegraphics[width=3.1cm]{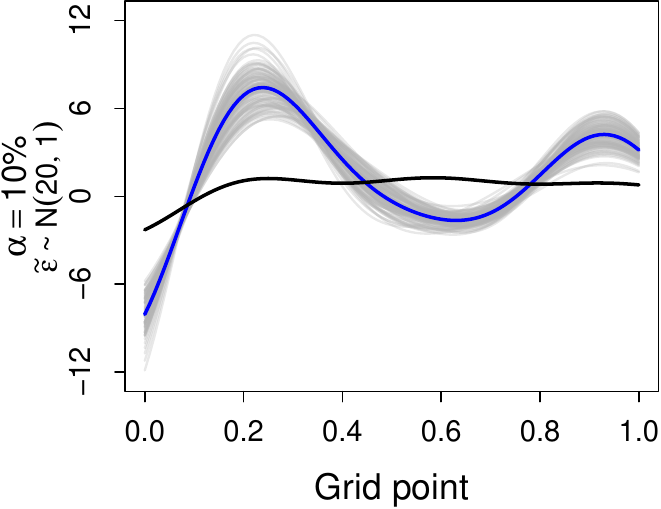}
\includegraphics[width=2.8cm]{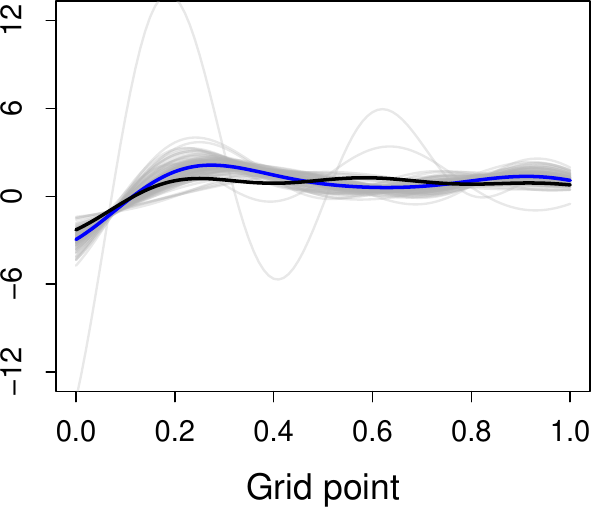}
\includegraphics[width=2.8cm]{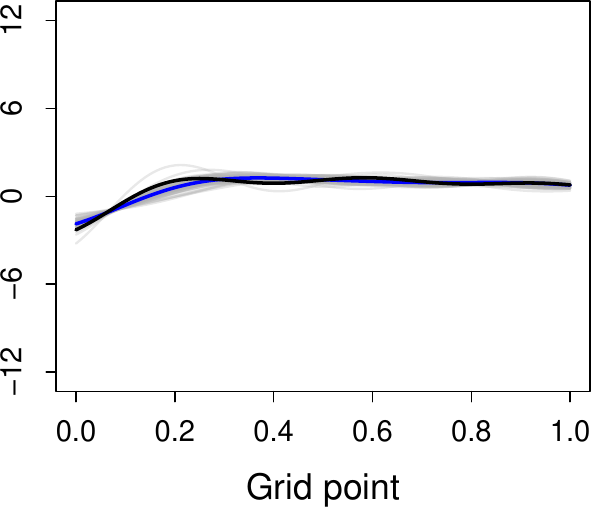}
\includegraphics[width=2.8cm]{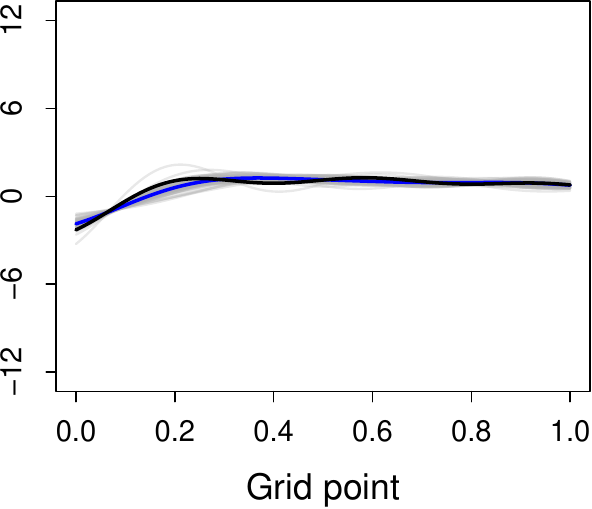}
\includegraphics[width=2.8cm]{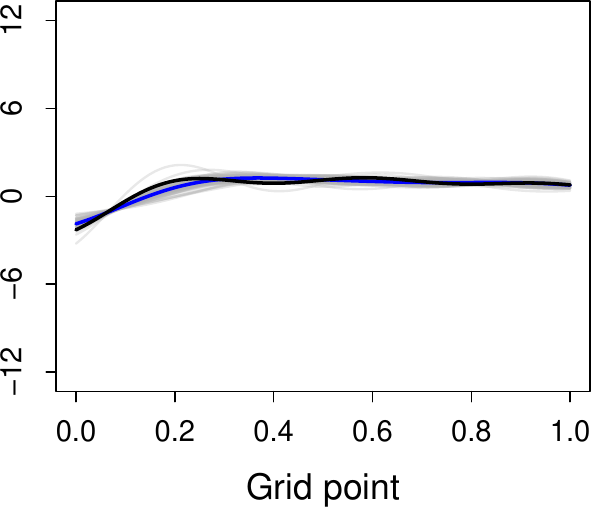}
\includegraphics[width=2.8cm]{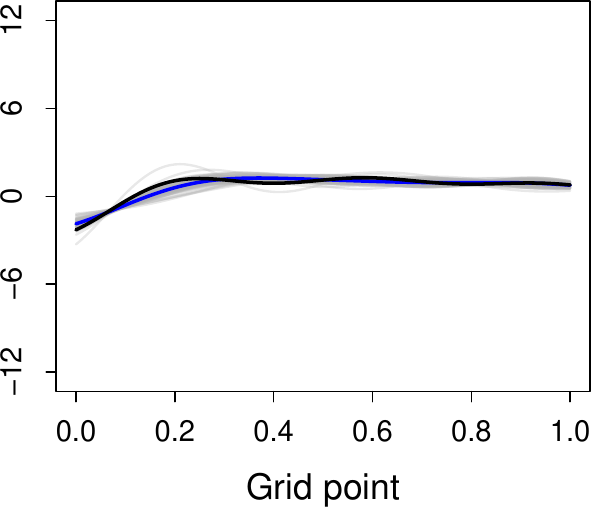}
\caption{\small{Plots of the generated true regression coefficient functions (black lines), estimated coefficient functions (gray lines), and mean functions of the estimated coefficient functions (blue lines). FPCA (first column), RFPCA (second column), FCSAR-D\textsuperscript{*} (third column), FCSAR-A\textsuperscript{*} (fourth column), FCSAR-D (fifth column) and FCSAR-A (sixth column). The first row corresponds to the uncontaminated case; the second and third rows are generated under $5\%$ contamination ($\widetilde{\varepsilon}\sim N(10,1)$ and $\widetilde{\varepsilon}\sim N(20,1)$, respectively), and the fourth and fifth rows under $10\%$ contamination ($\widetilde{\varepsilon}\sim N(10,1)$ and $\widetilde{\varepsilon}\sim N(20,1)$, respectively). 
The figure is presented for $n=250$ and $\vartheta=0.5$.}}
\label{fig:Fig_2}
\end{figure}

We note that the contamination mechanism considered in this section is stylized. In particular, contamination is introduced through random replacement; the same units are contaminated in both the functional predictor and the response; and response outliers are generated by positive mean shifts. Accordingly, the simulation results should be interpreted as evidence for this joint leverage-plus-vertical contamination regime, rather than as an exhaustive assessment over all possible contamination structures such as spatial clustering, negative response outliers, or misspecified spatial weights.

\subsection{Computational diagnostics and implementation details}\label{subsec:compdiag}

Because the hybrid IRLS-Newton step is part of the methodological contribution, we also report representative computational diagnostics for the competing procedures. All computations were carried out in \Rlogo\ on a DELL workstation equipped with a 13th Gen Intel(R) Core(TM) i9-13900HX CPU (2.20 GHz) and 64.0 GB RAM. Since a full diagnostic breakdown for every simulation configuration would require substantial additional space, Table~\ref{tab:compdiag} reports representative settings covering both uncontaminated and contaminated regimes, namely $n\in\{100,500\}$, $\vartheta\in\{0.50,0.75\}$, and contamination levels $\alpha\in\{0,0.10\}$.
\begin{table}[!htb]
\centering
\caption{\small Representative computational diagnostics for selected simulation configurations. Reported values are Monte Carlo averages over 500 replications. ``Time'' denotes elapsed wall-clock time (in seconds), ``Converged'' is the percentage of replications satisfying the stopping criterion, ``Outer'' is the average number of outer iterations, ``Inner'' is the average total number of inner Newton iterations for the spatial parameter update, and $K$ is the average retained number of principal components.}\label{tab:compdiag}
\begin{small}
\setlength{\tabcolsep}{7pt}
\begin{tabular}{@{}llccccc@{}}
\toprule
Setting & Method & Time & Converged (\%) & Outer & Inner & $K$ \\
\midrule
\multirow{6}{*}{$n=100,\ \vartheta=0.50,\ \alpha=0\%$}
& FPCA              & 0.01   & 100.0 & 8.30   & 0      & 2.48 \\
& RFPCA             & 0.24   & 100.0 & 7.50   & 0      & 2.87 \\
& FCSAR-D\textsuperscript{*}     & 0.01   & 99.0  & 3.56   & 0      & 2.87 \\
& FCSAR-A\textsuperscript{*}     & 0.03   & 100.0 & 10.53  & 0      & 2.87 \\
& FCSAR-D           & 0.28   & 100.0 & 12.43  & 113.55 & 2.87 \\
& FCSAR-A           & 0.34   & 100.0 & 17.49  & 159.73 & 2.87 \\
\midrule
\multirow{6}{*}{$n=500,\ \vartheta=0.75,\ \alpha=0\%$}
& FPCA              & 0.04    & 100.0 & 7.10   & 0      & 2.45 \\
& RFPCA             & 12.29   & 100.0 & 5.21   & 0      & 2.79 \\
& FCSAR-D\textsuperscript{*}     & 1.76    & 98.8  & 5.32   & 0      & 2.79 \\
& FCSAR-A\textsuperscript{*}     & 2.74    & 100.0 & 9.73   & 0      & 2.79 \\
& FCSAR-D           & 19.56   & 100.0 & 9.51   & 86.31  & 2.79 \\
& FCSAR-A           & 29.39   & 100.0 & 12.83  & 115.92 & 2.79 \\
\midrule
\multirow{6}{*}{$n=100,\ \vartheta=0.50,\ \alpha=10\%,\ \mu_{\widetilde{\varepsilon}}=10$}
& FPCA              & 0.01   & 100.0 & 10.78  & 0      & 2.81 \\
& RFPCA             & 2.73   & 95.0  & 37.27  & 0      & 2.93 \\
& FCSAR-D\textsuperscript{*}     & 0.03   & 100.0 & 4.29   & 0      & 2.93 \\
& FCSAR-A\textsuperscript{*}     & 0.08   & 100.0 & 11.97  & 0      & 2.93 \\
& FCSAR-D           & 3.13   & 100.0 & 61.04  & 567.84 & 2.93 \\
& FCSAR-A           & 4.31   & 100.0 & 78.88  & 733.71 & 2.93 \\
\midrule
\multirow{6}{*}{$n=500,\ \vartheta=0.75,\ \alpha=10\%,\ \mu_{\widetilde{\varepsilon}}=20$}
& FPCA              & 0.06     & 100.0 & 28.00   & 0       & 3.00 \\
& RFPCA             & 306.77   & 0.0   & 100.00  & 0       & 2.98 \\
& FCSAR-D\textsuperscript{*}     & 1.99     & 100.0 & 4.74    & 0       & 2.98 \\
& FCSAR-A\textsuperscript{*}     & 3.84     & 100.0 & 11.26   & 0       & 2.98 \\
& FCSAR-D           & 657.31   & 24.0  & 239.61  & 2272.40 & 2.98 \\
& FCSAR-A           & 810.97  & 10.0  & 246.60  & 2396.06 & 2.98 \\
\bottomrule
\end{tabular}
\end{small}
\end{table}
Table~\ref{tab:compdiag} reveals a clear and consistent computational ordering among the competing procedures. In the uncontaminated representative settings, all methods converge reliably, and the retained dimension remains low and stable, with FPCA-based fits selecting roughly $2.4$-$2.5$ components on average and RFPCA-based fits selecting roughly $2.8$-$2.9$ components. Within this regime, the classical FPCA and the non-Fisher-consistent redescending benchmarks are the least expensive procedures, whereas the proposed Fisher-consistent redescending estimators require substantially more time. This additional burden is already visible under clean data through larger outer iteration counts and, for the proposed methods only, a nonzero total number of inner Newton updates. Thus, even in the absence of contamination, the proposed algorithm should not be interpreted as a computational shortcut; rather, it incurs extra cost because it solves the full bias-corrected estimating system and updates the spatial parameter through an explicit inner Newton loop.

The differences become much more pronounced in contaminated settings, particularly when contamination is combined with stronger spatial dependence. In the representative case with $n=500$, $\vartheta=0.75$, $\alpha=10\%$, and $\mu_{\widetilde{\varepsilon}}=20$, the existing RFPCA-Huber procedure becomes computationally unstable, with a convergence rate dropping to $0\%$, while its elapsed time remains very large because the algorithm repeatedly reaches the iteration cap. By contrast, the non-Fisher-consistent redescending benchmarks FCSAR-D\textsuperscript{*} and FCSAR-A\textsuperscript{*} continue to converge essentially always in the representative configurations shown in the table. The proposed Fisher-consistent procedures remain far more expensive than the starred alternatives, and in the same severe regime, they also exhibit lower convergence rates and very large average outer and inner iteration counts. This comparison is especially informative because it shows that the main computational cost is not merely incurred by using redescending losses but by combining them with the bias-corrected Fisher-consistent second-stage updating scheme. The proposed FCSAR procedures achieve their robustness gains at a non-negligible computational cost, whereas the starred redescending benchmarks show that a simpler non-Fisher-consistent second stage can be considerably cheaper but may sacrifice some of the inferential and predictive improvements documented elsewhere in the simulation study.

\section{Empirical data analysis: Air quality data}\label{sec:real_data}

To evaluate the practical performance of the proposed methodology, we analyze a real air-quality data set obtained through the \texttt{saqgetr} \Rlogo \ package \citep{saqgetr}, which provides direct access to measurements collected by the European Environment Agency. After preprocessing and retaining stations with sufficient data availability in 2022 and 2023, the final dataset consists of 202 monitoring stations distributed across France. The spatial distribution of these stations is displayed in Figure~\ref{fig:Fig_3}. Our primary goal is to investigate whether monthly $\mathrm{NO}_{\mathrm{x}}$ profiles can explain annual average $\mathrm{PM}_{10}$ concentrations after accounting for spatial dependence among neighboring monitoring stations, and to assess whether the proposed redescending procedures yield more reliable estimation and prediction than the competing methods in the presence of atypical observations. This specification is scientifically meaningful because $\mathrm{NO}_{\mathrm{x}}$ is a major indicator of traffic- and combustion-related emissions and is also involved in atmospheric processes that contribute to particulate formation, while $\mathrm{PM}_{10}$ is a widely studied pollutant with direct environmental and public-health relevance.

\begin{figure}[!htbp]
\centering
\includegraphics[width=10cm]{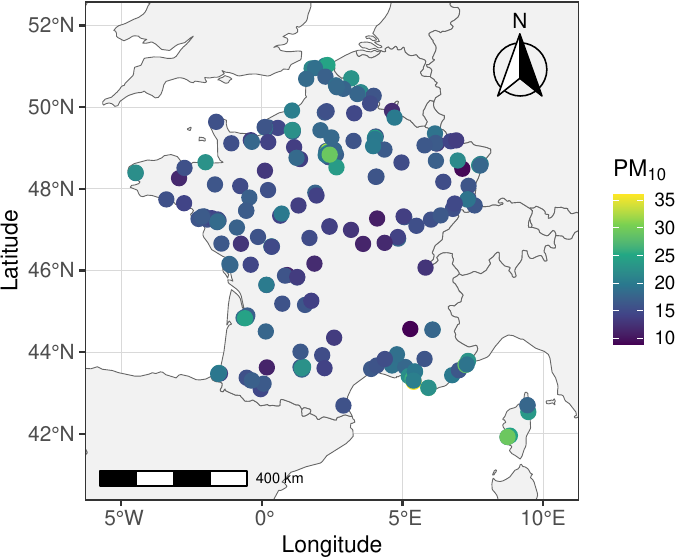}
\caption{\small{Spatial distribution of the 202 air-quality monitoring stations used in the empirical analysis across France. The color scale represents annual average $\mathrm{PM}_{10}$ concentrations in 2022 (in $\mu$g/m$^3$), with darker colors indicating higher levels. For descriptive interpretation, the displayed values may be viewed as low ($<15$), moderate ($15$-$25$), and relatively high ($>25$) within the observed sample range.}} \label{fig:Fig_3}
\end{figure}

Figure~\ref{fig:Fig_3} suggests some broad spatial heterogeneity in annual average $\mathrm{PM}_{10}$ concentrations across France. In particular, the map is compatible with somewhat higher values in parts of the northern and eastern monitoring network and comparatively lower values in several western and southwestern locations, although the pattern is not purely monotone in any single geographic direction. We therefore do not interpret Figure~\ref{fig:Fig_3} as evidence of a simple directional trend alone; rather, it indicates geographically structured variation, which is examined more formally through Local Moran's $I$ analysis in Figure~\ref{fig:Fig_5}. For descriptive purposes, the observed $\mathrm{PM}_{10}$ values in Figure~\ref{fig:Fig_3} can be read as low ($<15$), moderate ($15$-$25$), and relatively high ($>25$) within the sample range shown on the color scale.

For each monitoring station $i\in\{1,\ldots,202\}$, let $Y_i$ denote the annual average $\mathrm{PM}_{10}$ concentration and let $\X_i(t)$ denote the $\mathrm{NO}_{\mathrm{x}}$ concentration observed as a function of month, where the functional support is taken as $\mathcal{I}=[1,12]$. In this application, data from 2022 are used as the in-sample estimation period, whereas data from 2023 are reserved for out-of-sample prediction. Accordingly, the empirical SSoFRM considered in this section is
\begin{equation*}
Y_i = \beta_0 + \vartheta \sum_{j=1}^{202} w_{ij} Y_j + \int_{1}^{12} \X_i(t) \beta(t) dt + \varepsilon_i, \qquad i\in\{1,\ldots,202\}.
\end{equation*}

Figure~\ref{fig:Fig_4} presents the observed $\mathrm{NO}_{\mathrm{x}}$ curves and the $\mathrm{PM}_{10}$ values for 2022 and 2023. From this figure, there is heterogeneity across stations in both years, indicating that air-quality dynamics vary considerably across France. We investigate potential atypical observations separately for the functional predictor and the scalar response. For the monthly $\mathrm{NO}_{\mathrm{x}}$ trajectories, potential functional outliers are identified using the magnitude-shape plot based on directional outlyingness proposed by \citet{Dai}, implemented in the \Rlogo\ package \texttt{fdaoutlier} \citep{fdaoutlier}. For the scalar response $\mathrm{PM}_{10}$, potential atypical observations are identified using the classical boxplot rule. The diagnostic results indicate that both variables contain approximately $4\%$ atypical observations, including unusually small and unusually large values. These observations are not removed from the analysis. Instead, they are retained in the data set and handled through the robust estimation strategy: the RFPCA step reduces the influence of atypical functional trajectories, and the second-stage redescending M-estimation step limits the effect of extreme response-side residuals. We emphasize, however, that this diagnostic analysis is marginal in nature: it was designed to identify unusual functional and scalar observations, not to formally assess whether the flagged stations are spatially influential in the sense of exerting disproportionate impact on the fitted spatial dependence structure. A diagnostic study of spatial influence was not conducted here and would be a useful extension of the present empirical analysis.
\begin{figure}[!htb]
\centering
\includegraphics[width=8.7cm]{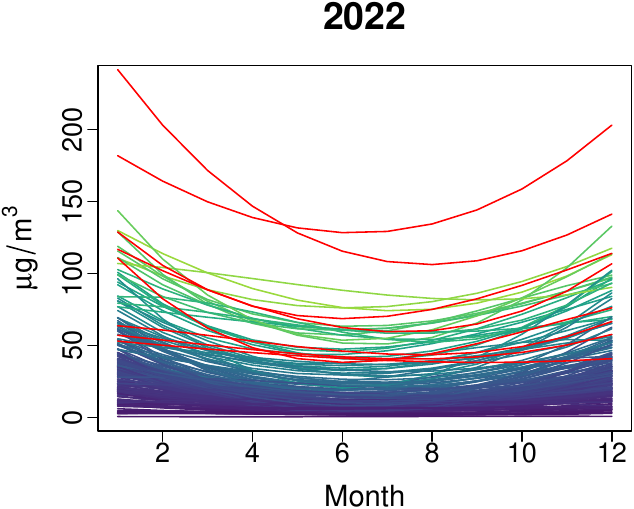}
\qquad
\includegraphics[width=7.95cm]{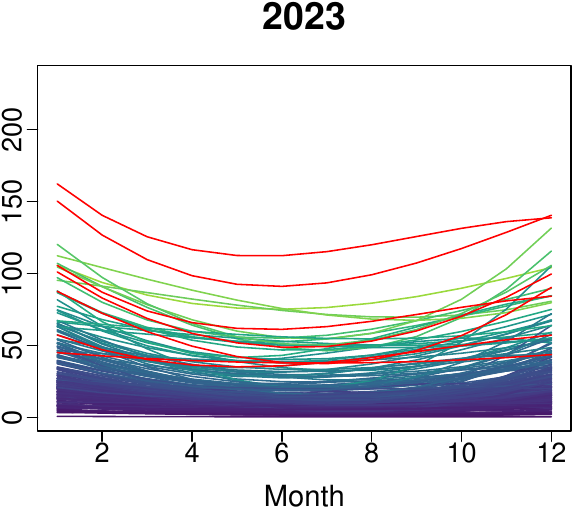}
\\  
\includegraphics[width=8.7cm]{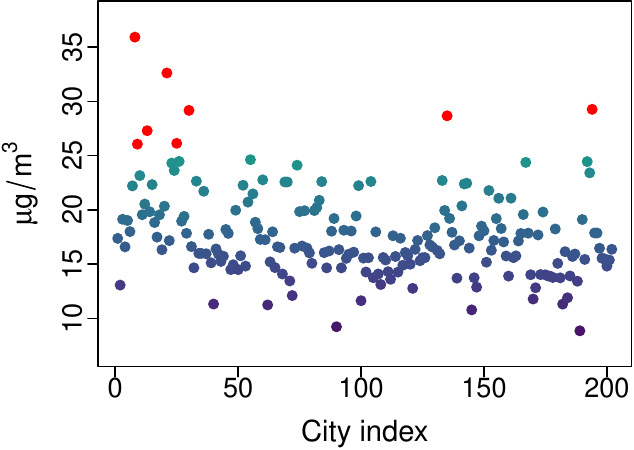}
\qquad
\includegraphics[width=7.95cm]{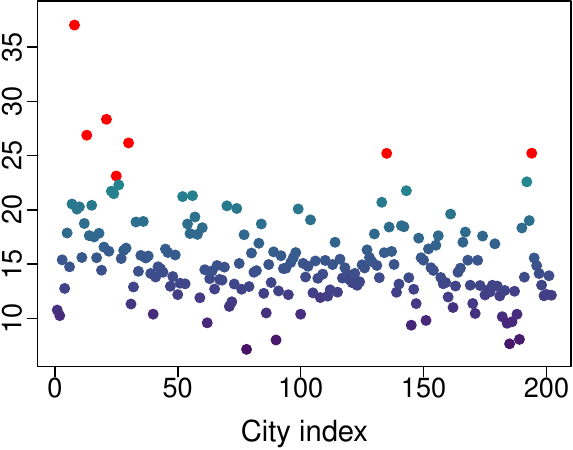}
\caption{\small{Observed air-pollution measurements from the 202 monitoring stations in France. The top panels display the monthly $\mathrm{NO}_{x}$ trajectories, treated as functional predictors on the support $\mathcal{I}=[1,12]$, for 2022 (left) and 2023 (right). The bottom panels display the corresponding annual average $\mathrm{PM}_{10}$ concentrations, used as scalar responses, for 2022 (left) and 2023 (right). Different colors correspond to different monitoring stations, and observations identified as potential outliers are highlighted in red.}}\label{fig:Fig_4}
\end{figure}

To assess the sensitivity of the results to alternative specifications of spatial interaction, we consider three spatial weight matrices $\mathbf{W}=[w_{ij}]_{1\le i,j\le 202}$ reflecting different definitions of spatial proximity. First, we define a spatial weight matrix based on inverse-squared great-circle distances, denoted by $\mathbf{W}$\textsuperscript{IDW} \citep{Ertur2007}. Specifically, for $i\neq j$, we define $w_{ij} = d_{ij}^{-2}/ \sum_{\ell\neq i} d_{i\ell}^{-2}$, $w_{ii}=0$, where $d_{ij}$ is the Haversine distance between stations $i$ and $j$, computed from their latitude and longitude coordinates. More precisely, $d_{ij}=Rc$ and $c=2 \mathrm{atan2} \left(\sqrt{a},\sqrt{1-a}\right)$ with $a = \sin^2 \left(\frac{\Delta u}{2}\right) + \cos(u_i)\cos(u_j)\sin^2 \left(\frac{\Delta v}{2}\right)$ where $u_i$ and $u_j$ denote latitudes in radians, $v_i$ and $v_j$ denote longitudes in radians, $\Delta u=u_j-u_i$, $\Delta v=v_j-v_i$, and $R=6371$ km is the Earth's mean radius. This construction assigns greater influence to stations that are geographically closer. It is therefore well suited to air-pollution data, for which nearby stations typically experience more similar meteorological and emission conditions than distant ones. Second, we construct a k-nearest neighbors spatial weight matrix, denoted by $\mathbf{W}$\textsuperscript{kNN}. For each location $i$, the $\text{k}=7$ nearest neighbors are identified based on Euclidean distances between spatial coordinates. The choice of $\text{k}=7$ is motivated by connectivity considerations: it ensures a fully connected spatial graph, provides a more stable neighborhood structure, and avoids potential sensitivity at the connectivity threshold. The weights are then defined as $w_{ij} = 1/7$ if location $j$ is among the $7$ nearest neighbors of $i$, and $w_{ij}=0$ otherwise, with $w_{ii}=0$. This specification ensures that each location is influenced by a fixed number of neighboring observations, yielding a stable and well-connected spatial structure. Finally, we consider a binary distance-band spatial weight matrix, denoted by $\mathbf{W}$\textsuperscript{DB}. Let $d_{ij}$ denote the Haversine distance between locations $i$ and $j$, as defined above. The weights are defined as $w_{ij}=1$ if $d_{ij} \le d_0$ and $w_{ij}=0$ otherwise, with $w_{ii}=0$. The threshold parameter $d_0$ is selected in a data-driven manner to ensure consistency with the k-nearest neighbors specification. Specifically, for each location, the distance to its $7$\textsuperscript{th} nearest neighbor is computed, and $d_0$ is defined as the $90$\textsuperscript{th} percentile of these distances, providing a robust threshold that avoids sensitivity to extreme distance. This specification defines spatial interaction within a fixed geographic radius and serves as an alternative discrete representation of spatial proximity. Note that, in all specifications, the spatial weight matrices are row-standardized so that $\sum_{j=1}^{202} w_{ij}=1$ for each $i$.

Before fitting the model, we examine the Local spatial structure of the response through the local Moran's $I$ statistic. The Moran's $I$ scatter plot in Figure~\ref{fig:Fig_5} shows a clear positive association between centered $\mathrm{PM}_{10}$ concentrations and their spatially lagged counterparts. This indicates positive local spatial autocorrelation in the 2022 $\mathrm{PM}_{10}$ field: stations with relatively high $\mathrm{PM}_{10}$ values tend to be surrounded by stations with similarly high values, while stations with relatively low values tend to be located near other low-valued stations. Hence, the empirical evidence strongly supports the inclusion of the SAR lag term $\vartheta \sum_{j=1}^{202} w_{ij}Y_j$ in the SSoFRM. At the same time, the visible spread of points away from the main cloud is also compatible with the presence of local anomalies, again suggesting that a robust estimation strategy is desirable.

\begin{figure}[!htbp]
\centering
\includegraphics[width=5.475cm]{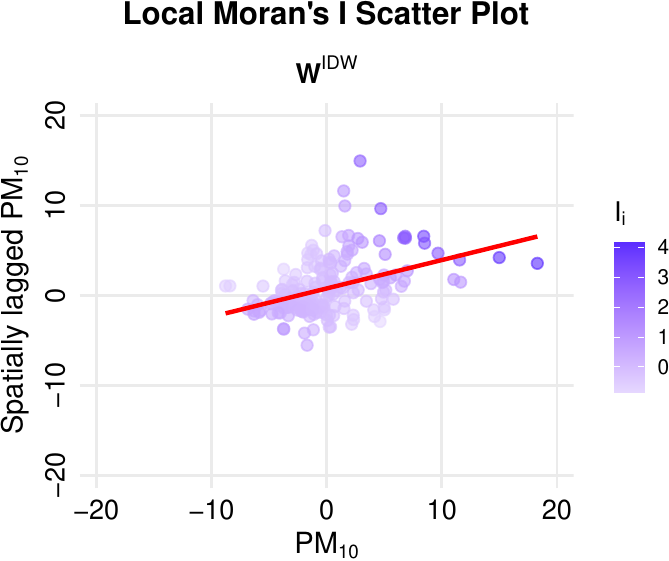}
\quad
\includegraphics[width=5.475cm]{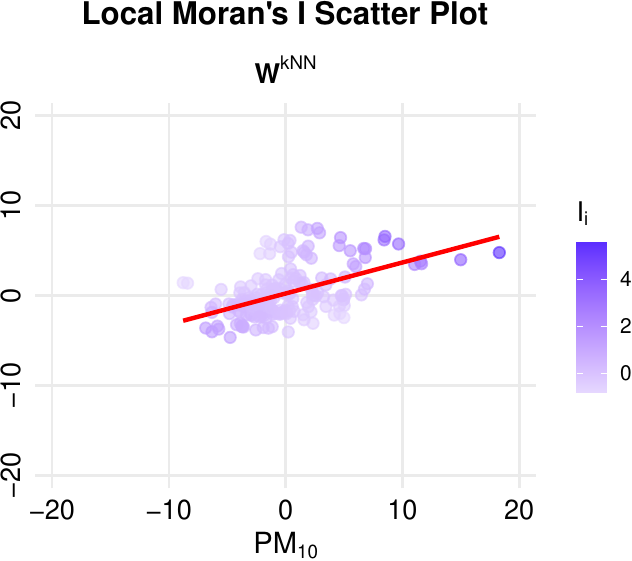}
\quad
\includegraphics[width=5.475cm]{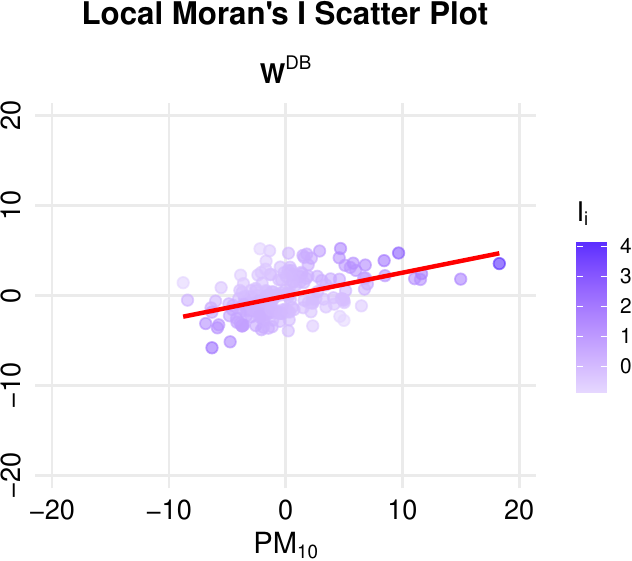}
\caption{\small{\commUB{Local} Moran's $I$ scatter plots for the 2022 annual average $\mathrm{PM}_{10}$ concentrations. The horizontal axis shows the centered $\mathrm{PM}_{10}$ values, while the vertical axis shows their spatial lag based on three alternative row-standardized weight matrices: the inverse-distance weighting $\mathbf{W}$\textsuperscript{IDW} (left panel), the k-nearest neighbors $\mathbf{W}$\textsuperscript{kNN} (middle panel), and the distance-band $\mathbf{W}$\textsuperscript{DB} (right panel).}}\label{fig:Fig_5}
\end{figure}

We fit the six competing procedures, namely FPCA, RFPCA, FCSAR-D\textsuperscript{*}, FCSAR-A\textsuperscript{*}, FCSAR-D, and FCSAR-A, to the 2022 data under each of the three spatial weight matrix specifications $\mathbf{W}$\textsuperscript{IDW}, $\mathbf{W}$\textsuperscript{kNN}, and $\mathbf{W}$\textsuperscript{DB}. The fitted models are then used to predict the 2023 annual average $\mathrm{PM}_{10}$ concentrations from the 2023 monthly $\mathrm{NO}_{\mathrm{x}}$ curves. Under the $95\%$ explained-variation rule, the FPCA-based fit retains one principal component, whereas all RFPCA-based procedures retain two principal components. Because the robust methods share the same first-stage RFPCA decomposition, the retained dimension is identical for RFPCA, FCSAR-D\textsuperscript{*}, FCSAR-A\textsuperscript{*}, FCSAR-D, and FCSAR-A. To reduce the dominance of a few extreme prediction errors in the 2023 test period, we summarize out-of-sample performance by both the sum of squared error, $\mathrm{SSE}=\sum_{i=1}^{202}(\widehat{Y}_i-Y_i)^2$, and the corresponding mean squared prediction error, reported for the untrimmed case ($0\%$ trimming) as well as for $5\%$, $10\%$, and $20\%$ trimming. The resulting estimates of $\widehat{\vartheta}$ and the out-of-sample prediction summaries are reported in Tables~\ref{tab:tab_real_vartheta} and \ref{tab:tab_real_errors}, respectively.
\begin{table}[!htbp]
\centering
\caption{\small Estimated spatial dependence parameter $\widehat{\vartheta}$ from the 2022 air-quality data under different spatial weight matrix specifications.}
\label{tab:tab_real_vartheta}
\renewcommand{\arraystretch}{1.0}
\tabcolsep 0.22in
\begin{tabular}{@{}lcccccc@{}}
\toprule
$\mathbf{W}$ & FPCA & RFPCA & FCSAR-D\textsuperscript{*} & FCSAR-A\textsuperscript{*} & FCSAR-D & FCSAR-A \\
\midrule
$\mathbf{W}$\textsuperscript{IDW} & 0.27 & 0.28 & 0.25 & 0.24 & 0.27 & 0.25 \\
$\mathbf{W}$\textsuperscript{kNN} & 0.40 & 0.20 & 0.38 & 0.36 & 0.38 & 0.35 \\
$\mathbf{W}$\textsuperscript{DB} & 0.47 & 0.19 & 0.44 & 0.44 & 0.44 & 0.41 \\
\bottomrule
\end{tabular}
\end{table}

\begin{table}[!htb]
\centering
\caption{\small Out-of-sample SSE and MSPE values for predicting 2023 annual average $\mathrm{PM}_{10}$ concentrations under different spatial weight matrix specifications, reported for the untrimmed case (i.e., $0\%$ trimming) and various trimming percentages (TP).}
\label{tab:tab_real_errors}
\renewcommand{\arraystretch}{1.0}
\tabcolsep 0.1in
\begin{tabular}{@{}lccrrrrrr@{}}
\toprule
TP & $\mathbf{W}$ & Statistic & FPCA & RFPCA & FCSAR-D\textsuperscript{*} & FCSAR-A\textsuperscript{*} & FCSAR-D & FCSAR-A \\
\midrule
$0\%$ & $\mathbf{W}$\textsuperscript{IDW} & SSE & 2369.35 & 2147.86 & 1836.62 & 1828.49 & 2058.84 & 2050.80 \\
&  & MSPE & 11.73 & 10.63 & 9.09 & 9.05 & 10.19 & 10.15 \\
\cmidrule(lr){2-9}
& $\mathbf{W}$\textsuperscript{kNN} & SSE & 2360 & 2244.38 & 1843.93 & 1844.93 & 2061.66 & 1995.30 \\
&  & MSPE & 11.68 & 11.11 & 9.13 & 9.13 & 10.21 & 9.88 \\
\cmidrule(lr){2-9}
& $\mathbf{W}$\textsuperscript{DB} & SSE & 2396.94 & 2267.77 & 1968.42 & 1959.41 & 2053.92 & 2028.69 \\
&  & MSPE & 11.87 & 11.23 & 9.74 & 9.70 & 10.17 & 10.04 \\
\midrule
$5\%$ & $\mathbf{W}$\textsuperscript{IDW} & SSE & 1644.19 & 1411.81 & 976.11 & 974.15 & 1305.32 & 1290.99 \\
&  & MSPE & 8.56 & 7.35 & 5.08 & 5.07 & 6.80 & 6.72 \\
\cmidrule(lr){2-9}
& $\mathbf{W}$\textsuperscript{kNN} & SSE & 1655.95 & 1507.77 & 1090.15 & 1084.49 & 1330.07 & 1251.82 \\
&  & MSPE & 8.62 & 7.85 & 5.68 & 5.65 & 6.93 & 6.52 \\
\cmidrule(lr){2-9}
& $\mathbf{W}$\textsuperscript{DB} & SSE & 1675.32 & 1516.55 & 1047.26 & 1040.98 & 1277.22 & 1248.14 \\
&  & MSPE & 8.73 & 7.90 & 5.45 & 5.42 & 6.65 & 6.50 \\
\midrule
$10\%$ & $\mathbf{W}$\textsuperscript{IDW} & SSE & 1300.54 & 1087.51 & 678.58 & 679.27 & 998.92 & 984.67 \\
&  & MSPE & 7.15 & 5.98 & 3.73 & 3.73 & 5.49 & 5.41 \\
\cmidrule(lr){2-9}
& $\mathbf{W}$\textsuperscript{kNN} & SSE & 1331.84 & 1172.61 & 806.81 & 800.94 & 1033.51 & 963.81 \\
&  & MSPE & 7.32 & 6.44 & 4.43 & 4.40 & 5.68 & 5.30 \\
\cmidrule(lr){2-9}
& $\mathbf{W}$\textsuperscript{DB} & SSE & 1343.31 & 1181.22 & 734.10 & 729.95 & 983.33 & 954.47 \\
&  & MSPE & 7.38 & 6.49 & 4.03 & 4.01 & 5.40 & 5.24 \\
\midrule
$20\%$ & $\mathbf{W}$\textsuperscript{IDW} & SSE & 883.44 & 708.23 & 381.30 & 382.07 & 629 & 615.89 \\
&  & MSPE & 5.45 & 4.37 & 2.35 & 2.36 & 3.88 & 3.80 \\
\cmidrule(lr){2-9}
& $\mathbf{W}$\textsuperscript{kNN} & SSE & 885.86 & 758 & 502.33 & 498.74 & 657.83 & 603.50 \\
&  & MSPE & 5.47 & 4.68 & 3.10 & 3.08 & 4.06 & 3.73 \\
\cmidrule(lr){2-9}
& $\mathbf{W}$\textsuperscript{DB} & SSE & 896.07 & 763.44 & 445.47 & 441.07 & 631.37 & 607.84 \\
&  & MSPE & 5.53 & 4.71 & 2.75 & 2.72 & 3.90 & 3.75 \\
\bottomrule
\end{tabular}
\end{table}

Tables~\ref{tab:tab_real_vartheta} and \ref{tab:tab_real_errors} yield three main conclusions. First, all methods estimate a positive spatial dependence parameter under all three spatial weight matrix specifications, so the substantive conclusion of positive spatial autocorrelation in annual average $\mathrm{PM}_{10}$ concentrations across France is robust to the choice of estimator and to the definition of spatial proximity. This is fully consistent with the local Moran's $I$ analysis in Figure~\ref{fig:Fig_5}. At the same time, the estimated magnitude of $\widehat{\vartheta}$ is more sensitive to the choice of spatial weight matrix than to the estimation method itself. Under inverse-distance weighting, the estimated values lie in a narrow moderate range, whereas the $k$-nearest neighbors and distance-band specifications produce a wider spread, especially for FPCA and RFPCA. Thus, the empirical application suggests that the existence of positive spatial dependence is stable, but its estimated strength depends non-negligibly on how spatial interaction is encoded. 

Second, the predictive ranking is much clearer than the ranking for $\widehat{\vartheta}$. Across all three spatial weight matrices and across all trimming proportions, the FPCA-based procedure produces the largest SSE and MSPE values, indicating the weakest out-of-sample predictive performance. Replacing FPCA with RFPCA systematically improves prediction, confirming that robust functional dimension reduction is already beneficial in this data set. However, the strongest predictive performance is achieved by the redescending second-stage procedures. In particular, the non-Fisher-consistent redescending competitors FCSAR-D\textsuperscript{*} and FCSAR-A\textsuperscript{*} attain the smallest SSE and MSPE values in nearly all entries of Table~\ref{tab:tab_real_errors}, while the proposed Fisher-consistent procedures FCSAR-D and FCSAR-A consistently improve upon FPCA and RFPCA but do not dominate their starred counterparts in this empirical example.

Third, the relative ordering of the methods is highly stable across trimming proportions. Moving from the untrimmed case to $5\%$, $10\%$, and $20\%$ trimming reduces the absolute magnitude of both SSE and MSPE for all procedures, as expected, but it does not materially change the ranking pattern. The FPCA method remains the weakest predictor; RFPCA provides a clear improvement; the proposed Fisher-consistent estimators occupy an intermediate but still clearly competitive position; and the starred redescending competitors remain the strongest in terms of pure predictive accuracy. Among the three spatial weight matrices, the inverse-distance and distance-band specifications tend to produce slightly smaller prediction errors than the $k$-nearest neighbors specification. However, the differences are not large enough to alter the overall substantive interpretation. These results suggest that, in this application, robustification of both the functional decomposition and the second-stage spatial regression is essential for reliable prediction, while Fisher consistency primarily contributes inferential coherence rather than the smallest empirical test-error values in every configuration.

The estimated regression coefficient functions are displayed in Figure~\ref{fig:Fig_6}. Across all three spatial weight matrix specifications, the estimated coefficient functions are predominantly positive over the monthly domain, implying that larger $\mathrm{NO}_{\mathrm{x}}$ concentrations are generally associated with higher annual average $\mathrm{PM}_{10}$ concentrations after accounting for spatial dependence. This positive relationship is reasonable, because $\mathrm{NO}_{\mathrm{x}}$ is both an indicator of combustion-related emissions and a contributor to secondary particle formation. Although the broad functional pattern is similar across methods, the nonrobust FPCA estimate exhibits the greatest local fluctuation, whereas the RFPCA-based procedures produce visibly smoother coefficient functions. Within the robust group, the four redescending procedures yield highly comparable shapes across the three weight specifications, and the proposed Fisher-consistent curves FCSAR-D and FCSAR-A remain smooth and stable throughout the domain. Thus, even though the empirical prediction table favors the starred redescending competitors in terms of SSE and MSPE, Figure~\ref{fig:Fig_6} shows that the proposed Fisher-consistent estimators still recover a coherent and scientifically interpretable functional effect of monthly $\mathrm{NO}_{\mathrm{x}}$ on annual average $\mathrm{PM}_{10}$.
\begin{figure}[!htbp]
\centering
\includegraphics[width=5.475cm]{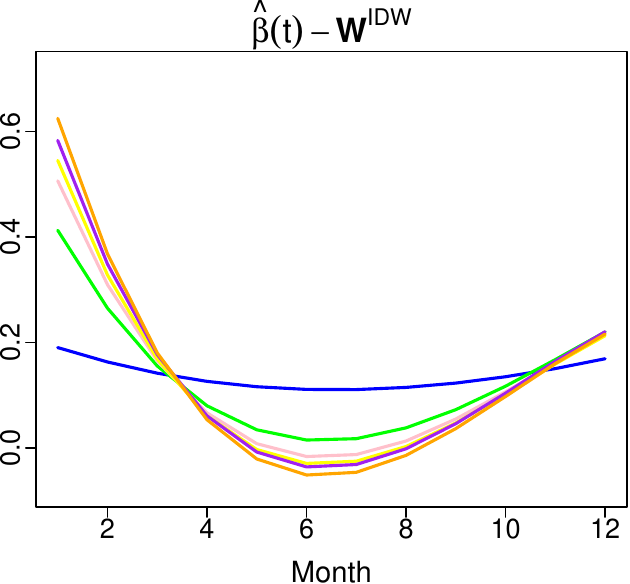}
\quad
\includegraphics[width=5.475cm]{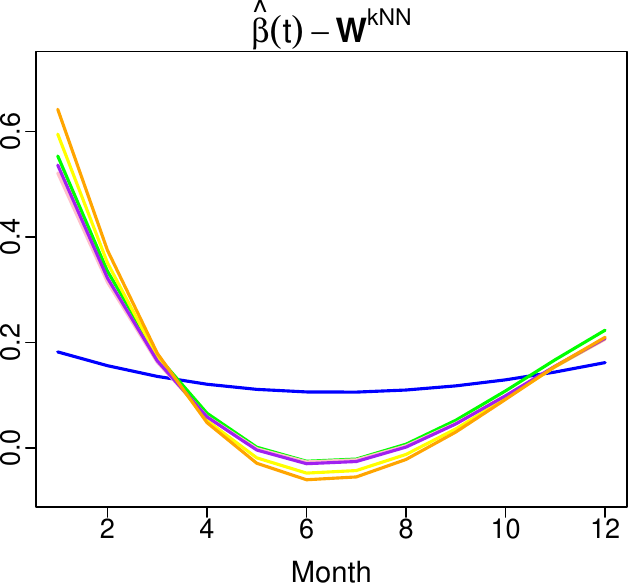}
\quad
\includegraphics[width=5.475cm]{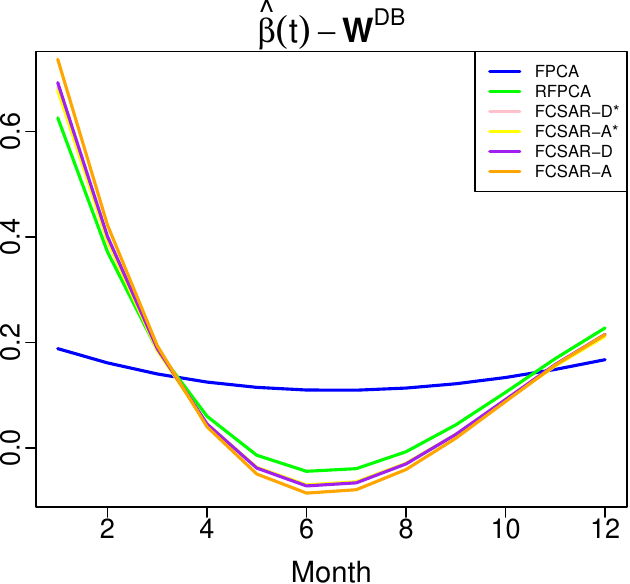}
\caption{\small{Estimated regression coefficient functions $\widehat{\beta}(t)$ obtained from the air-quality data by FPCA, RFPCA, FCSAR-D\textsuperscript{*}, FCSAR-A\textsuperscript{*}, FCSAR-D, and FCSAR-A under three alternative spatial weight matrix: the inverse-distance weighting $\mathbf{W}$\textsuperscript{IDW} (left panel), the k-nearest neighbors $\mathbf{W}$\textsuperscript{kNN} (middle panel), and the distance-band $\mathbf{W}$\textsuperscript{DB} (right panel).}}\label{fig:Fig_6}
\end{figure}

\section{Conclusion}\label{sec:conclusion}

This study has developed a Fisher-consistent redescending estimation procedure for SSoFRM by strengthening the second-stage spatial estimation step in an RFPCA-based robust regression pipeline. In relation to the existing literature, \cite{Huang2021} introduced robustness into SSoFRM through a $t$-distribution-based formulation estimated by EM, while \cite{Beyaztas2026} combined robust functional dimension reduction with a Huber-type robust second-stage estimator for the same model class. Our contribution is to refine this latter direction by replacing the second-stage spatial estimator with a jointly bias-corrected redescending M-estimation system, together with a unified scale equation and a hybrid IRLS-Newton algorithm. Both the simulation study and the empirical analysis indicate that this refinement yields improved resistance to severe contamination while preserving stable estimation of spatial dependence. The additional computational diagnostics show that the proposed procedures are more computationally demanding than the classical FPCA and RFPCA-Huber benchmarks, so the contribution of the hybrid IRLS-Newton step should be interpreted as improved algorithmic transparency and a fully specified bias-corrected update, rather than as a wall-clock speed advantage.

We also provide a large-sample theory for the proposed estimator within a finite-dimensional projected setting. In particular, we show that the population estimating system is correctly centered at the true parameter, yielding Fisher consistency under identification. Under a fixed truncation level, an exact-span condition excluding truncation bias, and additional regularity conditions on identification, uniform convergence, and limiting empirical-process behavior, we derive consistency and asymptotic normality of the full parameter estimator within a triangular-array SAR setting, and obtain the asymptotic distribution of the reconstructed slope function by explicitly accounting for both regression-parameter uncertainty and the estimation of the RFPCA eigenfunctions.

The numerical findings strongly support the practical value of the proposed approach. In simulation settings without contamination, that is, when no vertical outliers are added to the response, and no leverage contamination is introduced into the functional predictor, the proposed procedures remain highly competitive with classical and existing robust alternatives. This indicates that the additional robustness does not come at a substantial loss of efficiency when the data are generated without outlier contamination. Under contamination, however, the benefits become much more pronounced. The Monte Carlo study shows that the proposed redescending procedures substantially improve the estimation of the spatial parameter, the scale parameter, and the slope function, while also delivering markedly better predictive performance than both the classical FPCA-based likelihood estimator and the existing RFPCA-Huber approach. The results are especially favorable in settings with stronger contamination and stronger spatial dependence, where the Andrews-based version often provides the best overall performance.

Although the proposed methodology performs well, the theoretical scope should be interpreted with some care. First, the asymptotic development treats the truncation level in the robust functional decomposition as fixed. It assumes that truncation does not introduce approximation bias, effectively requiring the slope function to lie in the projection span used in the analysis. Thus, the current results do not cover growing-$K$ asymptotics or data-driven component selection. Second, the current asymptotic theory does not cover a general spatially dependent functional predictor setting. Although the functional covariates are spatially indexed in the model and in the motivating applications, Assumption~\ref{as:2} treats the predictor curves as i.i.d. copies of an underlying process to make the robust projection-pursuit RFPCA step mathematically tractable. Thus, the present theory introduces spatial dependence via the SAR response mechanism and the spatial weight matrix, rather than through an explicit dependence structure for the predictors themselves. Third, the theory is developed under i.i.d. symmetric errors that are independent of the regressors, which isolates the spatial dependence to the response through the SAR structure but may be restrictive in applications with more complex dependence or asymmetric contamination. Fourth, the consistency and asymptotic normality results rely on high-level assumptions that guarantee identification, uniform convergence of the estimating map, and suitable limiting behavior of the empirical process. These conditions are standard in Z-estimation arguments, but they are not intended to be the weakest possible primitive conditions for the present problem. Accordingly, the asymptotic results in this paper should be viewed as a theoretical treatment of the proposed estimator in a finite-dimensional projected SSoFRM setting, rather than as a fully general theory. In addition, while the asymptotic results justify inference for linear functionals and $\mathcal{L}^2$ summaries of the slope function, they do not yet provide a full theory for pointwise inference at fixed $t$ or for the construction of uniform confidence bands over the entire domain.

These limitations suggest several directions for future research. One natural extension is to allow multiple functional and scalar predictors within the same Fisher-consistent redescending framework. Another is to develop data-driven selection rules for tuning constants and truncation parameters with accompanying theory. It would also be of substantial interest to generalize the methodology to more complex spatial structures, such as spatial error, Durbin-type, spatiotemporal, or panel-functional models. Further work could investigate generalized response settings, including binary, count, or survival outcomes. Finally, extending the robust functional decomposition step to irregularly observed or longitudinal functional predictors would broaden the scope of the method and provide a useful bridge between robust spatial functional regression and more general longitudinal functional modeling.

\section*{Acknowledgments}


This research was supported by the Scientific and Technological Research Council of Turkey (TUBITAK) (under grant no.~124F096), the Australian Research Council Discovery Project (grant no.~DP230102250), and the Australian Research Council Future Fellowship (grant no.~FT240100338).

\section*{Code and data availability}

The \Rlogo\ code for the proposed approach is documented in the {\tt fcsar} package, available at \url{https://github.com/UfukBeyaztas/fcsar}. The air quality data used in the empirical analysis of this study are provided as online supplementary material.



\newpage
\appendix
\section{Model setup and notation}\label{sec:appendix_setup}

We consider the SSoFRM. For $n$ spatial units, after applying RFPCA with truncation level $K$, we obtain the finite‑dimensional approximation:
\begin{equation*}
\mathbf{Y} = \vartheta \mathbf{W} \mathbf{Y} + \mathbf{Z} \boldsymbol{\gamma} + \boldsymbol{\varepsilon},
\end{equation*}
where $\mathbf{Y} = (Y_1,\ldots,Y_n)^\top \in \mathbb{R}^n$ is the scalar response vector, $\mathbf{W} = (w_{ij})_{n\times n}$ is the row‑normalized spatial weight matrix, $\mathbf{Z} = [\mathbf{1}_n, \boldsymbol{\Xi}] \in \mathbb{R}^{n\times (K+1)}$ with $\boldsymbol{\Xi} = (\widehat{\xi}_{ik})_{n\times K}$ containing the (centered) RFPCA scores, $\boldsymbol{\gamma} = (\beta_0,b_1,\ldots,b_K)^\top \in \mathbb{R}^{K+1}$ is the regression coefficient vector, $\vartheta \in (-1,1)$ is the spatial dependence parameter, and $\boldsymbol{\varepsilon} = (\varepsilon_1,\ldots,\varepsilon_n)^\top \in \mathbb{R}^n$ with
$\varepsilon_i=\sigma_0 u_i$, where $u_i \stackrel{\mathrm{i.i.d.}}{\sim} F$ for some symmetric distribution $F$
satisfying $\mathbb{E}(u_i)=0$ and $\mathbb{E}(u_i^2)=1$.

Let $\boldsymbol{\theta} = (\boldsymbol{\gamma}^\top,\sigma,\vartheta)^\top \in \mathbb{R}^{K+3}$ denote the full parameter vector, and $\boldsymbol{\theta}_0 = (\boldsymbol{\gamma}_0^\top,\sigma_0,\vartheta_0)^\top$ the true parameter value. Define the residual vector $\mathbf{e}(\boldsymbol{\theta}) = (\mathbf{I}_n - \vartheta \mathbf{W})\mathbf{Y} - \mathbf{Z}\boldsymbol{\gamma}$, with $i$\textsuperscript{th} component $e_i(\boldsymbol{\theta}) = \{(\mathbf{I}_n - \vartheta \mathbf{W})\mathbf{Y} - \mathbf{Z}\boldsymbol{\gamma}\}_i$. 

The Fisher‑consistent redescending M‑estimator $\widehat{\boldsymbol{\theta}}$ solves the following system of unbiased estimating equations:
\begin{align}
\frac{1}{n}\sum_{i=1}^n \psi_2 \left(\frac{e_i(\boldsymbol{\theta})}{\sigma}\right)\mathbf{z}_i &= \mathbf{0}_{K+1}, \label{eq:est_eq_gamma} \\
\frac{1}{n}\sum_{i=1}^n \left\lbrace \psi_2 \left(\frac{e_i(\boldsymbol{\theta})}{\sigma}\right)\frac{e_i(\boldsymbol{\theta})}{\sigma} - \kappa \right\rbrace &= 0, \label{eq:est_eq_sigma} \\
\frac{1}{n}\sum_{i=1}^n \left\lbrace \psi_2 \left(\frac{e_i(\boldsymbol{\theta})}{\sigma}\right)(\mathbf{W}\mathbf{Y})_i - \mathrm{bias}_i(\boldsymbol{\theta}) \right\rbrace &= 0, \label{eq:est_eq_rho}
\end{align}
where $\kappa = \mathbb{E}\bigl\{\psi_2(\varepsilon_i/\sigma_0)(\varepsilon_i/\sigma_0)\bigr\}$ is a consistency constant ensuring Fisher consistency of the scale estimator, $\mathbf{z}_i \in \mathbb{R}^{K+1}$ is the $i$\textsuperscript{th} row of $\mathbf{Z}$, and the bias correction term is defined as $\mathrm{bias}_i(\boldsymbol{\theta}) = \bigl[\mathbf{W}(\mathbf{I}_n-\vartheta\mathbf{W})^{-1}\bigr]_{ii}\sigma\kappa$ with $[\cdot]_{ii}$ denoting the $i$\textsuperscript{th} diagonal element. This term compensates for the endogeneity of the spatial lag, ensuring that $\mathbb{E}\{\psi_2(\varepsilon_i/\sigma_0)(\mathbf{W}\mathbf{Y})_i\} = \mathrm{bias}_i(\boldsymbol{\theta}_0)$, so that equation~\eqref{eq:est_eq_rho} has zero expectation at the true parameters.

\newpage
\section{Bias correction for the initial estimator}\label{sec:appendix_bias}

To obtain a $\sqrt{n}$-consistent initial value $\vartheta^{(0)}$ for the hybrid IRLS–Newton algorithm described in Section~\ref{subsec:algorithm}, we employ the analytical bias correction derived by \citet{Zhan2025}. This correction is applied to the NLS estimator of the spatial parameter and is based on a second‑order expansion of its bias. The derivation is reproduced here for completeness.

First, consider the NLS estimator obtained by minimizing $\|\mathbf{e}(\boldsymbol{\gamma},\vartheta)\|_2^2/2$ with respect to $(\boldsymbol{\gamma},\vartheta)$. Let $\mathbf{Z}=[\mathbf{1}_n,\boldsymbol{\Xi}]$ and define the residual‑maker matrix $\mathbf{M}:=\mathbf{I}_n-\mathbf{Z}(\mathbf{Z}^\top\mathbf{Z})^{-1}\mathbf{Z}^\top$, which projects onto the orthogonal complement of the column space of $\mathbf{Z}$. Eliminating $\boldsymbol{\gamma}$ via the Frisch–Waugh–Lovell theorem yields the NLS estimator of $\vartheta$,
\begin{equation*}
\widehat{\vartheta}_{\mathrm{NLS}} = \frac{\mathbf{Y}^\top \mathbf{W}^\top \mathbf{M}\mathbf{Y}}{\mathbf{Y}^\top \mathbf{W}^\top \mathbf{M}\mathbf{W}\mathbf{Y}}.
\end{equation*}

Define $\mathbf{S}=\mathbf{S}(\vartheta)=(\mathbf{I}_n-\vartheta\mathbf{W})^{-1}$ and the mean $\boldsymbol{\mu}_Y=\mathbb{E}(\mathbf{Y})=\mathbf{S}\mathbf{Z}\boldsymbol{\gamma}$. Under the conventional SAR covariance representation, we have $\operatorname{Cov}(\mathbf{Y})=\sigma^2 \mathbf{Q}_{YY}$, where $\mathbf{Q}_{YY}=\mathbf{S}\mathbf{S}^\top$. Write the bias as $\operatorname{bias}(\widehat{\vartheta}_{\mathrm{NLS}}) :=\mathbb{E}(\widehat{\vartheta}_{\mathrm{NLS}}-\vartheta) =\mathbb{E}(\eta / \nu)$, with the quadratic forms
\begin{equation*}
\eta:=\mathbf{Y}^\top \mathbf{W}^\top \mathbf{M}\mathbf{S}^{-1}\mathbf{Y}, \qquad 
\nu:=\mathbf{Y}^\top \mathbf{W}^\top \mathbf{M}\mathbf{W}\mathbf{Y}.
\end{equation*}
Using a second‑order Taylor expansion for the ratio of quadratic forms, we obtain an approximation
\begin{equation}\label{eq:bias_approx_app}
\operatorname{bias}(\widehat{\vartheta}_{\mathrm{NLS}}) \approx \frac{\mathbb{E}(\eta)}{\mathbb{E}(\nu)} -\frac{\operatorname{Cov}(\eta,\nu)}{\{\mathbb{E}(\nu)\}^2} + \frac{\operatorname{Var}(\nu) \mathbb{E}(\eta)}{\{\mathbb{E}(\nu)\}^3}.
\end{equation}

Set $\mathbf{A}:=\mathbf{W}^\top \mathbf{M}(\mathbf{I}_n-\vartheta\mathbf{W}) =\mathbf{W}^\top \mathbf{M}\mathbf{S}^{-1}$ and $\mathbf{B}:=\mathbf{W}^\top \mathbf{M}\mathbf{W}$. Standard results for moments of quadratic forms give
\begin{align*}
\mathbb{E}(\eta) &= \operatorname{tr} \bigl(\mathbf{A} \sigma^2\mathbf{Q}_{YY}\bigr) + \boldsymbol{\mu}_Y^\top \mathbf{A}\boldsymbol{\mu}_Y, \\
\mathbb{E}(\nu) &= \operatorname{tr} \bigl(\mathbf{B} \sigma^2\mathbf{Q}_{YY}\bigr) + \boldsymbol{\mu}_Y^\top \mathbf{B}\boldsymbol{\mu}_Y, \\
\operatorname{Cov}(\eta,\nu) &= 2\operatorname{tr} \bigl(\mathbf{A} \sigma^2\mathbf{Q}_{YY} \mathbf{B} \sigma^2\mathbf{Q}_{YY}\bigr)
+4 \boldsymbol{\mu}_Y^\top \mathbf{A} \sigma^2\mathbf{Q}_{YY} \mathbf{B}\boldsymbol{\mu}_Y, \\
\operatorname{Var}(\nu) &= 2\operatorname{tr} \bigl(\mathbf{B} \sigma^2\mathbf{Q}_{YY} \mathbf{B} \sigma^2\mathbf{Q}_{YY}\bigr)
+4 \boldsymbol{\mu}_Y^\top \mathbf{B} \sigma^2\mathbf{Q}_{YY} \mathbf{B}\boldsymbol{\mu}_Y.
\end{align*}

To evaluate these expressions, we plug in the NLS estimates. Compute $(\widehat{\vartheta}_{\mathrm{NLS}},\widehat{\boldsymbol{\gamma}}_{\mathrm{NLS}})$, set $\widehat{\mathbf{S}}:=\mathbf{S}(\widehat{\vartheta}_{\mathrm{NLS}})$, $\widehat{\boldsymbol{\mu}}_Y:=\widehat{\mathbf{S}}\mathbf{Z}\widehat{\boldsymbol{\gamma}}_{\mathrm{NLS}}$, $\widehat{\mathbf{Q}}_{YY}:=\widehat{\mathbf{S}}\widehat{\mathbf{S}}^\top$, and estimate $\sigma^2$ by the usual NLS residual variance $\widehat{\sigma}^2_{\mathrm{NLS}}=\|\widehat{\mathbf{e}}_{\mathrm{NLS}}\|_2^2/(n-K-1)$ with $\widehat{\mathbf{e}}_{\mathrm{NLS}}=(\mathbf{I}_n-\widehat{\vartheta}_{\mathrm{NLS}}\mathbf{W})\mathbf{Y}-\mathbf{Z}\widehat{\boldsymbol{\gamma}}_{\mathrm{NLS}}$. Denote by $\widehat{\operatorname{bias}}(\widehat{\vartheta}_{\mathrm{NLS}})$ the resulting plug‑in estimate of \eqref{eq:bias_approx_app}. The bias‑corrected spatial parameter is then
$\widehat{\vartheta}_{\mathrm{BC}} := \widehat{\vartheta}_{\mathrm{NLS}}-\widehat{\operatorname{bias}}(\widehat{\vartheta}_{\mathrm{NLS}})$. 

This corrected value is used as the initial $\vartheta^{(0)}$ in Algorithm~\ref{alg:FCSAR}. Because $\widehat{\vartheta}_{\mathrm{BC}}$ is $\sqrt{n}$-consistent under standard regularity conditions \citep[see, e.g.,][]{Ma2020}, it provides an excellent starting point for the subsequent Newton‑Raphson iterations, ensuring rapid convergence of the unbiased estimator.

\newpage
\section{Regularity conditions}\label{sec:appendix_conditions}

We require the following assumptions:
\begin{assumption}\label{as:1}
The parameter space $\Theta$ is compact with $\boldsymbol{\theta}_0$ in its interior. Specifically, $\boldsymbol{\gamma} \in \mathcal{G}\subset \mathbb{R}^{K+1}$ compact, $\sigma \in [\sigma_{\min}, \sigma_{\max}] \subset (0, \infty)$, and $\vartheta \in (-1 + \tau, 1 - \tau)$ for some $\tau > 0$.
\end{assumption}

\begin{assumption}\label{as:2}
For the asymptotic theory, the functional predictors $\mathcal{X}_1,\ldots,\mathcal{X}_n$ are assumed to be i.i.d. copies of a square-integrable process $\mathcal{X}\in\mathcal{H}$ with $\mathbb{E}\|\mathcal{X}\|_{\mathcal{H}}^{4}<\infty$. The covariance operator of $\mathcal{X}$ has eigenvalues $\lambda_1\ge\lambda_2\ge\cdots>0$ with $\sum_{k\ge1}\lambda_k<\infty$, and there is an eigen-gap $\lambda_k>\lambda_{k+1}$ for $k \in \{1,\ldots,K\}$. The truncation level $K$ is fixed. The RFPCA estimators are consistent: for each $k\le K$,
$\|\widehat{\phi}_k-\phi_k\|_{\mathcal{H}}=o_p(1)$ and $\widehat{\xi}_{ik}-\xi_{ik}=o_p(1)$, where $\xi_{ik}=\langle \mathcal{X}_i-\mu,\phi_k\rangle_{\mathcal{H}}$.
\end{assumption}

\begin{assumption}\label{as:3}
The M‑scale functional $\sigma_M$ with loss function $\rho_1$ and tuning constants $(c,\delta)$ is
\begin{enumerate}
\item[(i)] weakly continuous, i.e., if $\mathbb{P}_n \stackrel{\omega}{\rightarrow} \mathbb{P}$, then $\sigma_M(\mathbb{P}_n) \to \sigma_M(\mathbb{P})$;
\item[(ii)] Fisher‑consistent for the standard deviation at the normal model: $\sigma_M(\Phi) = 1$.
\end{enumerate}
For the specific choice
\begin{equation*}
\rho_1(u)=\frac{u^2}{2}\Bigl(1-\frac{u^2}{c^2}+\frac{u^4}{3c^4}\Bigr)\mathbf{1}(|u|\le c)+\frac{c^2}{6}\mathbf{1}(|u|>c)
\end{equation*}
with $c=1.56$ and $\delta=0.5$, both properties hold \citep{Bali2011}.
\end{assumption}

\begin{assumption}\label{as:4}
The regression coefficient function $\beta(t)$ lies in the span of $\phi_1(t), \ldots, \phi_K(t)$, i.e., $\beta(t) = \sum_{k=1}^K b_k \phi_k(t)$. This ensures that the truncation does not introduce bias.
\end{assumption}

\begin{assumption}\label{as:5}
The errors $\varepsilon_i$ are independent of the regressors $\mathbf{z}_i$, i.e., $\varepsilon_i \perp\!\!\!\perp \mathbf{z}_i$ for all $i$, and they are i.i.d. with symmetric density $f_\varepsilon$, $\mathbb{E}[\varepsilon_i] = 0$, $\operatorname{Var}(\varepsilon_i) = \sigma_0^2$, and $\mathbb{E}[|\varepsilon_i|^4] < \infty$.
\end{assumption}

\begin{assumption}\label{as:6}
The $\psi_2$-function (derivative of the redescending loss $\rho_2$) satisfies:
\begin{enumerate}
\item $\mathbb{E}\{\psi_2(\varepsilon_i/\sigma_0)\} = 0$ (automatically holds if $\psi_2$ is odd and $f_\varepsilon$ symmetric).
\item $\mathbb{E}\{\psi_2'(\varepsilon_i/\sigma_0)\} = a > 0$.
\item $\mathbb{E}\{\psi_2^2(\varepsilon_i/\sigma_0)\} = b < \infty$.
\item $\mathbb{E}\{\psi_2(\varepsilon_i/\sigma_0)(\varepsilon_i/\sigma_0)\} = \kappa$, with $0 < \kappa < \infty$.
\item $\mathbb{E}\{|\psi_2'(\varepsilon_i/\sigma_0)| (\varepsilon_i/\sigma_0)^2\} < \infty$ (a mild moment condition).
\end{enumerate}
\end{assumption}

\begin{assumption}\label{as:7}
For each $n\ge1$, the spatial weight matrix $\mathbf{W}_n=[w_{n,ij}]$ is non-stochastic, row-normalized ($\sum_{j=1}^n w_{n,ij}=1$, $w_{n,ii}=0$), and satisfies
\begin{enumerate}
\item $\|\mathbf{W}_n\|_\infty = \mathcal{O}(1)$,
\item $\|\mathbf{W}_n\|_1 = \mathcal{O}(1)$,
\item $\vartheta_0\in(-1,1)$ and $(\mathbf{I}_n-\vartheta_0\mathbf{W}_n)$ is invertible for all sufficiently large $n$.
\end{enumerate}
\end{assumption}

\begin{assumption}\label{as:8}
The design matrices $\mathbf{Z}_n=[\mathbf{1}_n,\boldsymbol{\Xi}_n]$ satisfy:
\begin{enumerate}
\item $n^{-1}\mathbf{Z}_n^\top\mathbf{Z}_n \stackrel{p}{\to} \bm{\mathcal Q}$, where $\bm{\mathcal Q}$ is positive definite;
\item $\max_{1\le i\le n}\|\mathbf{z}_{n,i}\| = o_p(n^{1/2})$;
\item $\mathbb{E}(\mathbf{z}_{n,i}\mathbf{z}_{n,i}^\top)$ exists and is positive definite uniformly in $n$.
\end{enumerate}
\end{assumption}

\begin{assumption}\label{as:9}
The spatial multiplier is uniformly well-behaved:
\begin{enumerate}
\item $\sup_{n\ge1}\sup_{\vartheta\in\Theta_\vartheta}\|(\mathbf{I}_n-\vartheta\mathbf{W}_n)^{-1}\|_\infty < \infty$,
\item $\sup_{n\ge1}\sup_{\vartheta\in\Theta_\vartheta}\|(\mathbf{I}_n-\vartheta\mathbf{W}_n)^{-1}\|_1 < \infty$.
\end{enumerate}
\end{assumption}

\begin{assumption}\label{as:10}
Let $\boldsymbol{\Psi}_n(\boldsymbol{\theta}) := \frac{1}{n}\sum_{i=1}^n \boldsymbol{\psi}_{2,n,i}(\boldsymbol{\theta})$ and $\boldsymbol{\Psi}(\boldsymbol{\theta}) := \lim_{n\to\infty}\mathbb{E}\{\boldsymbol{\Psi}_n(\boldsymbol{\theta})\}$, $\boldsymbol{\theta}\in\Theta$. Then, $\boldsymbol{\theta}_0$ is the unique zero of $\boldsymbol{\Psi}$ on $\Theta$. Equivalently, for every $\epsilon>0$, $\inf_{\|\boldsymbol{\theta}-\boldsymbol{\theta}_0\|\ge \epsilon}
\|\boldsymbol{\Psi}(\boldsymbol{\theta})\|>0$.
\end{assumption}

\begin{assumption}\label{as:11}
For each $n\ge1$, the projected model is
\begin{equation*}
\mathbf{Y}_n = \vartheta_0 \mathbf{W}_n \mathbf{Y}_n + \mathbf{Z}_n \boldsymbol{\gamma}_0 + \boldsymbol{\varepsilon}_n,
\end{equation*}
where $\mathbf{W}_n$ is a non-stochastic $n\times n$ spatial weight matrix with zero diagonal. Moreover:
\begin{enumerate}
\item[(i)] $\sup_{n\ge1}\|\mathbf{W}_n\|_1<\infty$ and $\sup_{n\ge1}\|\mathbf{W}_n\|_\infty<\infty$;
\item[(ii)] for every $\vartheta\in\Theta_\vartheta\subset(-1,1)$, the spatial multiplier $\mathbf{S}_n(\vartheta):=(\mathbf{I}_n-\vartheta\mathbf{W}_n)^{-1}$ exists, and
\begin{equation*}
\sup_{n\ge1}\sup_{\vartheta\in\Theta_\vartheta}
\Bigl\{\|\mathbf{S}_n(\vartheta)\|_1+\|\mathbf{S}_n(\vartheta)\|_\infty\Bigr\}<\infty;
\end{equation*}
\item[(iii)] the sequence $\{\mathbf{W}_n\}$ is such that the projected SAR model is well-defined for all sufficiently large $n$.
\end{enumerate}
\end{assumption}

\begin{assumption}\label{as:12}
There exists an initial estimator $\widehat{\boldsymbol{\theta}}^{(0)}$ that is $\sqrt{n}$-consistent: $\sqrt{n}\bigl(\widehat{\boldsymbol{\theta}}^{(0)}-\boldsymbol{\theta}_0\bigr)=\mathcal{O}_p(1)$. This may be taken, for example, as the bias-corrected NLS estimator used to initialize Algorithm~\ref{alg:FCSAR}.
\end{assumption}

\begin{assumption}\label{as:13}
Let $\boldsymbol{\Psi}_n(\boldsymbol{\theta}) = \frac{1}{n}\sum_{i=1}^n \boldsymbol{\psi}_{2,n,i}(\boldsymbol{\theta}),$ and $\boldsymbol{\Psi}(\boldsymbol{\theta}) = \lim_{n\to\infty}\mathbb{E}\{\boldsymbol{\Psi}_n(\boldsymbol{\theta})\}$. Then $\sup_{\boldsymbol{\theta}\in\Theta} \|\boldsymbol{\Psi}_n(\boldsymbol{\theta})-\boldsymbol{\Psi}(\boldsymbol{\theta})\|
\xrightarrow{p}0$.
\end{assumption}

\begin{assumption}\label{as:14}
The expected Jacobian matrix $\boldsymbol{\mathcal A}(\boldsymbol{\theta}_0) := \lim_{n\to\infty} \mathbb{E}\left\{ \frac{\partial}{\partial\boldsymbol{\theta}^\top} \boldsymbol{\Psi}_n(\boldsymbol{\theta}_0) \right\}$ exists and is nonsingular. Moreover, $\frac{1}{\sqrt{n}}\sum_{i=1}^n \boldsymbol{\psi}_{2,n,i}(\boldsymbol{\theta}_0) \Rightarrow N(\mathbf 0,\boldsymbol{\mathcal B})$ for some finite positive semidefinite matrix $\boldsymbol{\mathcal B}$, where $\boldsymbol{\mathcal B} = \lim_{n\to\infty}
\operatorname{Var}\left( \frac{1}{\sqrt{n}}\sum_{i=1}^n
\boldsymbol{\psi}_{2,n,i}(\boldsymbol{\theta}_0)
\right)$. Finally, there exists a neighborhood
$\mathcal N(\boldsymbol{\theta}_0)\subset\Theta$
such that $\sup_{\boldsymbol{\theta}\in\mathcal N(\boldsymbol{\theta}_0)}
\left\| \frac{\partial}{\partial\boldsymbol{\theta}^\top}\boldsymbol{\Psi}_n(\boldsymbol{\theta}) -
\boldsymbol{\mathcal A}(\boldsymbol{\theta}_0) \right\| \xrightarrow{p}0$
\end{assumption}

Assumptions~\ref{as:2}-\ref{as:3} ensure that the robust projection-pursuit RFPCA is well-defined and that the estimated scores and eigenfunctions used to construct $\mathbf{Z}_n$ are consistent. The eigen-gap and moment conditions in Assumption~\ref{as:2} guarantee identifiability and stability of the leading robust principal directions. Note that Assumption~\ref{as:2} is a theoretical simplification used to control the robust projection step. It should not be interpreted as covering a general spatially dependent functional predictor setting. In the present asymptotic framework, spatial dependence enters through the SAR response structure and the sequence of weight matrices $\{\mathbf{W}_n\}_{n\ge1}$, whereas the predictor curves are treated as i.i.d. for the purpose of establishing consistency of the generated RFPCA quantities. Population Fisher consistency of the RFPCA functional is established separately in Lemma~\ref{lem:rfpca_fisher} under its ellipticity conditions. Assumption~\ref{as:4} is an exact truncation condition ensuring that the projected SAR model does not incur approximation bias. Assumptions~\ref{as:5}-\ref{as:7} guarantee that the bias-corrected estimating equations are correctly centered at the true parameter. Assumptions~\ref{as:1} and~\ref{as:8}-\ref{as:10} provide compactness, design regularity, boundedness of the spatial multiplier, and identification of the unique zero of the population estimating map. Assumption~\ref{as:11} places the asymptotic analysis in a triangular-array SAR framework through the sequence of spatial weight matrices $\{\mathbf{W}_n\}_{n\ge1}$. Assumption~\ref{as:12} provides a $\sqrt{n}$-consistent initializer for the numerical algorithm. Assumptions~\ref{as:13}-\ref{as:14} supply the uniform convergence of the estimating map, existence of the limiting Jacobian, and asymptotic normality of the empirical estimating process required for Z-estimation and the sandwich limit $\boldsymbol{\mathcal V} = \boldsymbol{\mathcal A}^{-1} \boldsymbol{\mathcal B} \boldsymbol{\mathcal A}^{-\top}$.

\newpage
\section{Technical lemmas}\label{sec:appendix_lemmas}

\begin{lemma}\label{lemma:bias_term}
Under Assumptions~\ref{as:5} and~\ref{as:6}, the bias correction term $\mathrm{bias}_i(\boldsymbol{\theta})$ satisfies:
\begin{enumerate}
\item At $\boldsymbol{\theta}_0$: $\mathrm{bias}_i(\boldsymbol{\theta}_0) = \bigl\{\mathbf{W} (\mathbf{I}_n - \vartheta_0 \mathbf{W})^{-1}\bigr\}_{ii} \sigma_0 \kappa$.
\item $\mathbb{E}\bigl\{\psi_2(\varepsilon_i/\sigma_0) (\mathbf{W}\mathbf{Y})_i\bigr\} = \mathrm{bias}_i(\boldsymbol{\theta}_0)$.
\item $\displaystyle\frac{\partial \mathrm{bias}_i(\boldsymbol{\theta})}{\partial \vartheta} = \bigl\{\mathbf{W} (\mathbf{I}_n - \vartheta \mathbf{W})^{-1} \mathbf{W} (\mathbf{I}_n - \vartheta \mathbf{W})^{-1}\bigr\}_{ii}   \sigma \kappa$.
\end{enumerate}
\end{lemma}

\begin{proof}[Proof of Lemma~\ref{lemma:bias_term}]
We prove each part.

1. Direct substitution: 
\begin{equation*}
\mathrm{bias}_i(\boldsymbol{\theta}_0) = \bigl\{\mathbf{W} (\mathbf{I}_n - \vartheta_0 \mathbf{W})^{-1}\bigr\}_{ii} \sigma_0 \kappa.
\end{equation*}

2. Compute the expectation:
\begin{align*}
\mathbb{E}\bigl\{\psi_2(\varepsilon_i/\sigma_0)(\mathbf{W}\mathbf{Y})_i\bigr\} &= \mathbb{E}\Bigl\{ \psi_2(\varepsilon_i/\sigma_0) \sum_{j=1}^n w_{ij} Y_j \Bigr\} \\
&= \sum_{j=1}^n w_{ij}  \mathbb{E}\bigl\{\psi_2(\varepsilon_i/\sigma_0) Y_j\bigr\}.
\end{align*}
From the model: $\mathbf{Y} = (\mathbf{I}_n - \vartheta_0 \mathbf{W})^{-1}(\mathbf{Z}\boldsymbol{\gamma}_0 + \boldsymbol{\varepsilon})$. Hence $Y_j = \mathbf{s}_j^\top (\mathbf{Z}\boldsymbol{\gamma}_0 + \boldsymbol{\varepsilon})$ where $\mathbf{s}_j^\top$ denotes the $j$\textsuperscript{th} row of $\mathbf{S} := (\mathbf{I}_n - \vartheta_0 \mathbf{W})^{-1}$. Thus
\begin{align*}
\mathbb{E}\bigl\{\psi_2(\varepsilon_i/\sigma_0) Y_j\bigr\} &= \mathbb{E}\bigl\{\psi_2(\varepsilon_i/\sigma_0) \mathbf{s}_j^\top \mathbf{Z}\boldsymbol{\gamma}_0\bigr\} + \mathbb{E}\bigl\{\psi_2(\varepsilon_i/\sigma_0) \mathbf{s}_j^\top \boldsymbol{\varepsilon}\bigr\} \\
&= \mathbf{s}_j^\top \mathbf{Z}\boldsymbol{\gamma}_0  \mathbb{E}\{\psi_2(\varepsilon_i/\sigma_0)\} + \sum_{k=1}^n s_{jk}  \mathbb{E}\{\psi_2(\varepsilon_i/\sigma_0) \varepsilon_k\}.
\end{align*}
By Assumption~\ref{as:6}~(1), $\mathbb{E}\{\psi_2(\varepsilon_i/\sigma_0)\} = 0$. For $k \neq i$, independence of the errors (Assumption~\ref{as:5}) gives $\mathbb{E}\{\psi_2(\varepsilon_i/\sigma_0) \varepsilon_k\} = 0$, while for $k=i$, Assumption~\ref{as:6}~(4) yields $\mathbb{E}\{\psi_2(\varepsilon_i/\sigma_0) \varepsilon_i\} = \sigma_0 \kappa$. Consequently, $\mathbb{E}\{\psi_2(\varepsilon_i/\sigma_0) Y_j\} = s_{ji} \sigma_0 \kappa$.
Therefore
\begin{equation*}
\mathbb{E}\bigl\{\psi_2(\varepsilon_i/\sigma_0)(\mathbf{W}\mathbf{Y})_i\bigr\} = \sum_{j=1}^n w_{ij} s_{ji} \sigma_0 \kappa = \bigl(\mathbf{w}_i^\top \mathbf{S}\bigr)_i \sigma_0 \kappa
= \mathrm{bias}_i(\boldsymbol{\theta}_0),
\end{equation*}
where $\mathbf{S} = (\mathbf{I}_n - \vartheta_0 \mathbf{W})^{-1}$.

3. Differentiate $\mathrm{bias}_i(\boldsymbol{\theta})$:
\begin{align*}
\frac{\partial \mathrm{bias}_i(\boldsymbol{\theta})}{\partial \vartheta} &= \frac{\partial}{\partial \vartheta} \Bigl(
\bigl[\mathbf{W}(\mathbf{I}_n-\vartheta\mathbf{W})^{-1}\bigr]_{ii}\sigma\kappa
\Bigr) \\
&= \mathbf{e}_i^\top \mathbf{W} \frac{\partial}{\partial\vartheta} (\mathbf{I}_n-\vartheta\mathbf{W})^{-1}
\mathbf{e}_i \sigma\kappa.
\end{align*}
Using
\begin{equation*}
\frac{\partial}{\partial\vartheta} (\mathbf{I}_n-\vartheta\mathbf{W})^{-1} = (\mathbf{I}_n-\vartheta\mathbf{W})^{-1} \mathbf{W} (\mathbf{I}_n-\vartheta\mathbf{W})^{-1},
\end{equation*}
we obtain
\begin{equation*}
\frac{\partial \mathrm{bias}_i(\boldsymbol{\theta})}{\partial \vartheta} = \bigl[ \mathbf{W} (\mathbf{I}_n-\vartheta\mathbf{W})^{-1} \mathbf{W} (\mathbf{I}_n-\vartheta\mathbf{W})^{-1} \bigr]_{ii} \sigma\kappa.
\end{equation*}
This completes the proof.
\end{proof}

\begin{lemma}\label{lem:rfpca_fisher}
Assume that $\mathcal{X}$ is an elliptical random element in $\mathcal{H}$ with parameters $(\mu,\Gamma_0)$ in the sense of \citet[][Section~5]{Bali2011}, i.e. $\mathcal{X}\sim E(\mu,\Gamma_0)$, where $\Gamma_0$ is self-adjoint, positive semidefinite, and compact. Assume that $\mathbb{E}\|\mathcal{X}\|_{\mathcal{H}}^2<\infty$ so that the covariance operator $\Gamma$ exists. Let $\sigma_M$ be the M-scale functional used in the projection-pursuit index, calibrated so that it is Fisher-consistent for the standard deviation at the normal law (Assumption~\ref{as:3}~(ii)).
If the eigenvalues of $\Gamma_0$ satisfy $\lambda_{0,1}>\lambda_{0,2}>\cdots>\lambda_{0,K}>\lambda_{0,K+1}$, then the population robust projection-pursuit PCA functionals satisfy
\begin{equation*}
\phi_{\mathrm{R},k}(\mathbb{P}_{\mathcal{X}})=\phi_{0,k},\qquad
\lambda_{\mathrm{R},k}(\mathbb{P}_{\mathcal{X}})=c \lambda_{0,k},
\end{equation*}
where $(\lambda_{0,k},\phi_{0,k})$ are the eigenpairs of $\Gamma_0$ and $c>0$ is the proportionality constant in~\eqref{S3} below. In particular, if $\Gamma_0=\Gamma$ and the calibration yields $c=1$, then $\phi_{\mathrm{R},k}(\mathbb{P}_{\mathcal{X}})=\phi_k$ and $\lambda_{\mathrm{R},k}(\mathbb{P}_{\mathcal{X}})=\lambda_k$.
Moreover, for any Fisher-consistent location functional $\mu_R(\mathbb{P}_{\mathcal{X}})=\mu$, the corresponding population scores are $\xi_k(\mathbb{P}_{\mathcal{X}})=\langle \mathcal{X}-\mu,\phi_{\mathrm{R},k}(\mathbb{P}_{\mathcal{X}})\rangle$.
\end{lemma}

\begin{proof}[Proof of Lemma~\ref{lem:rfpca_fisher}]
Since $\mathcal{X}\sim E(\mu,\Gamma_0)$, \citet[][Section~5]{Bali2011} show that the projection-pursuit scale index satisfies the following assumption:
\begin{equation}\label{S3}
\exists c>0\ \text{and compact }\Gamma_0 \text{ such that }\ \sigma_M^2(\mathbb{P}_{\mathcal{X}}[a])=c \langle a,\Gamma_0 a\rangle
\quad \forall a \in\mathcal{H}, 
\end{equation}
where $\mathbb{P}_{\mathcal{X}}[a]$ denotes the law of $\langle a,\mathcal{X}\rangle$. This is precisely the condition stated in \cite{Bali2011}. Under~\eqref{S3} and the strict separation of eigenvalues of $\Gamma_0$, \citet[][Lemma~5.1]{Bali2011} yields
\begin{equation*}
\phi_{\mathrm{R},k}(\mathbb{P}_{\mathcal{X}})=\phi_{0,k},\qquad
\lambda_{\mathrm{R},k}(\mathbb{P}_{\mathcal{X}})=c \lambda_{0,k},
\end{equation*}
for each $k \in \{1,\ldots,K\}$, where $(\lambda_{0,k},\phi_{0,k})$ are the eigenpairs of $\Gamma_0$. This proves the first claim.

If, in addition, $\mathbb{E}\|\mathcal{X}\|_{\mathcal{H}}^2<\infty$, then the covariance operator $\Gamma$ exists. For elliptical elements, when the covariance exists, it is proportional to the shape operator $\Gamma_0$, i.e. $\Gamma=\alpha \Gamma_0$ for some $\alpha>0$ \citep[][Lemma 2.2~(b)]{Bali2009}. Hence, when $\Gamma_0$ is chosen as the covariance operator itself (or, equivalently, absorbed into the calibration), and $\sigma_M$ is calibrated to be Fisher-consistent at the normal law, we may take $c=1$ and $\Gamma_0=\Gamma$, which gives $\phi_{\mathrm{R},k}=\phi_k$ and $\lambda_{\mathrm{R},k}=\lambda_k$.

Finally, by definition of the projection scores in the population functional, for any Fisher-consistent location functional $\mu_R(\mathbb{P}_{\mathcal{X}})=\mu$, we have $\xi_k(\mathbb{P}_{\mathcal{X}})=\langle \mathcal{X}-\mu,\phi_{\mathrm{R},k}(\mathbb{P}_{\mathcal{X}})\rangle$. This completes the proof.
\end{proof}

\begin{lemma}\label{lem:fisher_consistency}
Under Assumptions~\ref{as:5} and~\ref{as:6}, and with the bias correction term $\mathrm{bias}_i(\boldsymbol{\theta})$ defined in Lemma~\ref{lemma:bias_term}, the estimating functions in
\eqref{eq:est_eq_gamma}-\eqref{eq:est_eq_rho} satisfy $\mathbb{E}\bigl\{\boldsymbol{\psi}_{2,n,i}(\boldsymbol{\theta}_0)\bigr\}=\boldsymbol{0}$, where $\boldsymbol{\theta}_0=(\boldsymbol{\gamma}_0^\top,\sigma_0,\vartheta_0)^\top$ is the true parameter. Consequently, $\boldsymbol{\theta}_0$ is a root of the population estimating equations $\mathbb{E}\{\boldsymbol{\psi}_{2,n,i}(\boldsymbol{\theta})\}=\boldsymbol{0}$. If this root is unique (or the solution is defined as the root in a neighborhood of $\boldsymbol{\theta}_0$), then the corresponding M-functional is Fisher consistent, and any consistent root $\widehat{\boldsymbol{\theta}}_n$ of $\sum_{i=1}^n \boldsymbol{\psi}_{2,n,i}(\boldsymbol{\theta})=\boldsymbol{0}$ inherits Fisher consistency.
\end{lemma}

\begin{proof}[Proof of Lemma~\ref{lem:fisher_consistency}]
We verify the unbiasedness of each component of $\boldsymbol{\psi}_{2,n,i}(\boldsymbol{\theta}_0)$. At the true parameter, the model implies
\begin{equation*}
(\mathbf{I}_n-\vartheta_0\mathbf{W})\mathbf{Y}-\mathbf{Z}\boldsymbol{\gamma}_0=\boldsymbol{\varepsilon},
\end{equation*}
hence $e_i(\boldsymbol{\theta}_0)=\varepsilon_i$.

\textit{(i)} At $\boldsymbol{\theta}_0$, the contribution is $\psi_2(\varepsilon_i/\sigma_0)\mathbf{z}_i$. Using Assumption~\ref{as:5} (independence of $\varepsilon_i$ and $\mathbf{z}_i$) and Assumption~\ref{as:6}~(1),
\begin{equation*}
\mathbb{E} \left\lbrace\psi_2 \left(\frac{\varepsilon_i}{\sigma_0}\right)\mathbf{z}_i\right\rbrace =\mathbb{E} \left[\mathbf{z}_i \mathbb{E} \left\{\psi_2 \left(\frac{\varepsilon_i}{\sigma_0}\right)\Bigm|\mathbf{z}_i\right\}\right] =\mathbb{E}(\mathbf{z}_i) \cdot 0=\mathbf{0}_{K+1}.
\end{equation*}

\textit{(ii)} At $\boldsymbol{\theta}_0$, the contribution is
$\psi_2(\varepsilon_i/\sigma_0)(\varepsilon_i/\sigma_0)-\kappa$.
By Assumption~\ref{as:6}~(4), $\kappa=\mathbb{E} \left\{\psi_2(\varepsilon_i/\sigma_0)(\varepsilon_i/\sigma_0)\right\}$,
so
\begin{equation*}
\mathbb{E} \left\{\psi_2 \left(\frac{\varepsilon_i}{\sigma_0}\right)\frac{\varepsilon_i}{\sigma_0}-\kappa\right\}=0.
\end{equation*}

\textit{(iii)} At $\boldsymbol{\theta}_0$, the contribution is
$\psi_2(\varepsilon_i/\sigma_0)(\mathbf{W}\mathbf{Y})_i-\mathrm{bias}_i(\boldsymbol{\theta}_0)$. We show $\mathbb{E}\{\psi_2(\varepsilon_i/\sigma_0)(\mathbf{W}\mathbf{Y})_i\}=\mathrm{bias}_i(\boldsymbol{\theta}_0)$. From \eqref{eq:proj_SAR},
\begin{equation*}
\mathbf{Y}=(\mathbf{I}_n-\vartheta_0\mathbf{W})^{-1}(\mathbf{Z}\boldsymbol{\gamma}_0+\boldsymbol{\varepsilon})
=:\mathbf{S}_0(\mathbf{Z}\boldsymbol{\gamma}_0+\boldsymbol{\varepsilon}),
\end{equation*}
and writing $Y_j=\mathbf{s}_j^\top\mathbf{Z}\boldsymbol{\gamma}_0+\sum_{k=1}^n s_{jk}\varepsilon_k$ gives
\begin{equation*}
(\mathbf{W}\mathbf{Y})_i=\sum_{j=1}^n w_{ij}\mathbf{s}_j^\top\mathbf{Z}\boldsymbol{\gamma}_0
+\sum_{j=1}^n w_{ij}\sum_{k=1}^n s_{jk}\varepsilon_k.
\end{equation*}
Therefore,
\begin{equation*}
\mathbb{E} \left\{\psi_2 \left(\frac{\varepsilon_i}{\sigma_0}\right)(\mathbf{W}\mathbf{Y})_i\right\} = \underbrace{\mathbb{E} \left\{\psi_2 \left(\frac{\varepsilon_i}{\sigma_0}\right)\right\}}_{=0} \mathbb{E} \left\{\sum_{j=1}^n w_{ij}\mathbf{s}_j^\top\mathbf{Z}\boldsymbol{\gamma}_0\right\}
+\sum_{j=1}^n w_{ij}\sum_{k=1}^n s_{jk} \mathbb{E} \left\{\psi_2 \left(\frac{\varepsilon_i}{\sigma_0}\right)\varepsilon_k\right\}.
\end{equation*}
By independence, $\mathbb{E}\{\psi_2(\varepsilon_i/\sigma_0)\varepsilon_k\}=0$ for $k\neq i$, and by Assumption~\ref{as:6}~(4), $\mathbb{E}\{\psi_2(\varepsilon_i/ \sigma_0) \varepsilon_i\} = \sigma_0\kappa$. Hence,
\begin{equation*}
\mathbb{E} \left\{\psi_2 \left(\frac{\varepsilon_i}{\sigma_0}\right)(\mathbf{W}\mathbf{Y})_i\right\} =\sum_{j=1}^n w_{ij}s_{ji} \sigma_0\kappa =\bigl(\mathbf{w}_i^\top\mathbf{S}_0\bigr)_i \sigma_0\kappa =\mathrm{bias}_i(\boldsymbol{\theta}_0),
\end{equation*}
where the last equality uses Lemma~\ref{lemma:bias_term}. Therefore
\begin{equation*}
\mathbb{E} \left\{\psi_2 \left(\frac{\varepsilon_i}{\sigma_0}\right)(\mathbf{W}\mathbf{Y})_i-\mathrm{bias}_i(\boldsymbol{\theta}_0)\right\}=0.
\end{equation*}

Combining (i)-(iii) yields $\mathbb{E}\{\boldsymbol{\psi}_{2,n,i}(\boldsymbol{\theta}_0)\}=\boldsymbol{0}$.
\end{proof}

\section{Proof of main Theorems}\label{sec:proofs}

\begin{proof}[Proof of Theorem~\ref{thm:FCSAR_fisher}]
The proof combines the Fisher consistency of the RFPCA functional (Lemma~\ref{lem:rfpca_fisher}) with that of the bias-corrected M-estimator for the finite-dimensional SAR model (Lemma~\ref{lem:fisher_consistency}).

By Lemma~\ref{lem:rfpca_fisher}, the population RFPCA functional satisfies \(\phi_{\mathrm{R},k}(\mathbb{P}_{\mathcal{X}}) = \phi_k\) and \(\xi_k(\mathbb{P}_{\mathcal{X}}) = \langle \mathcal{X}-\mu,\phi_k\rangle\). Hence, at the true distribution, the design matrix \(\mathbf{Z}(\mathbb{P}_{\mathcal{X} Y}) = [\mathbf{1}_n,\boldsymbol{\Xi}]\) is exactly the matrix of true scores, and the finite-dimensional approximation \eqref{eq:proj_SAR} holds with the true parameters \(\boldsymbol{\gamma}_0 = (\beta_0,\mathbf{b}^\top)^\top\) (Assumption~\ref{as:4} guarantees that \(\beta(t) = \sum_{k=1}^K b_k\phi_k(t)\)).

Lemma~\ref{lem:fisher_consistency} proves that the estimating functions $\boldsymbol{\psi}_{2,n,i}(\boldsymbol{\theta})$ are unbiased at $\boldsymbol{\theta}_0$, so that the corresponding population estimating equation has root $\boldsymbol{\theta}_0$. By assumption~(iii), this root is unique. Therefore, the M-functional defined by the population bias-corrected estimating equation is Fisher consistent at $\boldsymbol{\theta}_0 = (\boldsymbol{\gamma}_0^\top,\sigma_0,\vartheta_0)^\top$. Consequently,
\begin{equation*}
\widehat{\boldsymbol{\gamma}}(\mathbb{P}_{\mathcal{X} Y}) = \boldsymbol{\gamma}_0, \qquad
\widehat{\sigma}(\mathbb{P}_{\mathcal{X} Y}) = \sigma_0, \qquad
\widehat{\vartheta}(\mathbb{P}_{\mathcal{X} Y}) = \vartheta_0.
\end{equation*}

The estimated coefficient function is \(\widehat{\beta}(t) = \sum_{k=1}^K \widehat{b}_k \widehat{\phi}_k(t)\). By Lemma~\ref{lem:rfpca_fisher}, \(\widehat{\phi}_k(\mathbb{P}_{\mathcal{X} Y}) = \phi_k\); by Lemma~\ref{lem:fisher_consistency}, \(\widehat{b}_k(\mathbb{P}_{\mathcal{X} Y}) = b_k\). Thus, \(\widehat{\beta}(\mathbb{P}_{\mathcal{X} Y})(t) = \sum_{k=1}^K b_k\phi_k(t) = \beta(t)\) for all \(t\). All four components of the FCSAR estimator are therefore Fisher consistent.
\end{proof}

\begin{proof}[Proof of Theorem~\ref{thm:consistency_main}]
Let $\boldsymbol{\Psi}(\boldsymbol{\theta}) := \lim_{n\to\infty}\mathbb{E}\{\boldsymbol{\Psi}_n(\boldsymbol{\theta})\}$ and $\boldsymbol{\Psi}_n(\boldsymbol{\theta}) := \frac{1}{n}\sum_{i=1}^n \boldsymbol{\psi}_{2,n,i}(\boldsymbol{\theta})$. By Lemma~\ref{lem:fisher_consistency}, $\boldsymbol{\Psi}(\boldsymbol{\theta}_0)=\mathbf{0}$.
By Assumption~\ref{as:10}, $\boldsymbol{\theta}_0$ is the unique zero of $\boldsymbol{\Psi}$ on $\Theta$. Further, by Assumption~\ref{as:13}, $\sup_{\boldsymbol{\theta}\in\Theta} \|\boldsymbol{\Psi}_n(\boldsymbol{\theta})-\boldsymbol{\Psi}(\boldsymbol{\theta})\|
\xrightarrow{p}0$.

Fix $\epsilon>0$. By Assumption~\ref{as:10} and compactness of $\Theta$, $\delta(\epsilon) := \inf_{\|\boldsymbol{\theta}-\boldsymbol{\theta}_0\|\ge \epsilon} \|\boldsymbol{\Psi}(\boldsymbol{\theta})\| >0$. Consider the event
\begin{equation*}
E_n(\epsilon) := \left\{ \sup_{\boldsymbol{\theta}\in\Theta}
\|\boldsymbol{\Psi}_n(\boldsymbol{\theta})-\boldsymbol{\Psi}(\boldsymbol{\theta})\| < \frac{\delta(\epsilon)}{2} \right\}.
\end{equation*}
On $E_n(\epsilon)$, for every $\boldsymbol{\theta}$ satisfying
$\|\boldsymbol{\theta}-\boldsymbol{\theta}_0\|\ge \epsilon$,
\begin{equation*}
\|\boldsymbol{\Psi}_n(\boldsymbol{\theta})\| \ge \|\boldsymbol{\Psi}(\boldsymbol{\theta})\| - \sup_{\boldsymbol{\vartheta}\in\Theta} \|\boldsymbol{\Psi}_n(\boldsymbol{\vartheta})-\boldsymbol{\Psi} (\boldsymbol{\vartheta})\| \ge \delta(\epsilon)-\frac{\delta(\epsilon)}{2} = \frac{\delta(\epsilon)}{2}.
\end{equation*}

Now let $\widehat{\boldsymbol{\theta}}_n$ satisfy
$\boldsymbol{\Psi}_n(\widehat{\boldsymbol{\theta}}_n)=\mathbf{0}$, or more generally $\|\boldsymbol{\Psi}_n(\widehat{\boldsymbol{\theta}}_n)\|=o_p(1)$. Then, with probability tending to one, $\|\boldsymbol{\Psi}_n(\widehat{\boldsymbol{\theta}}_n)\|
< \frac{\delta(\epsilon)}{2}$. Hence, on the event $E_n(\epsilon)$ and for all sufficiently large $n$,
$\widehat{\boldsymbol{\theta}}_n$ cannot satisfy
$\|\widehat{\boldsymbol{\theta}}_n-\boldsymbol{\theta}_0\|\ge \epsilon$. Therefore $\mathbb{P} \left( \|\widehat{\boldsymbol{\theta}}_n-\boldsymbol{\theta}_0\|\ge \epsilon \right)\to 0$. Since $\epsilon>0$ is arbitrary, it follows that $\widehat{\boldsymbol{\theta}}_n \xrightarrow{p} \boldsymbol{\theta}_0$. 
\end{proof}

\begin{proof}[Proof of Theorem~\ref{thm:asymptotic_normality_main}]
We first derive the explicit form of the limiting Jacobian matrix
$\boldsymbol{\mathcal A}(\boldsymbol{\theta}_0)$. Write $u_i:=\varepsilon_i/\sigma_0$. At $\boldsymbol{\theta}_0$,
$e_i(\boldsymbol{\theta}_0)=\varepsilon_i$, $\partial e_i/\partial\boldsymbol{\gamma}=-\mathbf{z}_i$, and $\partial e_i/\partial\vartheta=-(\mathbf{W}\mathbf{Y})_i$. Also, $(\mathbf{W}\mathbf{Y})_i$ is data and therefore has zero derivative with respect to the parameters. Let $d_i(\vartheta) := \bigl[\mathbf{W}(\mathbf{I}_n-\vartheta\mathbf{W})^{-1}\bigr]_{ii}$, so that $\mathrm{bias}_i(\boldsymbol{\theta}) = d_i(\vartheta)\sigma\kappa$. Then, 
\begin{equation*}
\frac{\partial}{\partial\sigma}\mathrm{bias}_i(\boldsymbol{\theta}) = d_i(\vartheta)\kappa = \frac{\mathrm{bias}_i(\boldsymbol{\theta})}{\sigma},
\qquad
\frac{\partial}{\partial\boldsymbol{\gamma}}\mathrm{bias}_i(\boldsymbol{\theta}) = \mathbf{0}.
\end{equation*}

Define
\begin{equation*}
\boldsymbol{\psi}_{2,n,i}(\boldsymbol{\theta})=
\begin{bmatrix}
\psi_2(e_i/\sigma)\mathbf{z}_i\\
\psi_2(e_i/\sigma)(e_i/\sigma)-\kappa\\
\psi_2(e_i/\sigma)(\mathbf{W}\mathbf{Y})_i-\mathrm{bias}_i(\boldsymbol{\theta})
\end{bmatrix}.
\end{equation*}
Then $\boldsymbol{\mathcal A} = \boldsymbol{\mathcal A}(\boldsymbol{\theta}_0) = \lim_{n\to\infty} \mathbb E\left\{
\frac{\partial}{\partial\boldsymbol{\theta}^\top} \boldsymbol{\Psi}_n(\boldsymbol{\theta}_0) \right\}$ is the $(K+3)\times(K+3)$ matrix 
\begin{equation*}
\boldsymbol{\mathcal A} =
\begin{bmatrix}
A_{\gamma\gamma} & A_{\gamma\sigma} & A_{\gamma\vartheta}\\
A_{\sigma\gamma} & A_{\sigma\sigma} & A_{\sigma\vartheta}\\
A_{\vartheta\gamma} & A_{\vartheta\sigma} & A_{\vartheta\vartheta}
\end{bmatrix},
\end{equation*}
with blocks
\begin{align*}
A_{\gamma\gamma} &= -\frac{1}{\sigma_0} \mathbb{E}\left\{
\psi_2'(u_i)\mathbf{z}_i\mathbf{z}_i^\top \right\}, \\
A_{\gamma\sigma} &= -\frac{1}{\sigma_0^2} \mathbb{E}\left\{
\psi_2'(u_i)\varepsilon_i\mathbf{z}_i \right\}, \\
A_{\gamma\vartheta} &= -\frac{1}{\sigma_0} \mathbb{E}\left\{
\psi_2'(u_i)\mathbf{z}_i(\mathbf{W}\mathbf{Y})_i \right\}, \\
A_{\sigma\gamma} &= -\frac{1}{\sigma_0} \mathbb{E}\left\{
\bigl[\psi_2'(u_i)u_i+\psi_2(u_i)\bigr]\mathbf{z}_i^\top
\right\}, \\
A_{\sigma\sigma} &= -\frac{1}{\sigma_0} \mathbb{E}\left\{
\psi_2'(u_i)u_i^2+\psi_2(u_i)u_i \right\}, \\
A_{\sigma\vartheta} &= -\frac{1}{\sigma_0} \mathbb{E}\left\{
\bigl[\psi_2'(u_i)u_i+\psi_2(u_i)\bigr](\mathbf{W}\mathbf{Y})_i
\right\}, \\
A_{\vartheta\gamma} &= -\frac{1}{\sigma_0} \mathbb{E}\left\{
\psi_2'(u_i)(\mathbf{W}\mathbf{Y})_i\mathbf{z}_i^\top \right\}, \\
A_{\vartheta\sigma} &= -\frac{1}{\sigma_0} \mathbb{E}\left\{
\psi_2'(u_i)u_i(\mathbf{W}\mathbf{Y})_i \right\} - \mathbb{E}\left\{ d_i(\vartheta_0)\kappa \right\}, \\
A_{\vartheta\vartheta} &= -\frac{1}{\sigma_0} \mathbb{E}\left\{
\psi_2'(u_i)(\mathbf{W}\mathbf{Y})_i^2 \right\} - \mathbb{E}\left\{ \frac{\partial}{\partial\vartheta}\mathrm{bias}_i(\boldsymbol{\theta}_0) \right\}, 
\end{align*}
where $\frac{\partial}{\partial\vartheta}\mathrm{bias}_i(\boldsymbol{\theta}_0) = \bigl\{\mathbf{W} (\mathbf{I}_n-\vartheta_0\mathbf{W})^{-1} \mathbf{W} (\mathbf{I}_n-\vartheta_0\mathbf{W})^{-1} \bigr\}_{ii} \sigma_0\kappa$.

Now write $\boldsymbol{\Psi}_n(\boldsymbol{\theta}) = \frac{1}{n}\sum_{i=1}^n \boldsymbol{\psi}_{2,n,i}(\boldsymbol{\theta})$ and $\boldsymbol{\Psi}(\boldsymbol{\theta}) = \lim_{n\to\infty}\mathbb{E}\{\boldsymbol{\Psi}_n(\boldsymbol{\theta})\}$. By Lemma~\ref{lem:fisher_consistency}, $\boldsymbol{\Psi}(\boldsymbol{\theta}_0)=\mathbf{0}$. Assume first that
$\boldsymbol{\Psi}_n(\widehat{\boldsymbol{\theta}}_n)=\mathbf{0}$. A mean-value expansion around $\boldsymbol{\theta}_0$ gives
\begin{equation*}
\mathbf{0} = \boldsymbol{\Psi}_n(\boldsymbol{\theta}_0)
+ \frac{\partial}{\partial\boldsymbol{\theta}^\top}
\boldsymbol{\Psi}_n(\widetilde{\boldsymbol{\theta}}_n)
(\widehat{\boldsymbol{\theta}}_n-\boldsymbol{\theta}_0),
\end{equation*}
for some intermediate value $\widetilde{\boldsymbol{\theta}}_n$
on the line segment joining $\widehat{\boldsymbol{\theta}}_n$
and $\boldsymbol{\theta}_0$. Hence
\begin{equation*}
\sqrt{n}(\widehat{\boldsymbol{\theta}}_n-\boldsymbol{\theta}_0)
= - \left\{ \frac{\partial}{\partial\boldsymbol{\theta}^\top}
\boldsymbol{\Psi}_n(\widetilde{\boldsymbol{\theta}}_n)
\right\}^{-1} \frac{1}{\sqrt{n}} \sum_{i=1}^n
\boldsymbol{\psi}_{2,n,i}(\boldsymbol{\theta}_0).
\end{equation*}

By Theorem~\ref{thm:consistency_main}, $\widehat{\boldsymbol{\theta}}_n\xrightarrow{p}\boldsymbol{\theta}_0$, and therefore $\widetilde{\boldsymbol{\theta}}_n\xrightarrow{p}\boldsymbol{\theta}_0$. By Assumption~\ref{as:14}, $\sup_{\boldsymbol{\theta}\in\mathcal N(\boldsymbol{\theta}_0)}
\left\| \frac{\partial}{\partial\boldsymbol{\theta}^\top}
\boldsymbol{\Psi}_n(\boldsymbol{\theta}) - \boldsymbol{\mathcal A}(\boldsymbol{\theta}_0) \right\| \xrightarrow{p}0$ so that $\frac{\partial}{\partial\boldsymbol{\theta}^\top} \boldsymbol{\Psi}_n(\widetilde{\boldsymbol{\theta}}_n)
\xrightarrow{p} \boldsymbol{\mathcal A}(\boldsymbol{\theta}_0)$. Since $\boldsymbol{\mathcal A}(\boldsymbol{\theta}_0)$ is nonsingular by Assumption~\ref{as:14},
\begin{equation*}
\left\{ \frac{\partial}{\partial\boldsymbol{\theta}^\top}
\boldsymbol{\Psi}_n(\widetilde{\boldsymbol{\theta}}_n)
\right\}^{-1} \xrightarrow{p} \boldsymbol{\mathcal A}(\boldsymbol{\theta}_0)^{-1}.
\end{equation*}

Further, by Assumption~\ref{as:14}, $\frac{1}{\sqrt{n}} \sum_{i=1}^n \boldsymbol{\psi}_{2,n,i}(\boldsymbol{\theta}_0) \Rightarrow N(\mathbf 0,\boldsymbol{\mathcal B})$. Therefore, by Slutsky's theorem,
\begin{equation*}
\sqrt{n}(\widehat{\boldsymbol{\theta}}_n-\boldsymbol{\theta}_0)
\Rightarrow N \left( \mathbf 0, \boldsymbol{\mathcal A}^{-1}
\boldsymbol{\mathcal B} \boldsymbol{\mathcal A}^{-\top} \right).
\end{equation*}

If one allows approximate roots, the same conclusion holds provided $\|\boldsymbol{\Psi}_n(\widehat{\boldsymbol{\theta}}_n)\|=o_p(n^{-1/2})$ so that the expansion contributes only an asymptotically negligible remainder. This completes the proof.
\end{proof}

\begin{proof}[Proof of Theorem~\ref{thm:asymp_beta_main}]
We decompose $\sqrt{n}(\widehat{\beta}-\beta)$ into manageable parts. Write
\begin{equation*}
\widehat{\beta}(t)-\beta(t)=\sum_{k=1}^K (\widehat{b}_k-b_k)\phi_k(t)+\sum_{k=1}^K b_k\bigl\{\widehat{\phi}_k(t)-\phi_k(t)\bigr\}
+\sum_{k=1}^K (\widehat{b}_k-b_k)\bigl\{\widehat{\phi}_k(t)-\phi_k(t)\bigr\}.
\end{equation*}
From Theorem~\ref{thm:asymptotic_normality_main}, $\sqrt{n}(\widehat{\mathbf{b}}-\mathbf{b}) = \mathcal{O}_p(1)$. Under distinct eigenvalues, the RFPCA estimators satisfy $\|\widehat{\phi}_k-\phi_k\| = \mathcal{O}_p(n^{-1/2})$, e.g., Theorem~4.1 in \cite{Bali2011}. Hence the product term is $\mathcal{O}_p(n^{-1})$ and asymptotically negligible. Thus,
\begin{equation}\label{tag1}
\sqrt{n}\bigl(\widehat{\beta}-\beta\bigr)=\sqrt{n}\sum_{k=1}^K (\widehat{b}_k-b_k)\phi_k +\sqrt{n}\sum_{k=1}^K b_k\bigl(\widehat{\phi}_k-\phi_k\bigr)+o_p(1) \quad\text{in }\mathcal{L}^2[0,1].
\end{equation}

From Theorem~\ref{thm:asymptotic_normality_main}, the subvector $\widehat{\mathbf{b}}$ satisfies $\sqrt{n}\bigl(\widehat{\mathbf{b}}-\mathbf{b}\bigr)\xrightarrow{d}N(\mathbf{0},\boldsymbol{V}_{\mathbf{bb}})$. For the eigenfunctions, Theorem~3.1 in \cite{Bali2015} gives the influence function $\operatorname{IF}(\cdot;\phi_k,\mathbb{P}_{\mathcal{X}})$ and the Bahadur expansion
\begin{equation}\label{tag2}
\sqrt{n}\bigl(\widehat{\phi}_k-\phi_k\bigr)=\frac{1}{\sqrt{n}}\sum_{i=1}^n \operatorname{IF}(\mathcal{X}_i;\phi_k,\mathbb{P}_{\mathcal{X}})+o_p(1)
\quad\text{in }\mathcal{L}^2[0,1]
\end{equation}
Consequently, the $K$-tuple $(\widehat{\phi}_1-\phi_1,\ldots,\widehat{\phi}_K-\phi_K)$ converges jointly to a mean-zero Gaussian element $(\Delta_1,\ldots,\Delta_K)$ in $\mathcal{L}^2[0,1]^K$ with covariance
\begin{equation}\label{tag3}
\mathbb{E}\bigl\{\langle\Delta_k,f\rangle\langle\Delta_j,g\rangle\bigr\}
=\mathbb{E}\bigl\{\langle\operatorname{IF}(\mathcal{X};\phi_k),f\rangle\langle\operatorname{IF}(\mathcal{X};\phi_j),g\rangle\bigr\}, \qquad \forall f,g\in\mathcal{L}^2[0,1].
\end{equation}

Define $\Phi:\mathbb{R}^K\times\mathcal{L}^2[0,1]^K\to\mathcal{L}^2[0,1]$ by $\Phi(\mathbf{b},\phi_1,\ldots,\phi_K)=\sum_{k=1}^K b_k\phi_k$. This map is linear, hence Hadamard differentiable with derivative
\begin{equation*}
\Phi'_{(\mathbf{b},\boldsymbol{\phi})}(\mathbf{h}_b,\mathbf{h}_{\phi_1},\ldots,\mathbf{h}_{\phi_K}) = \sum_{k=1}^K h_{b,k}\phi_k + \sum_{k=1}^K b_k h_{\phi_k}.
\end{equation*}

To apply the functional delta method, we need the joint convergence of $(\widehat{\mathbf{b}},\widehat{\phi}_1,\ldots,\widehat{\phi}_K)$. The expansions~\eqref{tag1} and~\eqref{tag2} express the estimators as sample averages of influence functions. Under Assumptions~\ref{as:2}-\ref{as:7}, $\{(\mathcal{X}_i,\varepsilon_i)\}_{i\ge1}$ are i.i.d., hence the stacked influence functions form an i.i.d. sequence in the product Hilbert space. Stack these influence functions into a single Hilbert‑space‑valued random element. Assume, in addition, that
\begin{equation*}
\sqrt{n}\bigl(\widehat{\mathbf b}-\mathbf b,\widehat{\phi}_1-\phi_1,\ldots,\widehat{\phi}_K-\phi_K\bigr) \xrightarrow{d}
(\mathbf Z_b,\Delta_1,\ldots,\Delta_K)
\end{equation*}
in $\mathbb{R}^K\times\mathcal{L}^2[0,1]^K$, where $\mathbf Z_b\sim N(\mathbf 0,\boldsymbol{\mathcal V}_{\mathbf{bb}})$ and $(\Delta_1,\ldots,\Delta_K)$ is a mean-zero Gaussian element with covariance induced by the influence functions. The covariance between $\mathbf{Z}_b$ and the $\Delta_k$'s is given by the limit of
\begin{equation*}
\frac{1}{n}\sum_{i=1}^n \mathbb{E} \left\{ \boldsymbol{\mathcal{A}}_{\mathbf{bb}}^{-1}\boldsymbol{\psi}_{\mathbf{b},i}(\boldsymbol{\theta}_0) \operatorname{IF}(\mathcal{X}_i;\phi_k) \right\},
\end{equation*}
where $\boldsymbol{\psi}_{\mathbf{b},i}(\boldsymbol{\theta}_0)=\psi_2(\varepsilon_i/\sigma_0) \mathbf{z}_i$ (the component of the estimating function for $\mathbf{b}$). By Assumption~\ref{as:5}, $\varepsilon_i$ is independent of $\mathcal{X}_i$ (and thus of $\mathbf{z}_i$). Moreover, Assumption~\ref{as:6}~(1) gives $\mathbb{E}\{\psi_2(\varepsilon_i/\sigma_0)\}=0$. Hence, conditional on $\mathcal{X}_i$, the expectation of $\boldsymbol{\psi}_{\mathbf{b},i}$ is zero, implying
\begin{equation*}
\mathbb{E} \left\{ \boldsymbol{\mathcal{A}}_{\mathbf{bb}}^{-1}\boldsymbol{\psi}_{\mathbf{b},i}(\boldsymbol{\theta}_0) \operatorname{IF}(\mathcal{X}_i;\phi_k) \right\} = \boldsymbol{\mathcal{A}}_{\mathbf{bb}}^{-1} \mathbb{E} \left\{ \mathbb{E}[\boldsymbol{\psi}_{\mathbf{b},i}\mid\mathcal{X}_i] \operatorname{IF}(\mathcal{X}_i;\phi_k) \right\} = \mathbf{0}.
\end{equation*}
Thus the asymptotic cross‑covariance vanishes; $\mathbf{Z}_b$ and $(\Delta_1,\ldots,\Delta_K)$ are independent.

Now apply the functional delta method to the joint convergence:
\begin{equation*}
\sqrt{n}\bigl(\widehat{\beta}-\beta\bigr)
\xrightarrow{d} \Phi'_{(\mathbf{b},\boldsymbol{\phi})}(\mathbf{Z}_b,\Delta_1,\ldots,\Delta_K) = \sum_{k=1}^K Z_{b,k}\phi_k + \sum_{k=1}^K b_k\Delta_k \equiv \mathcal{Z}.
\end{equation*}

Finally, compute the covariance of $\langle\mathcal{Z},f\rangle$ and $\langle\mathcal{Z},g\rangle$. Using the independence between $\mathbf{Z}_b$ and the $\Delta_k$'s,
\begin{equation*}
\mathbb{E}\bigl\{\langle\mathcal{Z},f\rangle\langle\mathcal{Z},g\rangle\bigr\} = \mathbb{E} \left\{\Bigl(\sum_{k=1}^K Z_{b,k}\langle\phi_k,f\rangle\Bigr) \Bigl(\sum_{j=1}^K Z_{b,j}\langle\phi_j,g\rangle\Bigr)\right\} + \mathbb{E} \left\{\Bigl(\sum_{k=1}^K b_k\langle\Delta_k,f\rangle\Bigr)
\Bigl(\sum_{j=1}^K b_j\langle\Delta_j,g\rangle\Bigr)\right\}.
\end{equation*}
The first term gives $\sum_{k,j}(\boldsymbol{V}_{\mathbf{bb}})_{kj}\langle\phi_k,f\rangle\langle\phi_j,g\rangle$ because $\mathbb{E}(Z_{b,k}Z_{b,j})=(\boldsymbol{V}_{\mathbf{bb}})_{kj}$. The second term, using~\eqref{tag3}, equals $\sum_{k,j} b_k b_j \mathbb{E}\{\langle\operatorname{IF}(\mathcal{X};\phi_k),f\rangle\langle\operatorname{IF}(\mathcal{X};\phi_j),g\rangle\}$. This yields the covariance stated in the theorem.
\end{proof}

\newpage
\bibliographystyle{agsm}
\bibliography{Rob_Spatial_SoF_M}

\end{document}